\documentclass[aps,pra,twocolumn,10pt,superscriptaddress]{revtex4-2}
\pdfoutput=1
\usepackage[final]{graphicx}
\usepackage{times,bbm,amsmath,amssymb,amsthm}
\usepackage{epsfig,color}
\usepackage{xcolor}
\usepackage[colorlinks=true,citecolor=blue,linkcolor=blue]{hyperref}
\hypersetup{
   allcolors  = {blue},
}
\usepackage{cleveref}
\usepackage{subfiles}
\usepackage{float}
\usepackage{verbatim}
\usepackage[greek,english]{babel}
\usepackage{thumbpdf,enumerate}
\usepackage{booktabs}
\usepackage{sidecap}
\usepackage[scaled=.8]{couriers}    
\usepackage{pstricks}
\usepackage{multirow}
\usepackage{placeins}
\usepackage{relsize}  
\usepackage{pst-grad,bm}
\usepackage{epigraph}
\usepackage{booktabs}
\usepackage{bm}

\usepackage{soul}
\usepackage{ulem} 
\usepackage{acronym}
\usepackage{physics}
\usepackage{tikz}
\newsavebox{\imagebox}

\theoremstyle{remark}
\newtheorem{prob}{Theorem}
\theoremstyle{remark}
\newtheorem{lemma}{Lemma}

\PassOptionsToPackage{pdftex,
pdfencoding=auto,
pdfnewwindow=true,
pdfusetitle=true,
bookmarks=true,
bookmarksnumbered=true,
bookmarksopen=true,
pdfpagemode=UseThumbs,
bookmarksopenlevel=1,
pdfpagelabels=true,
breaklinks=true,
colorlinks=true,
}{hyperref}

\begin{document}

\title{{Quantum machine learning with Adaptive Boson Sampling via post-selection}}

%\author{

\author{Francesco Hoch}
\affiliation{Dipartimento di Fisica, Sapienza Universit\`{a} di Roma, Piazzale Aldo Moro 5, I-00185 Roma, Italy}

\author{Eugenio Caruccio}
\affiliation{Dipartimento di Fisica, Sapienza Universit\`{a} di Roma, Piazzale Aldo Moro 5, I-00185 Roma, Italy}

\author{Giovanni Rodari}
\affiliation{Dipartimento di Fisica, Sapienza Universit\`{a} di Roma, Piazzale Aldo Moro 5, I-00185 Roma, Italy}

\author{Tommaso Francalanci}
\affiliation{Dipartimento di Fisica, Sapienza Universit\`{a} di Roma, Piazzale Aldo Moro 5, I-00185 Roma, Italy}

\author{Alessia Suprano}
\affiliation{Dipartimento di Fisica, Sapienza Universit\`{a} di Roma, Piazzale Aldo Moro 5, I-00185 Roma, Italy}

\author{Taira Giordani}
\affiliation{Dipartimento di Fisica, Sapienza Universit\`{a} di Roma, Piazzale Aldo Moro 5, I-00185 Roma, Italy}

\author{Gonzalo Carvacho}
\affiliation{Dipartimento di Fisica, Sapienza Universit\`{a} di Roma, Piazzale Aldo Moro 5, I-00185 Roma, Italy}

\author{Nicol\`o Spagnolo}
\affiliation{Dipartimento di Fisica, Sapienza Universit\`{a} di Roma, Piazzale Aldo Moro 5, I-00185 Roma, Italy}

\author{Seid Koudia}
\affiliation{Leonardo S.p.A., Leonardo Labs, Quantum technologies lab, Via Tiburtina, KM 12.400, 00131 Roma, Italy}

\author{Massimiliano Proietti}
\affiliation{Leonardo S.p.A., Leonardo Labs, Quantum technologies lab, Via Tiburtina, KM 12.400, 00131 Roma, Italy}

\author{Carlo Liorni}
\affiliation{Leonardo S.p.A., Leonardo Labs, Quantum technologies lab, Via Tiburtina, KM 12.400, 00131 Roma, Italy}

\author{Filippo Cerocchi}
\affiliation{Leonardo S.p.A., Cyber \& Security Solutions Division, Via Laurentina - 760, 00143 Rome, Italy}

\author{Riccardo Albiero}
\affiliation{Dipartimento di Fisica, Politecnico di Milano, Piazza Leonardo da Vinci 32, 20133 Milano, Italy}
\affiliation{Istituto di Fotonica e Nanotecnologie, Consiglio Nazionale delle Ricerche (IFN-CNR), 
Piazza Leonardo da Vinci, 32, I-20133 Milano, Italy}

\author{Niki Di Giano}
\affiliation{Dipartimento di Fisica, Politecnico di Milano, Piazza Leonardo da Vinci 32, 20133 Milano, Italy}
\affiliation{Istituto di Fotonica e Nanotecnologie, Consiglio Nazionale delle Ricerche (IFN-CNR), 
Piazza Leonardo da Vinci, 32, I-20133 Milano, Italy}

\author{Marco Gardina}
\affiliation{Istituto di Fotonica e Nanotecnologie, Consiglio Nazionale delle Ricerche (IFN-CNR), 
Piazza Leonardo da Vinci, 32, I-20133 Milano, Italy}

\author{Francesco Ceccarelli}
\affiliation{Istituto di Fotonica e Nanotecnologie, Consiglio Nazionale delle Ricerche (IFN-CNR), 
Piazza Leonardo da Vinci, 32, I-20133 Milano, Italy}

\author{Giacomo Corrielli}
\affiliation{Istituto di Fotonica e Nanotecnologie, Consiglio Nazionale delle Ricerche (IFN-CNR), 
Piazza Leonardo da Vinci, 32, I-20133 Milano, Italy}

\author{Ulysse Chabaud}
\affiliation{DIENS, \'Ecole Normale Sup\'erieure, PSL University, CNRS, INRIA, 45 rue d’Ulm, Paris, 75005, France}

\author{Roberto Osellame}
\affiliation{Istituto di Fotonica e Nanotecnologie, Consiglio Nazionale delle Ricerche (IFN-CNR), 
Piazza Leonardo da Vinci, 32, I-20133 Milano, Italy}

\author{Massimiliano Dispenza}
\affiliation{Leonardo S.p.A., Leonardo Labs, Quantum technologies lab, Via Tiburtina, KM 12.400, 00131 Roma, Italy}

\author{Fabio Sciarrino}
\email{fabio.sciarrino@uniroma1.it}
\affiliation{Dipartimento di Fisica, Sapienza Universit\`{a} di Roma, Piazzale Aldo Moro 5, I-00185 Roma, Italy}

%---------------------------------------------

\begin{abstract}
   The implementation of large-scale universal quantum computation represents  {a} challenging and ambitious task on the road to quantum processing of information. In recent years, an intermediate approach has been pursued to demonstrate quantum computational advantage via non-universal computational models. A relevant example for photonic platforms has been provided by the Boson Sampling paradigm and its variants, which are known to be computationally hard while requiring at the same time only the manipulation of the generated photonic resources via linear optics and detection. Beside quantum computational advantage demonstrations, a promising direction towards possibly useful applications can be found in the field of quantum machine learning,
   considering the currently almost unexplored intermediate scenario between non-adaptive linear optics and universal photonic quantum computation. Here, we report the experimental implementation of quantum machine learning protocols by adding {adaptivity} via post-selection to a Boson Sampling platform based on universal programmable photonic circuits fabricated via femtosecond laser writing. Our experimental results demonstrate that Adaptive Boson Sampling is a viable route towards dimension-enhanced quantum machine learning with linear optical devices.
\end{abstract}
%---------------------------------------------
\maketitle
%---------------------------------------------

\section*{Introduction}

The realization of a universal quantum computer is one of the most challenging tasks faced by the quantum information community. Among the possible platforms, photon-based architectures have some special properties, namely the capability of being transmitted over long distances and a strong robustness to decoherence, which have enabled their application for communication and cryptography tasks \cite{Flamini_rev}. However, photonic systems are inherently characterized by the challenge of given by the need of introducing photon-photon interactions between photonic quantum states. While recent experimental progress has been made in this direction \cite{Volz2014,Feizpour2015,Tiarks2016,Hacker2016,Stolz2022,Kuriakose2022,DeSantis17,StraunstrupNL}, further technological advancements are required for the implementation of highly efficient nonlinear gates. The difficulty of carrying out this class of operations, is a major challenge in the realization of photonic universal quantum computation with the gate-based model \cite{chuang2010}, since it requires at least two-qubit operations. Several alternative schemes have already been proposed to overcome this limitation, each with its own strengths and weaknesses. For instance, the universal scheme based on linear optics by Knill, Laflamme, and Milburn \cite{Knill2001} (KLM) requires adaptive measurement procedures, in which the outcomes of intermediate measurements can drive the rest of the computation, together with ancillary resources that preclude efficient implementations on a large scale. Other approaches exploit nonlinear effects to realize quantum gates but their implementation is challenging  %with a limited success rate 
\cite{Chang2014, Dutt2024}. More recent schemes tailored to reduce the resource overhead of the KLM scheme consist in the preparation of large entangled states of many qubits and adaptive measurements on single qubits in the so-called measurement-based quantum computing framework \cite{Briegel2009} or on sub-systems in the fusion-based quantum computing variant \cite{Bartolucci2023}. These schemes thus translate the experimental challenge from applying two-qubit gates to preparing appropriate entangled resource states and performing suitable adaptive measurements.

In parallel, in the era of noisy intermediate-scale quantum (NISQ) devices, other approaches have been explored to demonstrate quantum computational advantage with photons. These strategies have focused on simpler architectures using only linear optical operations and non-adaptive single-photon detection. The Boson Sampling \cite{AA} and Gaussian Boson Sampling \cite{Hamilton2017} computational models are believed to be non-universal and require to generate samples from classically-hard-to-simulate probability distributions such as the ones that regulate the photon-counting statistics of indistinguishable particles at the outputs of random interferometers. Nowadays, many experiments have provided evidence of reaching the quantum advantage regime for this specific sampling task \cite{Zhong_GBS_supremacy, zhong2021phaseprogrammable, Madsen2022}, stimulating, on the one hand, the debate about possible classical strategies to mimic the results \cite{MartinezCifuentes2023classicalmodelsmay,TensorNetworkGBS, stanev2023testing}, and, on the other hand, the investigation of possible use-cases of such simplified photonic processors for applications beyond the original task \cite{GBSGraphTheory1, GBSGraphTheory2, Banchi_vibronic, Jahangiri_point_process}. Starting from the linear-optics based Boson Sampling model \cite{AA}, it is known that, either by adding enough nonlinearities \cite{SpagnoloNLBS} or enough ancillary photons and adaptivity \cite{Knill2001}, it is possible to recover universal quantum computation models.

In this context, the intermediate regime between non-universal, linear-optics-based schemes and universal computation still needs to be thoroughly explored. Indeed, such an intermediate scenario where a moderate amount of adaptivity is added to linear optics could disclose interesting computational applications within the reach of current photonic technologies. Recent works showed the possibility of enlarging the spectrum of applications by employing intermediate measurements and adaptivity in a standard linear optical quantum experiment. Particular attention in this direction has been devoted to the field of quantum machine learning (QML): photonic implementations of reservoir computing and extreme learning machines have been proposed \cite{Nokkala2021,Innocenti2023, Zambrini_RC} and first proof-of-principle experiments have been reported \cite{Spagnolo2022, Suprano2024}. Other proposals involve the definition of the quantum optical analogues of neural networks, which can be trained to asses specific tasks \cite{steinbrecherQONN,ewaniukQONN,stanevQONN, quantumneuralnetwork}. Notably, a recent work \cite{Chabaud2021quantummachine} introduced a scheme for QML which starts from the photonic Boson Sampling paradigm and then adds measurements and adaptivity on a subset of the optical modes. 

In particular, one may use as a feature map, i.e. the encoding of classical data in quantum states, the output of a Boson Sampling interferometer with $n$ input photons in $m$ optical modes when a subset of $r$ photons is measured in $k<m$ modes. Each photon detected in the $k$ channels activates a unitary transformation on the remaining modes such that there is a correspondence between the configuration $\bm{p}=(o_0, o_1, \cdots,o_k)$ of the possible ways to detect $r < n$ photons in $k$ modes, that is the input classical data of the feature map, and the interferometer implementing a transformation $U_{\bm{p}}$. Then, the state $\ket{\psi_{\bm{p}}}$ describing $n-r$ photons in the $m-k$ modes after $U_{\bm{p}}$ is interpreted as the quantum encoding of the classical data associated to $\bm{p}$ (see also Fig.\ \ref{fig:concept}a). This scheme, which we define as Adaptive Boson Sampling (ABS), can be employed as a subroutine for nonlinear kernel estimations as well as to perform classification tasks via full quantum information processing \cite{Chabaud2021quantummachine}.

In this work, we discuss and report on the experimental implementation of the ABS paradigm {via post-selection} for QML purposes with systems of growing complexity. In particular, we implement the feature map procedure described previously by employing universal \footnote{Here, universal refers to the fact that these circuits can implement any linear optical computation, not any quantum computation.} and fully reprogrammable integrated optical circuits fabricated via the femtosecond-laser-writing method \cite{osellame2012femtosecond, Corrielli2021}. Our experimental implementation of the ABS paradigm uses two different platforms. The first experiment was carried out on a universal 6-mode integrated circuit where we injected two photons generated by a spontaneous parametric down-conversion source. {Then, we scaled up the size of the ABS by injecting two and three photons on a universal 8-mode chip coupled to a bright semiconductor quantum dot source, which allowed us to perform a more sophisticated scheme both in terms of kernel size and quantum state dimension. In particular, the number of measured modes $k$ was $k=3$ for the 6-mode chip while $k=6$ and $k=5$ for 8-mode device (depending on the dimensionality of the produced output state)}. We make use of the implemented quantum kernel to successfully carry out different classification tasks of 1D and 2D datasets. Finally, we discuss the scaling and future applications of such an approach.

%---------------------------------------------

\begin{figure*}[t]
    \centering
\includegraphics[width=0.95\textwidth]{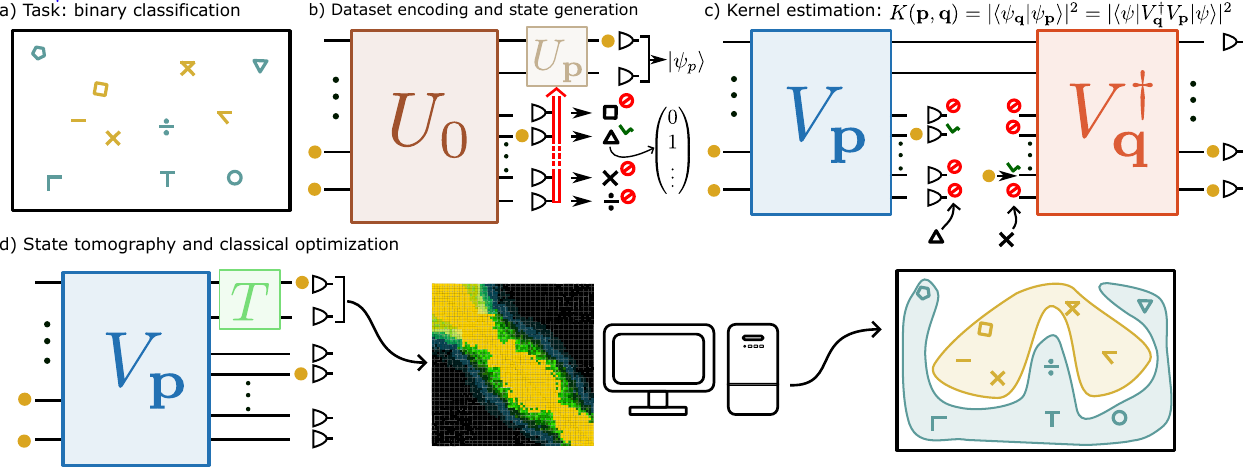}
\caption{\textbf{Adaptive boson sampling, tailored for quantum machine learning, via measurement post-selection}. a) The algorithm aims at solving a binary classification problem: in a 2D plane filled with different shapes, the goal is to classify the items according to the color feature. b) Each point of the dataset is encoded in a quantum state $|\psi_{\bm p}\rangle$, according to the optical circuit output mode, with which such state is triggered. Actually, in the ABS theoretical scheme, the detection of a photon, exiting from $U_0$, in one of the lower modes determines a specific adaptive transformation $U_{\bm p}$ that cooperates in the generation of the state $|\psi_{\bm p}\rangle$.  c) The dataset is classified using a kernel method, specifically a SVM. The kernel elements, defined as the overlap square moduli, can be directly derived from the sketched linear optical circuit, in which $V_{\bm p}$ is composed of $U_0$ and $U_{\bm p}$. The square modulus of the overlap can be experimentally obtained through a measurement post-selection of the fraction of coincidences in which photons leave the circuit from the same modes as they entered. d) Another way to evaluate the kernel arises from a post-selection reconstruction of the states with a tomography protocol,  exploiting the projective unitary $T$ that acts on the adaptive modes. This is the case of the experiment we implement here. After that, the kernel is provided to a classical hardware which manages the binary classification task. \label{fig:concept}}
    
\end{figure*}

\section*{Results}

\subsection*{Background} %and theoretical simulations}

In this section, we review the main concepts behind the ABS paradigm, as well as the model we consider for designing the experiment.
The ABS model requires photon-counting measurements on a subset of outputs of a standard Boson Sampling experiment \cite{Chabaud2021quantummachine}. Furthermore, the scheme includes adaptive operations that are activated by the measurement outcomes.
Let us now introduce the formalism of such a computational model. Let $n$ be the total number of input photons that are injected in a $m$-port interferometer. We indicate with $\bm{n}=(n_0, n_1, \cdots, n_m)$ the string which describes the arrangement of the $n=\sum_{i=0}^{m}n_i$ photons in the $m$ input ports. Adaptive measurements are then carried out on $k<m$ modes by detecting $r<n$ photons. The string defined as $\bm{p}=(o_0, o_1, \cdots, o_k)$, such that $r=\sum_{j=0}^{k} o_j$, indicates the output photons configuration detected at the $k$ modes. Each outcome in $\bm{p}$ activates a unitary operation on the modes that are not measured, in such a way that there will be a relationship between $\bm{p}$ and the adaptive interferometer $U_{\bm{p}}$ (see Fig.\ \ref{fig:concept}a). The output state of the device is then a multi-photon state of $n-r$ photons encoded in the $m-k$ modes. Hereafter we identify an ABS device through these parameters $[m,n,d,D]$. The parameter $d= \binom{m-k+n-r-1}{n-r}$ is the dimension of the Hilbert space of the output state. The last parameter $D$ indicates the number of classical strings that can be encoded in the ABS, which corresponds to the number of performed adaptive measurements.

The ABS scheme finds applications in the quantum machine learning context as previously demonstrated in Ref. \cite{Chabaud2021quantummachine}. In particular, the ABS feature maps can be employed for kernel-based methods. The quantum device computes the kernel between two data $\bm{p}$ and $\bm{q}$ as shown in Fig.\ \ref{fig:concept}{c}. The quantum algorithm requires to apply  the circuit $U_{\bm q}^\dagger$ ($U_{\bm p}^\dagger$) to $\bm{q}$ ($\bm{p}$) every time the data $\bm{p}$ ($\bm{q}$) is measured. The kernel element estimation will be given by the number of times the input state of the protocol is detected at the output, divided by the total number of rounds. {Such a procedure allows for an efficient estimation of kernels from single-photon counts as demonstrated in Ref. \cite{Chabaud2021quantummachine}.} 

In this work, we experimentally investigate the feature map scheme illustrated in Fig. \ref{fig:concept}b and apply it to a kernel estimation task.  As we will show in the following, the kernel will be computed through quantum tomography {(see Fig. \ref{fig:concept}d)} and quantum state fidelity rather than through the overlap estimation presented in Fig. \ref{fig:concept}c. Such a choice allows us to take into account experimental imperfections that could lead to the generation of imperfect %pure
states. However, quantum state tomography is a viable route only at low dimensionality while, at higher dimensions, i.e. larger number of photons and modes, the approach of Fig. \ref{fig:concept}c remains the most efficient way to compute kernels. Furthermore, the state fidelity is mathematically equivalent to the overlap estimation only for output states with high levels of purity.

{The present implementation relies on the emulation of the ABS dynamics through post-selection, allowing us to investigate the properties and functionalities triggered by the addition of adaptivity in a Boson Sampling experiment. Such a limitation derives from the current technological challenges in implementing fast enough reconfigurability on the scale of integrated optical circuits \cite{Wang2020_review}. However, we stress that the results obtained after the post-selection and processing cannot be reproduced with a single, non-adaptive Boson Sampling instance.}

\subsection*{Experimental apparatus and data analysis}

\begin{figure*}[t]
    \centering
\includegraphics[width=1\textwidth]{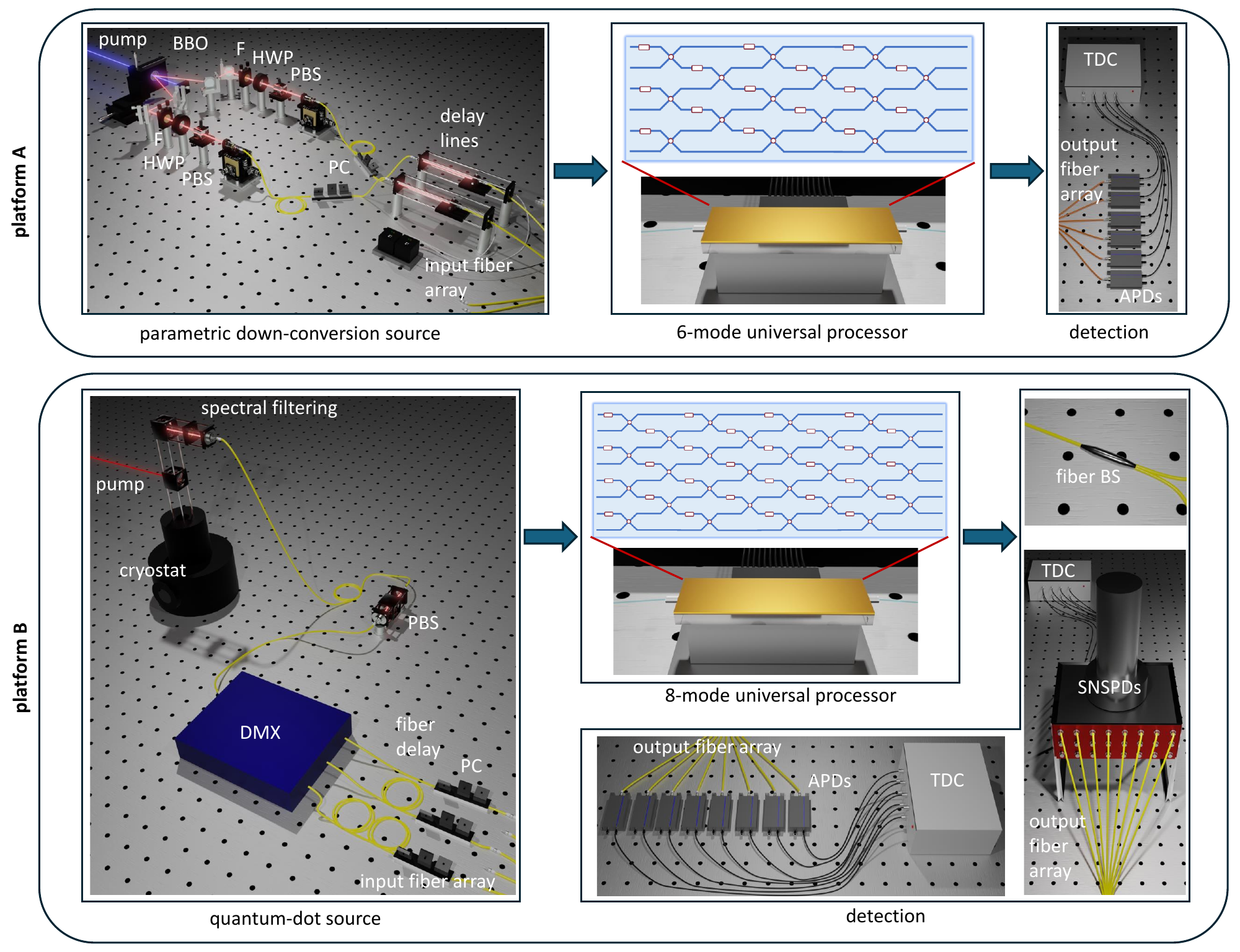}
    \caption{\textbf{Experimental platforms.} a) Platform A. Such a platform used for 2-photon experiments in 6 modes, envisages a parametric down-conversion source that generates pairs of photons at 785 nm and a 6-mode universal programmable integrated optical circuit. The two photons are synchronized in time through delay lines. Polarization controllers and filters are employed to have fully indistinguishable photons.
The operations of the chip are controlled via a power supply that applies currents to the heaters of the device. Finally, the time-to-digital converter processes the single-photon detector counts that are then analyzed for the experiment. b) Platform B. In the second platform, we employ a semiconductor quantum dot source and an 8-mode universal programmable chip. The brightness of the
source enables the implementation of up to 3-photon experiment. This time the photons emitted by the same quantum dot at different pump pulses are synchronized by a time-to-spatial demultiplexer in three different channels. Photon-detection has been performed by either avalanche photodiodes or by superconducting nanowire detectors. Photon number resolution in some of the experiments has been added by employing a probabilistic scheme based on mode-splitting via fiber beam-splitters. Legend: BBO (Beta-Barium Borate), F (frequency filter), HWP (half-wave plate), PBS (polarizing beam-splitter), PC (polarization compensation), APD (avalanche photodiode), TDC (time-to-digital converter), DMX
(demultiplexing), SNSPD (superconducting nanowire single-photon detectors), BS (beam-splitter). \label{fig:setup}
    }
    
\end{figure*}

%\subsection{Experimental apparatus and data analysis}

 %in we apply the feature map based on the ABS paradigm %in proof-of-principle experiments
% integrated optical devices. 
In the following, we report the results of the experimental implementation of ABS feature maps of increasing dimension by making use of $m$-mode universal integrated optical circuits and single-photon sources based both on parametric down-conversion and semiconductor quantum dots. As we will describe in what follows, we experimentally demonstrate four ABS schemes of increasing complexity following two main guidelines. On the one hand, we provide a demonstration of the ABS paradigm in a two-photon setting, by employing established components allowing for high quality interference between the photons (Platform A), thus providing a first assessment of the algorithm operation with readily available technology. On the other hand, to further benchmark and assess the capabilities of the method, we employ state-of-the-art components (Platform B) achieving implementations of the ABS paradigm of increasing complexity (B1, B2, B3), given the relevance of the method in view of long-term applications.
The two platforms are shown in Fig. \ref{fig:setup}. Both integrated devices, specifically a six- and eight-mode universal fully reconfigurable circuits, are fabricated through femtosecond laser writing writing \cite{Gattass2008, Corrielli2021}. The on-chip operations are controlled %changing the operation of the internal interferometers 
by thermo-optical phase shifting, through the application of external voltages over the 30 and 56 heaters on each integrated device. In particular, the optical circuits were developed according to the universal design reported in \cite{Clements:16}, in which variable beam-splitters and phase shifters enable the implementation of arbitrary unitary transformations. Our platforms are then exploited as a proof-of-principle for the application %implementation (?)
of a QML feature map based on the ABS paradigm. At this stage, the time response of the reconfiguration of the circuits is not fast enough to permit an active modulation modulation of the circuit based on the measurement outcomes; this implies that we performed the experiment in post-selection.

{In platforms A and B1, the circuit was programmed to demonstrate the experimental feasibility of the ABS protocol as described in \cite{Chabaud2021quantummachine}, following a structure with a cascaded set of adaptive unitaries $U_i$. This choice led to particularly symmetric configuration, in order to maximize the interference between the input photons while keeping the bunching probability at the output low and enhance the detection probability of coincidence events in post-selection. Specifically, as shown in Fig.\ref{fig:exp}, the adaptive portion of each $U_i$ was placed on the same diagonal, while the rest of the MZIs are all set at $\pi/4$ both for the internal and external phases.}

{\bf Platform A}: As initial step, we employ an integrated device with $m=6$ modes, designed to work with single-photon states at a wavelength of $\lambda=785$ nm. The single-photon source is a nonlinear crystal, that generates pairs of highly indistinguishable photons via the parametric down-conversion process. 
The parameters $[m,n,d,D]$ are $[6,2,2,3]$. %implements an ABS with $n=2$ 
We injected one photon in mode 3 and the other in mode 6 and we restricted to the scenarios for which one photon is detected in one of the three modes (2,3,6). %is the detection of $r=1$ photons in one of the three modes (2,3,6).
Such a measurement maps the strings $\bm{p}_i$, $\bm{p}_1 = (1,0,0)$, $\bm{p}_2 = (0,1,0)$ and $\bm{p}_3 = (0,0,1)$, into the qubits $\ket{\psi_{\bm{p}_i}}$ encoded as one-photon in the modes 4 and 5 of the device. Note that $o_0$ is always zero due to the interferometer connectivity, since photons enter from modes 3 and 6, meaning that we can encode only 3 different strings $\bm{p}$. The adaptive transformations $U_i$ are implemented by tuning the reflectivities $\theta_i$ of the beam-splitters highlighted in Fig.\ \ref{fig:exp}. These angles depend on the strings $\bm{p}_i$ as $\theta_i=\frac{\pi}{2}\sum_{j=0}^{j=i-1}o_j$, where the $o_j$ is the number of photons (0 or 1) detected in the modes 1, 2, 3 and 6 for $j=0$, $j=1$, $j=2$, $j=3$ respectively (see also Fig.\ \ref{fig:exp}a). Intuitively, the adaptive rule for the reflectivities works as follows: for example, if we measure $o_3=1$ we will apply the interferometer unitary $U_3$ which has $\theta_1 = \theta_2 =\theta_3 =0$. {In other words, the reflectivity of the $i$-th adaptive unitary will depend on the results of photon-counting measurements of the other photons in the previous modes, i.e. up to the $i$-th mode.}
The reconfigurable circuit is set to implement one of the $U_{i}$ and we perform tomography of dual-rail encoded single photon qubit states in the modes 4 and 5 conditioned to the presence of the other photon in the corresponding detector $o_i$. Such a tomography allows us to estimate the density matrices $\rho_{i}$ of the three qubits and the state fidelity $\mathcal{F}_i$ with the expected state, as well as to compute the $3 \times 3$ kernel through the mutual quantum state fidelity $ K(\bm{p}_i,\bm{p}_j) = \mathcal{F}(\rho_i,\rho_j)$ (see Supplemental Information). 
The fidelities with the theoretical pure qubits generated by an ideal ABS circuit, i.e. with perfect photons indistinguishability and perfect settings of the circuit are $\mathcal{F}_1 = 0.981 \pm 0.003$, $\mathcal{F}_2 = 0.999 \pm 0.001$, $\mathcal{F}_3 = 0.998 \pm 0.002$. These values confirm a very good agreement with the expectation and thus a very good level of photon indistinguishability as well as chip control. %Such a good agreement is confirmed by the comparison between the experimental density matrices obtained through tomographic analysis and the theoretical ones (see Fig.\ \ref{fig:exp}b)-d)). 
In Fig. \ref{fig:exp}b, we report the results regarding kernel estimation, while in panel c, the density matrices reconstructed from the tomography for the states $\rho_1$ and $\rho_2$. Further details about data analysis, the measurements and comparison with theoretical models % More details about the data analysis 
can be found in the Supplemental Information.\\

{\bf Platform B}: As the following step, we have performed additional experiments by using a bright quantum dot single-photon source, a time-to-spatial demultiplexer, a programmable 8-mode integrated device, avalanche photo-diodes (APD) and a Single Quantum Eos system of superconductive nano-wires single-photon detectors (SNSPDs)  (see Fig.\ref{fig:setup}b). By adopting this advanced apparatus we could investigate the capacity to enlarge the size of the kernels by exploring the Hilbert space of multi-photon states. The $m=8$ chip configuration is depicted in Fig.\ \ref{fig:exp}d and was designed to operate with single-photon states at the wavelength $\lambda=927$ nm. The generation and preparation stage allows for the synchronization of up to 3 photons generated by the same quantum dot source at different time pulses. A preliminary two-photon experiment with such a platform has been carried out and the results are reported in Supplemental Information. In the following, we show the ABS paradigm with three-photons.

%caso a 3 fotoni%
{\bf B1-} The first ABS encoded in this second platform aims at enlarging the number of adaptive measurements by directly exploiting the larger circuit depth given by the 8-mode device,
%The second experiment carried out with the 8-mode device aimed 
and at increasing the number of photons injected in the interferometer. Here, the detection is at first performed by APDs as in the previous platform. %in the interferometric setup. 
In particular, we injected 3 photons in the interferometer: this implies that the number of classical strings $\bm{p}$ that can be encoded is at least $D=\binom 62=15$, which are the possible ways to measure $2$ photons in the $k=6$ adaptive channels discarding the configurations which feature two photons in the same mode. This means that this ABS scheme is labeled as $[8,3,2,15]$ and encodes 15 classical strings into 15 dual-rail encoded qubits, describing the state of the third remaining photons exiting the interferometer from modes 6,7. Thus, we implement 15 different adaptive unitaries $U_i$ each of them associated with a different position of the pair of detected photons. More precisely, we employ as an adaptive rule for the reflectivities of the beam-splitters the same formula of the previous experiment. %The structure of the optical circuit (see Fig.\ \ref{fig:exp}d) is the same as the $[8,2,2,3]$ scheme, but this time, 
This time, each pair of detected photons activates a different configuration of the 5 angles $\theta_i$. In such a scenario, we obtain a $15 \times 15$ kernel that summarizes the mutual overlaps of the $15$ dual-rail qubits encoded in modes 6 and 7 of the chip, reconstructed via quantum state tomography. In Fig.\ref{fig:exp}e, we report the experimental kernel and the comparison with the expected results according to the theoretical modeling of an imperfect single-photon source. Such a model takes into account both the partial distinguishability of the photons generated by the source and the multiphoton contributions due to a second-order non-zero correlation function \cite{Olivier_2021, Pont_2022, Valeri2024} {(further details can be found in the Supplementary Note 6).} %The kernels $K(\bm{p}_i, \bm{p}_j)$ have been computed according to %an ideal experiment and according to our
%such a theoretical modeling and compared with the 
 The average measured fidelity between the density matrices measured in the experiment and the ones calculated according to the model is $\bar{\mathcal{F}} = 0.94 \pm 0.02$.  %\textcolor{red}{Analysis.QML.final con montecarlo e $\sigma_{F} = \sigma_i/\sqrt{15}$}.

{As a following step, the integrated circuit was programmed with a different geometry, as shown in Fig. \ref{fig:exp_8m3p_2}, in order to design an adaptive scheme that could be used in a machine learning context. Here we increase the depth of the static unitary $U_0$ and consider a smaller two (three) mode adaptive operation $U_i$. In this case, $U_0$ was chosen by drawing a random interferometer in such a way that sufficient statistics could be recorded in a post-selected regime. On the other, the adaptive transformation $U_i$ were designed in order to gradually span the underlying feature space: such an approach was inspired by classical kernel methods, specifically to Gaussian kernels \cite{yang2021parameter}. This will allow us to showcase a practical use of quantum states obtained via the ABS scheme and its corresponding kernels, in terms of a classification task.}

\begin{figure*}[t]
    \centering
    \includegraphics[width = 0.99\textwidth]{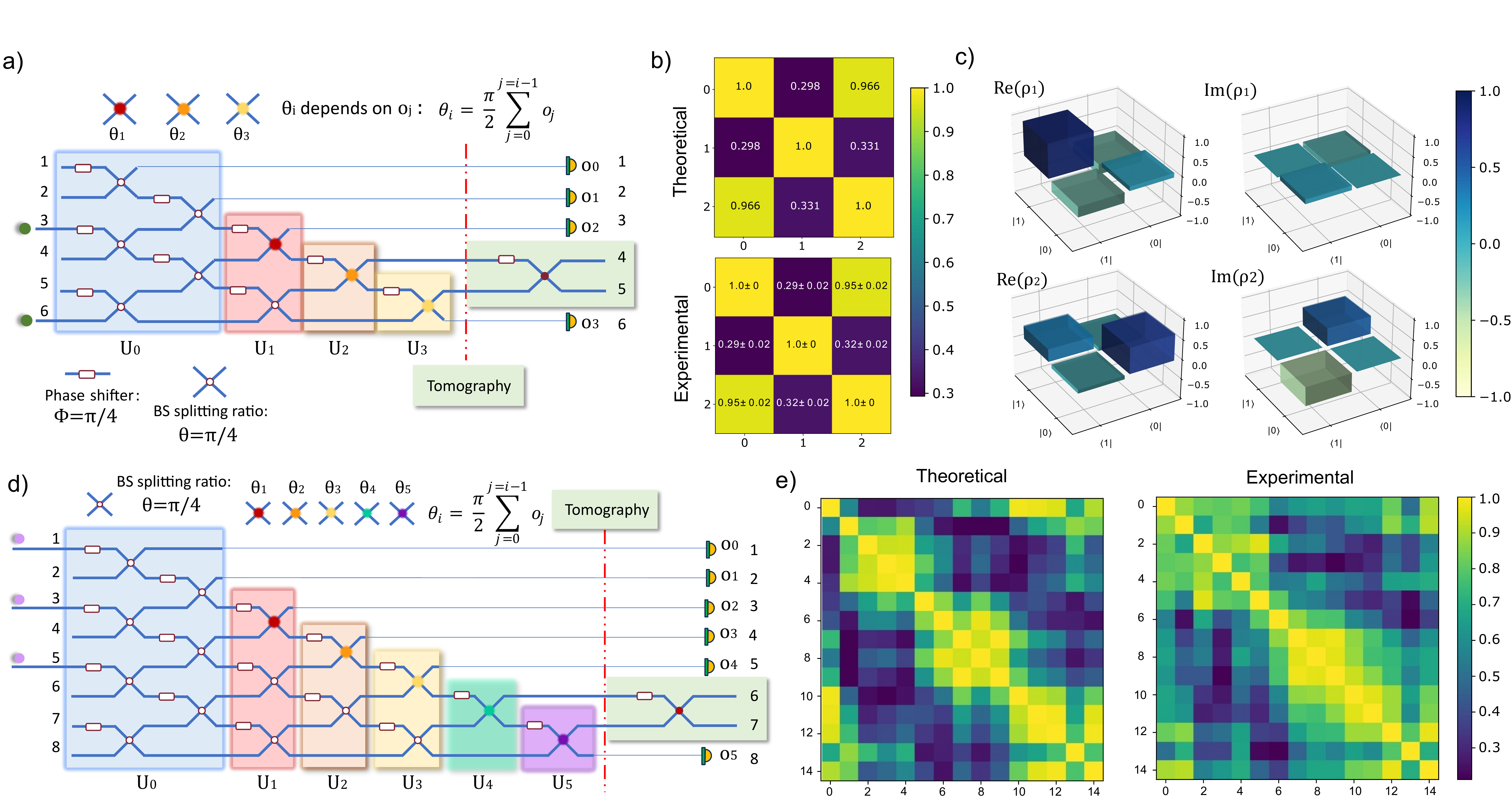}
    \caption{\textbf{Two and three photons in ABS interferometers - platforms A and B1.} a) The experiment implements an ABS scheme $[6,2,2,3]$. % feature map for qubits by making interfere $n=2$ photons in $m=6$ modes and applying $k=3$ adaptive measurements %$(o_1, o_2, o_3)$ 
    %on $r=1$ photons. 
    The circuit is encoded in a 6-mode universal programmable chip. In particular, we have a six-mode $U_0$ and then three adaptive transformations $U_i$. The phase shifters $\phi$ (rectangles in the figure) and beam-splitter reflectivities $\theta$ (circles) are set to angles $\theta,\phi = \pi/4$ except for the $\theta_i$ of the $U_i$ highlighted in red, orange and yellow that depends on the detection of one photon in the $o_j$. Pairs of indistinguishable photons generated by parametric down-conversion evolve in such an interferometer and a qubit tomography conditioned to the detection $o_i$ is performed in the green part of the circuit. b) Comparison between the numerically simulated kernel %
    and the experimentally one, the latter computed via the mutual state fidelity between the states reconstructed at the output of the programmable integrated optical circuit. c) Experimental $\rho_{i,\textrm{exp}}$ density matrices for the quantum states $\rho_{1}$ (top) and $\rho_{2}$ (bottom). We retrieved the density matrix by performing the qubit tomography with a tunable beam-splitter and phase-shifter. Uncertainties due to photon-counting statistics are smaller than the image scale. d) 3-photon experiment in the 8-mode device. In this scenario we have $r=2$ photons detected in $6$ adaptive modes. We have a total amount of $15$  transformations each of them triggered by the detection of two photons in a pair of the $6$ outputs. The optical circuit is divided into an 8-mode unitary $U_0$, and five transformations $U_i$, activated and combined according to the configurations of the $r$ photons detected in the $6$ output modes. The reflectivity values of the beam-splitters $\theta_i$ in red, orange, yellow, teal, and violet, depend on where the ancillary photons are detected, according to the formula displayed in the figure.
    e) Comparison of the $15\times 15$ kernels computed according to the theoretical modeling which assumes an imperfect single-photon source and the kernel reconstructed from the states measured at the output of the ABS scheme $[8,3,2,15]$. Both experiments were carried out with APDs. \label{fig:exp}
    }
\end{figure*}

{\bf B2-} We then implemented an alternative $[8,3,2,15]$ scheme reported in Fig. \ref{fig:exp_8m3p_2}a. This time we made use of superconductive nanowire single phitin detectors (SNSPDs), having higher detection efficiency at 927 nm. The goal of this further experiment is to engineer the adaptive operations to generate kernels which can be useful for classification. Here, again, we have 15 different unitaries $U_i$ each of them associated with a different pair of detected photons on the adaptive modes. At variance from the previous scenario, the programmable device has been divided into two sections, one implementing a fixed and randomly extracted transformation $U_0$ and the other a 2-mode adaptive unitary $U_i$ (Fig.\ \ref{fig:exp_8m3p_2}a). The transformations have been designed in order to obtain %output states resulting in 
a $15 \times 15$ kernel resembling the characteristics of a Gaussian one. In particular, depending on the string $\bm{p}=(o_0,\dots, o_5)$ measured in the adaptive modes, both the reflectivity of the beam-splitter and the phase in $U_i$ vary as $ \phi_i = \theta_i=\frac{k+\sum_{t=1}^j(5-t) }{5}$, where $(k,j)$ are the indices of the outputs in which the two photons are detected. %are incremented by $i/5$ where $i$ varies from 1 to 15. 
The kernel is estimated as before after performing, in post-selection conditions, a tomography of the dual-rail encoded qubit in modes 2 and 3. In the right panel of Fig.\ \ref{fig:exp_8m3p_2}a, we report the experimental kernel and the expected one. In this case, the average observed fidelity is $\bar{\mathcal{F}} = 0.987 \pm 0.003$ (see Supplemental Information). These results validate the capability of the Adaptive Boson Sampling scheme, implemented on a hybrid photonic platform, to engineer kernels by properly designing the set of implemented transformations and the correspondence between output states $\rho_i$ with output strings $\bm{p}_i$ mapped into the post-selection modes.

{\bf B3-} Finally, we implemented a %fourth 
third experiment aimed at increasing the number of final output modes, in this case to a dimension $d=3$ thus leading to qutrit states. Moreover, unlike the previous experiments, we also measured the events which feature bunching in the adaptive modes.  We added an in-fiber beamsplitter at output
mode 6 and, whenever we wanted to resolve a two-photon bunched state in one of the outputs of $U_0$, we programmed the bottom part of the chip to direct such mode to
the pseudo-number resolving configuration.
This implies that the number of classical strings $\bm{p}$ that can be encoded here is $D = \binom{5+2-1}{2}=15$, which are the possible ways to measure $r=2$ photons in the $k=5$ adaptive channels by also considering the configurations with two photons in the same mode. In this way, we implemented a $[8,3,3,15]$ scheme: here, the 15 different adaptive unitaries - each of them associated with a different pair of detected photons - were again designed in order to obtain a feature map leading to a simil-Gaussian kernel. The structure of the optical circuit is similar to configuration B2 and is reported in  Fig.\ \ref{fig:exp_8m3p_2}b. The reflectivities of the beam-splitters and the phases in the $U_i$ vary as  in the previous scheme B2 for $i \in [0,10]$, $ \phi_i = \theta_i=\frac{k+\sum_{t=1}^j(4-t) }{5}$, where $(k,j)$ are the indices of the outputs in which the two photons are detected, while, for $i>10$, $ \phi_i = \theta_i=\frac{11+j}{5}$, where $j$ is the output index where both photons are detected. Notably, unlike the previous experiments we now encode the output states as a single photon state in a superposition of three spatial modes. With this choice we implement a three-dimensional encoding, i.e. a qutrit, and to reconstruct its quantum state we thus need to carry out a measurement of the generalized Pauli operators in modes 1, 2, 3. % need now to reconstruct the quantum state of three-dimensional states: this was carried out by measuring the generalized Pauli operators in modes 1,2 and 3. 
In the right panel of Fig.\ \ref{fig:exp_8m3p_2}b, we report the experimental kernel and the comparison with the expected results according to the theoretical modeling of an imperfect single-photon source. In this case, in Fig.\ \ref{fig:exp_8m3p_2}c, we also report an example of the reconstructed density matrices for two states while the fidelity with the expected qutrit states, obtained averaging over the 15 quantum states being realized in the experiment, is $\bar{F} = 0.963 \pm 0.005$. Further details are reported in Supplemental Information.

\begin{figure*}[t]
          \centering
    \includegraphics[width = 1\textwidth]{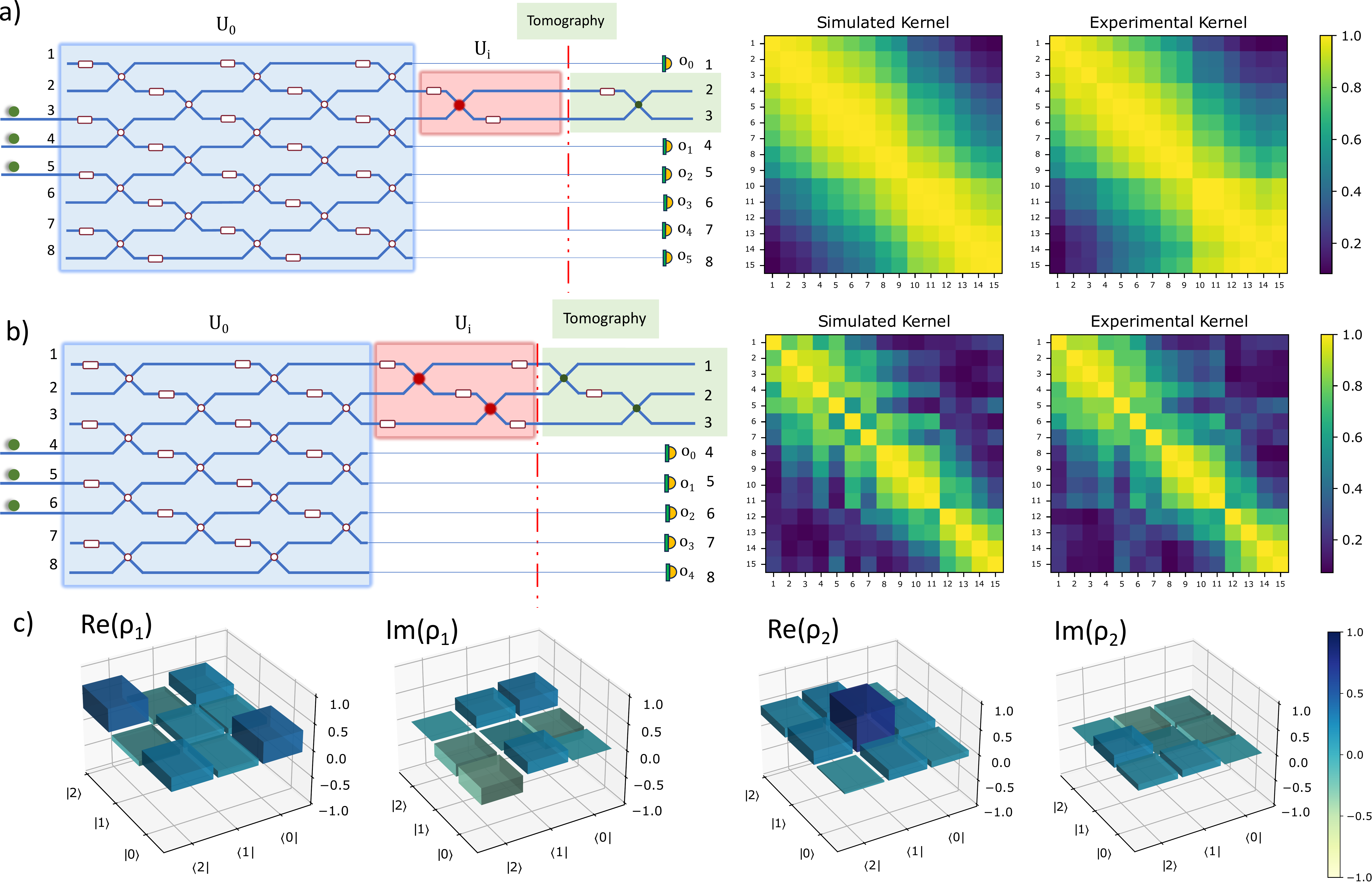}
    \caption{\textbf{Three photons in an eight-mode adaptive Boson Sampling interferometer for generating Gaussian kernels.} a) The  $[8,3,2,15]$ ABS scheme of platform B2. We synchronize $n=3$ photons emitted from the quantum dot source and we process them in the $m=8$ mode universal integrated circuit. The optical circuit is divided into an 8-mode randomly extracted unitary $U_0$ and a 2-mode adaptive unitary $U_i$. Triggered by the detection of $r=2$ photons in the $6$ adaptive modes $o_j$, the reflectivity value of the beam-splitter $\theta_i$ assumes $15$ different values allowing for the reconstruction of $15 \times 15$ kernels. The green part of the circuit highlights the tomography station in which the dual rail qubit encoded in the remaining photon conditioned on the detection of the $r$ photons in the other $o_j$ outputs is analyzed. On the right panel, we report the comparison between the $15 \times 15$ kernel simulated according to the theoretical model and the experimental one.
    b) The $[8,3,3,15]$ scheme that encodes classical data in qutrit states (platform B3). The 8-mode $U_0$ is followed by 3-mode adaptive $U_i$ in which the reflectivity of two beam-splitters has been properly programmed in order to implement $15$ different unitaries as before. Triggering on the detection of $r=2$ photons in the $5$ adaptive modes $o_j$, considering both configurations in which photons are bunched in the same mode or are output from different modes, a kernel $15 \time 15$ has been reconstructed. The green part of the circuit is again the tomography station that analyzes the three-rail qutrit. %encoded in the remaining photon conditioned on the detection of the $r$ photons in the other $o_j$ outputs. 
    The right panel reports the comparison between the $15 \times 15$ kernel simulated according to the theoretical model and the experimental one. %computed with different $V_i$ depending on the detection of the $r=2$ trigger photons. 
    c) Experimental $\rho_{i}$ density matrices for qutrits $\rho_1$ and $\rho_2$ to which correspond the following fidelity with the expected theoretical state, $F_1=0.995 \pm 0.001$ and $F_2=0.953 \pm 0.010$. The experimental density matrices are reconstructed by measuring in the tomography stage the generalized Pauli operators. Both experiments reported here were carried out with SNSPDs. \label{fig:exp_8m3p_2}}
    
\end{figure*}

\subsection*{Classification of data via Adaptive Boson Sampling}

Let us now show some practical applications of the ABS scheme for data classification, a prototypical machine learning task. The key idea here is to use the feature map generated by the ABS paradigm in the context of a Support Vector Machine (SVM). In this framework, the machine is trained using the quantum kernels derived from the ABS output states. The kernels can be estimated through the inner product $K(\bm{p},\bm{q})=|\langle\psi_{\bm{q}}|\psi_{\bm{p}}\rangle|^2$ as proposed in the original theoretical work \cite{Chabaud2021quantummachine} (see Fig.\ \ref{fig:concept}b), or via the state fidelity $K(\bm{p},\bm{q})= \mathcal{F}(\rho_{\bm{p}}, \rho_{\bm{q}})$ as we did in the experiment. Note that the two estimates are equivalent for pure states. 
In the following, we use the quantum kernels collected in the experiment to solve binary classification problems for both 1D and 2D datasets. 

{\bf i) 1D dataset classification}. We consider a dataset comprising 15 labeled points, each of them to be encoded in feature maps implemented in the platforms B2 and B3. The dataset consists of $15$ 1D data points $x_p$ with a binary label $y_p \in (-1,1)$. The labels are assigned so that the dataset is not linearly separable (see Fig.\ \ref{fig:class}a). Each data point is assigned to one of the 15 measurement outcomes of post-selected modes, indicated in Fig.\ \ref{fig:exp_8m3p_2} and to the corresponding outcome quantum state, qubits for platform B2 and qutrits for platform B3, according to the quantum feature map given by the ABS. Then, the quantum kernel will be represented the $15 \times 15$ matrix shown in Fig.\ \ref{fig:exp_8m3p_2}.
The dataset is randomly divided into a training set and a test set with the following train-test split: $80\%$ for training and  $20\%$ for the test. Due to the limited size of the dataset, to ensure that the choice of training and test set does not favourably bias the accuracy of the classification, we employ cross-validation by averaging the accuracy $A$ of the model on the test set over $100$ possible random partitions of the data in training and test sets. The results obtained via this cross-validation procedure with the two experimental kernels, corresponding to qubit and qutrit states, are $A = 0.90$ and $A = 0.80$ respectively, thus achieving successful classification. We also report in the histograms in Fig.\ \ref{fig:class} b-c the performances of other $50$ kernels obtained in scenarios B2 and B3 with different assignments of the post-selected modes to the adaptive operations. We report in the Supplementary Note 4 more details about the procedure by which these further kernels are obtained. 

 \begin{figure}[t]
    \centering
    \includegraphics[width = \columnwidth]{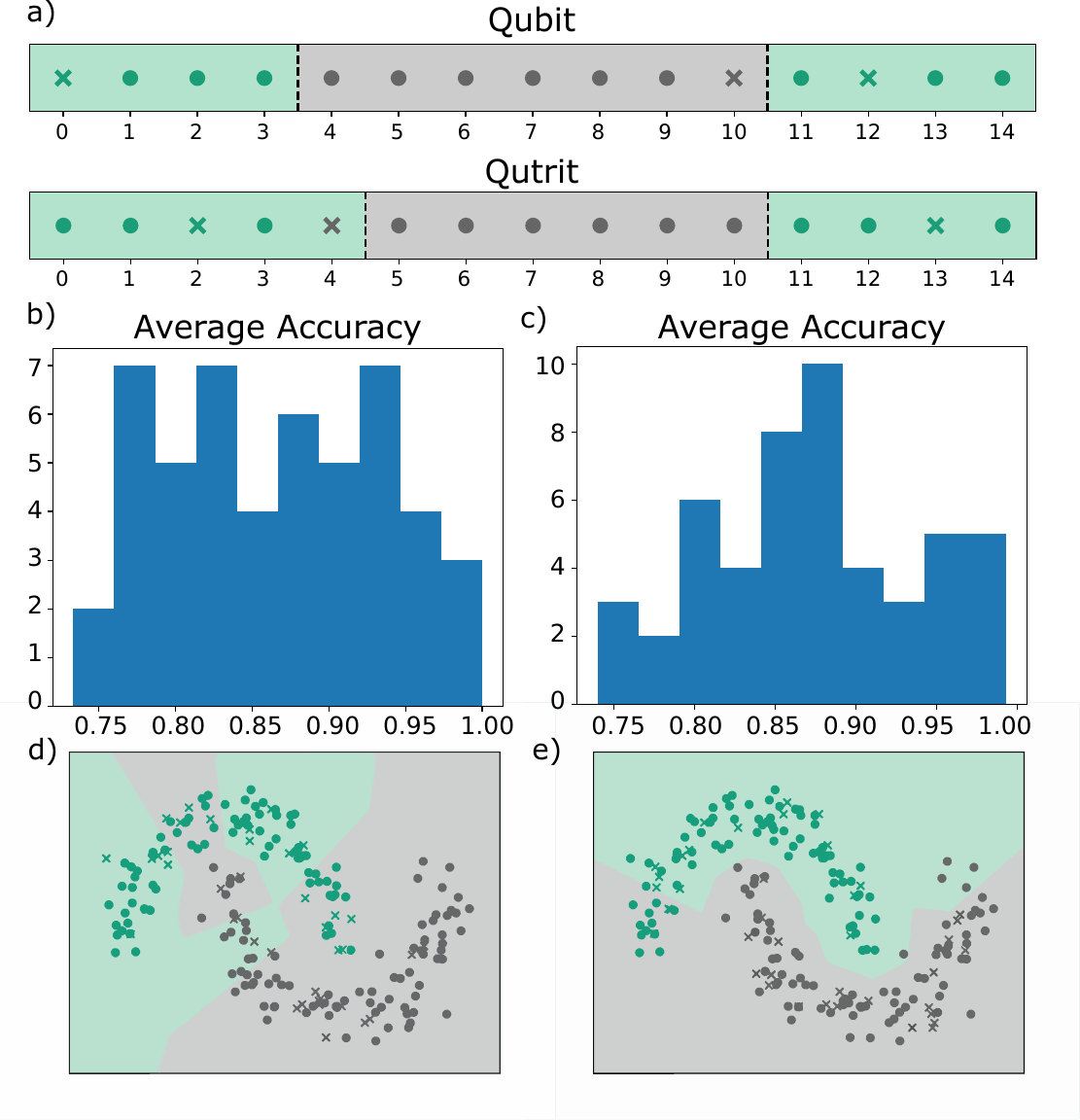}
    \caption{\textbf{Classification of 1D  and 2D datasets.} a) Example of classification of the 1D dataset performed by the SVM with the two quantum experimental kernels obtained from qubit states and qutrit states. The labels $y_i$ of the dataset are shown in the colors green and gray, while the symbols ‘o’ and ‘x' indicate the training and test set respectively. The background color represents the result of the classification. b-c) Histograms of the average accuracy for classification with kernels collected in the experiment by choosing different assignments of the post-selected modes to the adaptive operation. For each kernel, the accuracy is averaged over $100$ random partitions of data into training and test sets. In b) results for qubits' kernels and in c) the qutrits case. d-e) Classification of a 2D dataset done with a preliminary clustering algorithm (\textit{K-means}) followed by the application of a SVM with the two quantum experimental kernels obtained from d) qubit states and e) qutrit states. The correct label $y_i$ of the dataset is shown with the color (green/gray) of the symbols (‘o’: training set, ‘x’: test set). The background color represents the result of the classification. \label{fig:class}}
    
\end{figure}

{\bf ii) 2D dataset classification}. 
We then apply the ABS kernel matrices
to classify a dataset with more data points and higher dimensionality. We consider a 2D dataset of $200$ data points $r_i$ with label $y_i \in (-1,1)$ generated with the \texttt{make$\_$moons()} function of the \texttt{scikit-learn} Python library. The dataset is again not linearly separable in a 2D space (see Fig.\ \ref{fig:class}d). In order to classify the dataset with the quantum kernels and the SVM, one needs to associate to each data point from the continuous space $r_i$ a discrete index $x_{\bm p} \in [0,14]$ so that the kernel element associated to each pair of data $(r_i, r_j)$ is $K(r_i, r_j) = K(x_{\bm p}(r_i), x_{\bm q}(r_j)) = K(\bm{p},\bm{q})$. This map can be obtained through a standard clustering algorithm, such as the \textit{K-means} algorithm \cite{macqueen1967some}. 
Thus, we create $15$ clusters via pre-processing the data via a K-means algorithm. Then, each data point $r_i$ is associated with the corresponding cluster centroid index $x_{\bm p}$. At this stage, the dataset $\{r,y\}$ is randomly divided into training set and test set according to the proportions $80\%$ and $20\%$ and the SVM is trained with the quantum kernels. Again, we use the experimental data collected in platforms B2 and B3, and the output accuracy $A$ is averaged over $100$ different partitions of the data in training and test sets, in a cross-validated fashion. The results obtained with the two experimental kernels corresponding to qubit and qutrit states are $A = 0.65$ and $A = 0.90$ respectively (see Fig.\ \ref{fig:class}d-e), thus achieving successful classification.

\subsection*{Scaling up the approach}

In the previous sections we have demonstrated how ABS photonic platforms can be employed in the context of quantum machine learning and, more precisely, of kernel-based methods. Let us now discuss how this approach scales when increasing the dimension of the Hilbert space, i.e. when increasing the number of modes and the number of photons.

Given that brute-force classical simulation algorithms can simulate the proof-of-concept quantum experiments performed in this work \cite{Chabaud2021quantummachine}, it is natural to investigate how the ABS approach would perform when scaled up to the quantum computational advantage regime. Note that, while we considered quantum feature maps that encode classical data into single-qubit and single-qutrit quantum states as a proof-of-concept demonstration, this is not a limitation of the ABS paradigm. Indeed, allowing for additional unmeasured modes at the output of the ABS provides a natural encoding of classical data into larger Hilbert spaces composed of exponentially higher numbers of qubits or qudits dimensions. Moreover, the ABS quantum subroutine of kernel estimation becomes $\textsf{BQP}$-complete in the regime of many adaptive measurements (see subsection "Complexity-theoretic foundations of quantum machine learning with adaptive Boson Sampling" in Methods). 
{Concerning experimental realizations, imperfections like photon losses and partial distinguishability can undermine the complexity of the problem. Specifically, in an ABS regime with a constant number $k$ of adaptive measurements, the impact of losses will be similar to the case of Boson sampling. 
{We leave to future work the investigation of the possibility of using the ABS paradigm also allowing for some photon loss, that is, keeping as useful events also those when a fraction of the photons are not detected. This approach follows a similar approach to the one pursued when studying the complexity of standard Boson Sampling in the presence of photon loss \cite{Aronson_losses}.}
Furthermore, since ABS substantially relies on multiphoton interference effects, similarly to its standard version, one can rely on previous results known in the literature to identify the regimes in which the ABS framework remains intractable for a classical computer \cite{Oszmaniec2018,GarcaPatrn2019,Oh2021,Oh2022,oh2023,leverrier2015}.} 
Hence, it can be expected that as ABS devices are scaled up to more complex instances, they will be capable of solving problems that are intractable for classical computers, as long as the kernels can be efficiently estimated in the quantum regime.

Note that in the case of output states over many modes, the tomographic method used in our experiments to characterize the output quantum states (and estimate the quantum kernel matrix element in particular) is no longer viable, but the kernel matrix can still be estimated efficiently using a different design within the same platform, increasing the depth of the linear optical computation (see Fig.\ \ref{fig:concept}{c). Such a deeper adaptive circuit provides an estimate of the kernel as a state overlap, that is, for near pure states, mathematically equivalent to the estimates through state fidelity. The advantage of this approach is to obtain a scalable strategy for kernel estimation from single-photon counts that we foresee as essential for the applications of the ABS scheme in the quantum machine learning field.} Similarly, the use of post-selection in our experiments allows for the emulation of adaptive behaviour, i.e., by selecting only the output measurements that fulfil the control photon distribution, but its probabilistic nature requires a large overhead in sample complexity. This underlines the technological necessity of genuine adaptivity for scaling up ABS devices.

Finally, our classical data encoding strategy, which involves mapping clustered classical data to measurement outcomes that are randomly sampled by the quantum scheme, can be easily generalised to larger ABS instances.

When scaling up the ABS scheme, the outcome space becomes exponentially sized, which requires
{binning the outcome space and assigning clustered classical data to binned (rather than single) outcomes. As a consequence, the classical data are effectively mapped to mixed states, namely an ensemble of pure output states each given by a single adaptive output, and kernel estimation can also be done efficiently for binned outcomes using the ABS scheme, when the number of bins is a polynomial in the number of input photons \footnote{Note: U.\ Chabaud et al, ''Photonic quantum kernel methods beyond Boson
Sampling'', \textit{in preparation}.}. {Note that while binned outcomes may lead to efficient simulation of Boson Sampling, this is only known to be the case when the number of bins is constant \cite{Seron2024}.} Moreover, as the number of adaptive measurements increases, one expects that simulating binned ABS becomes harder (see Supplementary Note 7). We leave for future work a more detailed investigation on the simulability criteria of binned ABS.}

That being said, {other data encoding strategies are available for addressing larger datasets with near-term ABS devices, such as encoding classical data into non-adaptive initial interferometer parameters ($U_0$ in Fig.\ \ref{fig:concept}b) \cite{gan2022fock}. The combination with hybrid techniques, where classical pre-processing of the input data can be used to map it into the adaptive bit-strings as we have demonstrated in this work, could further increase the complexity of the problems an ABS platform addresses.}

Another strategy for increasing the complexity and the range of applications could be the adoption of a fully quantum approach based on a variational scheme, {where the data encoding strategy remains but the output part of the circuit is supplied with a parametrized linear optical circuit. As explained more in detail in Supplementary Note 8, through this procedure the classification of labeled classical data $\bm{p}$ can performed entirely by the quantum device.}
We leave these challenges for future works.

%---------------------------------------------

%\textcolor{red}{\section*{Conclusions and outlook}}
{\section*{Discussion}}

We have demonstrated experimentally a new approach to photonic quantum computations beyond Boson Sampling, where the inclusion of adaptive measurements enlarges the complexity of the output distribution. {Indeed, ABS has two key features: a) evolution conditioned
on the measurement of a subset of the evolved photons: this leads to a non-linear evolution of the remaining output
photons, b) feedforward conditioned on the measurement outcomes. In our experiment we have demonstrated feature a), and emulated feature b) by performing a post-processing
of the output of many linear optical instances carried out with different unitaries.} %The ongoing research development
%on optical delay lines and fast optical modulators would allow for a complete demonstration of adaptive evolution
%within a short time. } 

Notably, we employed an integrated photonic platform with up to eight modes combined with quantum dot sources and superconducting detectors, processing up to three single photons, to obtain a quantum kernel matrix from the output states. Such a matrix is the primer for classical optimisation routines defining the solution of data classification problems. As a proof-of-concept demonstration, we applied the quantum kernel matrix to the successful classification of non-trivial 1D and 2D datasets of points and observed in particular an improved accuracy of the quantum classification task when increasing the output Hilbert space dimension. There, the adaptivity of the {evolution} was emulated using post-selection, which may be improved by the use of faster phase shifters together with fiber delay lines to enable active modulation of the setup.
It is worth noting that the ABS scheme has the potential to go beyond quantum kernel methods {(see Supplementary Note 7)}, which may display a limited advantage over classical strategies \cite{Huang2021, Jerbi2023}.
 
While Boson Sampling has been deeply investigated, the adaptive variant addressed in this work opens a new path toward the near-term application of current state-of-the-art photonic platforms. 
In particular, the inclusion of adaptive measurements in non-universal models like Boson Sampling has been proven to lead to universal quantum computing \cite{Knill2001,AA}. {As a future investigation, an interesting direction would be to address the extension of such adaptive schemes for quantum machine learning in the Gaussian Boson Sampling framework, where recent experiments have reported large scale implementations \cite{Zhong_GBS_supremacy, Madsen2022}.} As such, we expect ABS devices to be applicable to other optimization tasks and, given the rapid development of integrated photonics and near-deterministic single-photon sources, we foresee ABS being applied more generally to other types of large-scale problems in the near future.

%---------------------------------------------

\section*{Methods}

\subsection*{Complexity-theoretic foundations of quantum machine learning with adaptive Boson Sampling}

In this section, we give complexity-theoretic foundations for the classical hardness of the computational task of quantum kernel estimation performed by ABS devices, when those are scaled to larger instances.

In particular, we show that the computational subroutines performed by ABS devices are generic instances of the problem of estimating the overlap of quantum states at the output of two quantum computations up to inverse-polynomial additive precision, defined formally as follows: 

\begin{prob}[{\small\textsf{APPROX-QCIRCUIT-OVER}}]
    Given a description of two quantum circuits $C$ and $C'$ acting on $n$ qubits with $m$ gates and on $n'$ qubits with $m'$ gates, respectively, where $m,n',m'$ are polynomials in $n$ with $n'\ge n$, and each gate acts on one or two qubits, and two numbers $a,b\in[0,1]$ with $b-a>1/\mathrm{poly}(n)$, distinguish between the following two cases: the overlap $\mathrm{Tr}[(|\psi\rangle\langle\psi|\otimes\mathbb I_{n'-n})|\psi'\rangle\langle\psi'|]$ is greater than $b$, or smaller than $a$, where we have defined $|\psi\rangle:=C|0\rangle^{\otimes n}$ and $|\psi'\rangle:=C'|0\rangle^{\otimes n'}$.
\end{prob}

This computational task captures the power of quantum computers:

\begin{lemma}\label{lem:kern-est-complete}
    {\small\textsf{APPROX-QCIRCUIT-OVER}} is $\textsf{BQP}$-complete.
\end{lemma}

\begin{proof}
The following probability estimation problem is a canonical $\textsf{BQP}$-complete problem \cite{Zhang2012}\footnote{Actually, both problems are formally \textsf{PromiseBQP}-complete.}:
\begin{prob}[{\small\textsf{APPROX-QCIRCUIT-PROB}}]
    Given a description of a quantum circuit $Q$ acting on $q$ qubits with $\mathrm{poly}(q)$ gates, where each gate acts on one or two qubits, and two numbers $\alpha,\beta\in[0,1]$ with $\beta-\alpha>1/\mathrm{poly}(q)$, distinguish between the following two cases: measuring the first qubit of the state $Q|0\rangle^{\otimes q}$ yields $1$ with probability $\ge\beta$ or $\le\alpha$.
\end{prob}

Now, note that {\small\textsf{APPROX-QCIRCUIT-OVER}} is at least as hard as {\small\textsf{APPROX-QCIRCUIT-PROB}} because any instance $(q,Q,\alpha,\beta)$ of the latter is an instance $(n,n',C,C',a,b)$ of the former with $n=1$, $C=\mathbb I_1$, $n'=q$, $C'=Q$, and $(a,b)=(\alpha,\beta)$. 

Moreover, we can use the SWAP test \cite{buhrman2001quantum}, in which the output probability of the outcome $1$ when comparing the states $|\phi\rangle$ and $|\psi\rangle$ is given by $\frac12-\frac12|\langle\phi|\psi\rangle|^2$, to show that any instance $(n,n',C,C',a,b)$ of {\small\textsf{APPROX-QCIRCUIT-OVER}} can be converted efficiently to an instance $(q,Q,\alpha,\beta)$ of {\small\textsf{APPROX-QCIRCUIT-PROB}} with $q=2n'+1$, $Q=(H\otimes\mathbb I_{2n'})\mathrm{cSWAP}(H\otimes C\otimes\mathbb I_{n'-n}\otimes C')|0\rangle^{\otimes(2n'+1)}$ (where the controlled-SWAP gate swaps the qubits $1+k$ and $1+k+n'$ for $k=1\dots n'$ when the first qubit is $|1\rangle$), and $(\alpha,\beta)=(\frac12-\frac12b,\frac12-\frac12a)$.

This shows that both problems are computationally equivalent and completes the proof.
\end{proof}

The computational power of ABS interferometers over $m$ modes, in the regime of $\mathrm{poly}(m)$ adaptive measurements, is given by the computational power of linear optical computations using single-photon states and vacuum states in input, together with photon-number measurements and feed-forward, which in turn is captured by the complexity class \textsf{BosonPadap} \cite{AA}. Crucially, the Knill--Laflamme--Milburn scheme for universal quantum computing based on dual-rail encoding shows that $\textsf{BosonPadap}=\textsf{BQP}$ \cite{Knill2001,AA}, i.e., any quantum circuit acting on $q$ qubits may be simulated by an ABS interferometer over $m=\mathrm{poly}(q)$ modes, with $n=\mathrm{poly}(q)$ input single photons and $k=\mathrm{poly}(q)$ adaptive measurements. With Lemma \ref{lem:kern-est-complete}, this implies that quantum kernel estimation based on ABS interferometers is a \textsf{BQP}-complete problem.
In other words, quantum kernel estimation using ABS is hard for classical computers unless $\textsf{BQP}=\textsf{BPP}$. 

While this makes the existence of an efficient classical algorithm for ABS quantum kernel estimation unlikely, we emphasize that it does not rule out that the learning task solved using quantum kernel estimation might be efficiently solved by another classical algorithm bypassing the need for kernel estimation.

%-------------------

\section*{Data availability}
The data that support the findings of this study are available from the corresponding author upon request.

\section*{Code availability}
The custom codes for this study that support the findings are available from the corresponding authors upon request.

%---------------------

%\bibliography{main.bib}

\begin{thebibliography}{64}%
\makeatletter
\providecommand \@ifxundefined [1]{%
 \@ifx{#1\undefined}
}%
\providecommand \@ifnum [1]{%
 \ifnum #1\expandafter \@firstoftwo
 \else \expandafter \@secondoftwo
 \fi
}%
\providecommand \@ifx [1]{%
 \ifx #1\expandafter \@firstoftwo
 \else \expandafter \@secondoftwo
 \fi
}%
\providecommand \natexlab [1]{#1}%
\providecommand \enquote  [1]{``#1''}%
\providecommand \bibnamefont  [1]{#1}%
\providecommand \bibfnamefont [1]{#1}%
\providecommand \citenamefont [1]{#1}%
\providecommand \href@noop [0]{\@secondoftwo}%
\providecommand \href [0]{\begingroup \@sanitize@url \@href}%
\providecommand \@href[1]{\@@startlink{#1}\@@href}%
\providecommand \@@href[1]{\endgroup#1\@@endlink}%
\providecommand \@sanitize@url [0]{\catcode `\\12\catcode `\$12\catcode `\&12\catcode `\#12\catcode `\^12\catcode `\_12\catcode `\%12\relax}%
\providecommand \@@startlink[1]{}%
\providecommand \@@endlink[0]{}%
\providecommand \url  [0]{\begingroup\@sanitize@url \@url }%
\providecommand \@url [1]{\endgroup\@href {#1}{\urlprefix }}%
\providecommand \urlprefix  [0]{URL }%
\providecommand \Eprint [0]{\href }%
\providecommand \doibase [0]{http://dx.doi.org/}%
\providecommand \selectlanguage [0]{\@gobble}%
\providecommand \bibinfo  [0]{\@secondoftwo}%
\providecommand \bibfield  [0]{\@secondoftwo}%
\providecommand \translation [1]{[#1]}%
\providecommand \BibitemOpen [0]{}%
\providecommand \bibitemStop [0]{}%
\providecommand \bibitemNoStop [0]{.\EOS\space}%
\providecommand \EOS [0]{\spacefactor3000\relax}%
\providecommand \BibitemShut  [1]{\csname bibitem#1\endcsname}%
\let\auto@bib@innerbib\@empty
%</preamble>
\bibitem [{\citenamefont {Flamini}\ \emph {et~al.}(2018)\citenamefont {Flamini}, \citenamefont {Spagnolo},\ and\ \citenamefont {Sciarrino}}]{Flamini_rev}%
  \BibitemOpen
  \bibfield  {author} {\bibinfo {author} {\bibfnamefont {F.}~\bibnamefont {Flamini}}, \bibinfo {author} {\bibfnamefont {N.}~\bibnamefont {Spagnolo}}, \ and\ \bibinfo {author} {\bibfnamefont {F.}~\bibnamefont {Sciarrino}},\ }\bibfield  {title} {\enquote {\bibinfo {title} {Photonic quantum information processing: a review},}\ }\href {\doibase 10.1088/1361-6633/aad5b2} {\bibfield  {journal} {\bibinfo  {journal} {Rep. Prog. Phys.}\ }\textbf {\bibinfo {volume} {82}},\ \bibinfo {pages} {016001} (\bibinfo {year} {2018})}\BibitemShut {NoStop}%
\bibitem [{\citenamefont {Volz}\ \emph {et~al.}(2014)\citenamefont {Volz}, \citenamefont {Scheucher}, \citenamefont {Junge},\ and\ \citenamefont {Rauschenbeutel}}]{Volz2014}%
  \BibitemOpen
  \bibfield  {author} {\bibinfo {author} {\bibfnamefont {J.}~\bibnamefont {Volz}}, \bibinfo {author} {\bibfnamefont {M.}~\bibnamefont {Scheucher}}, \bibinfo {author} {\bibfnamefont {C.}~\bibnamefont {Junge}}, \ and\ \bibinfo {author} {\bibfnamefont {A.}~\bibnamefont {Rauschenbeutel}},\ }\bibfield  {title} {\enquote {\bibinfo {title} {Nonlinear $\pi$ phase shift for single fibre-guided photons interacting with a single resonator-enhanced atom},}\ }\href {\doibase 10.1038/nphoton.2014.253} {\bibfield  {journal} {\bibinfo  {journal} {Nat. Photonics}\ }\textbf {\bibinfo {volume} {8}},\ \bibinfo {pages} {965--970} (\bibinfo {year} {2014})}\BibitemShut {NoStop}%
\bibitem [{\citenamefont {Feizpour}\ \emph {et~al.}(2015)\citenamefont {Feizpour}, \citenamefont {Hallaji}, \citenamefont {Dmochowski},\ and\ \citenamefont {Steinberg}}]{Feizpour2015}%
  \BibitemOpen
  \bibfield  {author} {\bibinfo {author} {\bibfnamefont {A.}~\bibnamefont {Feizpour}}, \bibinfo {author} {\bibfnamefont {M.}~\bibnamefont {Hallaji}}, \bibinfo {author} {\bibfnamefont {G.}~\bibnamefont {Dmochowski}}, \ and\ \bibinfo {author} {\bibfnamefont {A.~M.}\ \bibnamefont {Steinberg}},\ }\bibfield  {title} {\enquote {\bibinfo {title} {Observation of the nonlinear phase shift due to single post-selected photons},}\ }\href {\doibase 10.1038/nphys3433} {\bibfield  {journal} {\bibinfo  {journal} {Nat. Phys.}\ }\textbf {\bibinfo {volume} {11}},\ \bibinfo {pages} {905--909} (\bibinfo {year} {2015})}\BibitemShut {NoStop}%
\bibitem [{\citenamefont {Tiarks}\ \emph {et~al.}(2016)\citenamefont {Tiarks}, \citenamefont {Schmidt}, \citenamefont {Rempe},\ and\ \citenamefont {D\"{u}rr}}]{Tiarks2016}%
  \BibitemOpen
  \bibfield  {author} {\bibinfo {author} {\bibfnamefont {D.}~\bibnamefont {Tiarks}}, \bibinfo {author} {\bibfnamefont {S.}~\bibnamefont {Schmidt}}, \bibinfo {author} {\bibfnamefont {G.}~\bibnamefont {Rempe}}, \ and\ \bibinfo {author} {\bibfnamefont {S.}~\bibnamefont {D\"{u}rr}},\ }\bibfield  {title} {\enquote {\bibinfo {title} {Optical $\pi$ phase shift created with a single-photon pulse},}\ }\href {\doibase 10.1126/sciadv.160003} {\bibfield  {journal} {\bibinfo  {journal} {Sci. Adv.}\ }\textbf {\bibinfo {volume} {2}},\ \bibinfo {pages} {e1600036} (\bibinfo {year} {2016})}\BibitemShut {NoStop}%
\bibitem [{\citenamefont {Hacker}\ \emph {et~al.}(2016)\citenamefont {Hacker}, \citenamefont {Welte}, \citenamefont {Rempe},\ and\ \citenamefont {Ritter}}]{Hacker2016}%
  \BibitemOpen
  \bibfield  {author} {\bibinfo {author} {\bibfnamefont {B.}~\bibnamefont {Hacker}}, \bibinfo {author} {\bibfnamefont {S.}~\bibnamefont {Welte}}, \bibinfo {author} {\bibfnamefont {G.}~\bibnamefont {Rempe}}, \ and\ \bibinfo {author} {\bibfnamefont {S.}~\bibnamefont {Ritter}},\ }\bibfield  {title} {\enquote {\bibinfo {title} {A photon–photon quantum gate based on a single atom in an optical resonator},}\ }\href {\doibase 10.1038/nature18592} {\bibfield  {journal} {\bibinfo  {journal} {Nature}\ }\textbf {\bibinfo {volume} {536}},\ \bibinfo {pages} {193--196} (\bibinfo {year} {2016})}\BibitemShut {NoStop}%
\bibitem [{\citenamefont {Stolz}\ \emph {et~al.}(2022)\citenamefont {Stolz}, \citenamefont {Hegels}, \citenamefont {Winter}, \citenamefont {R\"{o}hr}, \citenamefont {Hsiao}, \citenamefont {Husel}, \citenamefont {Rempe},\ and\ \citenamefont {D\"{u}rr}}]{Stolz2022}%
  \BibitemOpen
  \bibfield  {author} {\bibinfo {author} {\bibfnamefont {T.}~\bibnamefont {Stolz}}, \bibinfo {author} {\bibfnamefont {H.}~\bibnamefont {Hegels}}, \bibinfo {author} {\bibfnamefont {M.}~\bibnamefont {Winter}}, \bibinfo {author} {\bibfnamefont {B.}~\bibnamefont {R\"{o}hr}}, \bibinfo {author} {\bibfnamefont {Y.-F.}\ \bibnamefont {Hsiao}}, \bibinfo {author} {\bibfnamefont {L.}~\bibnamefont {Husel}}, \bibinfo {author} {\bibfnamefont {G.}~\bibnamefont {Rempe}}, \ and\ \bibinfo {author} {\bibfnamefont {S.}~\bibnamefont {D\"{u}rr}},\ }\bibfield  {title} {\enquote {\bibinfo {title} {Quantum-logic gate between two optical photons with an average efficiency above 40$\%$},}\ }\href {\doibase 10.1103/PhysRevX.12.021035} {\bibfield  {journal} {\bibinfo  {journal} {Phys. Rev. X}\ }\textbf {\bibinfo {volume} {12}},\ \bibinfo {pages} {021035} (\bibinfo {year} {2022})}\BibitemShut {NoStop}%
\bibitem [{\citenamefont {Kuriakose}\ \emph {et~al.}(2022)\citenamefont {Kuriakose}, \citenamefont {Walker}, \citenamefont {Dowling}, \citenamefont {Kyriienko}, \citenamefont {Shelykh}, \citenamefont {St-Jean}, \citenamefont {Zambon}, \citenamefont {Lemaitre}, \citenamefont {Sagnes}, \citenamefont {Legratiet}, \citenamefont {Harouri}, \citenamefont {Ravets}, \citenamefont {Skolnick}, \citenamefont {Amo}, \citenamefont {Bloch},\ and\ \citenamefont {Krizhanovskii}}]{Kuriakose2022}%
  \BibitemOpen
  \bibfield  {author} {\bibinfo {author} {\bibfnamefont {T.}~\bibnamefont {Kuriakose}}, \bibinfo {author} {\bibfnamefont {P.~M.}\ \bibnamefont {Walker}}, \bibinfo {author} {\bibfnamefont {T.}~\bibnamefont {Dowling}}, \bibinfo {author} {\bibfnamefont {O.}~\bibnamefont {Kyriienko}}, \bibinfo {author} {\bibfnamefont {I.~A.}\ \bibnamefont {Shelykh}}, \bibinfo {author} {\bibfnamefont {P.}~\bibnamefont {St-Jean}}, \bibinfo {author} {\bibfnamefont {N.~C.}\ \bibnamefont {Zambon}}, \bibinfo {author} {\bibfnamefont {A.}~\bibnamefont {Lemaitre}}, \bibinfo {author} {\bibfnamefont {I.}~\bibnamefont {Sagnes}}, \bibinfo {author} {\bibfnamefont {L.}~\bibnamefont {Legratiet}}, \bibinfo {author} {\bibfnamefont {A.}~\bibnamefont {Harouri}}, \bibinfo {author} {\bibfnamefont {S.}~\bibnamefont {Ravets}}, \bibinfo {author} {\bibfnamefont {M.~S.}\ \bibnamefont {Skolnick}}, \bibinfo {author} {\bibfnamefont {A.}~\bibnamefont {Amo}}, \bibinfo {author} {\bibfnamefont {J.}~\bibnamefont {Bloch}}, \ and\ \bibinfo {author} {\bibfnamefont
  {D.}~\bibnamefont {Krizhanovskii}},\ }\bibfield  {title} {\enquote {\bibinfo {title} {Few-photon all-optical phase rotation in a quantum-well micropillar cavity},}\ }\href {\doibase 10.1038/s41566-022-01019-6} {\bibfield  {journal} {\bibinfo  {journal} {Nat. Photonics}\ }\textbf {\bibinfo {volume} {16}},\ \bibinfo {pages} {566--569} (\bibinfo {year} {2022})}\BibitemShut {NoStop}%
\bibitem [{\citenamefont {{De Santis}}\ \emph {et~al.}(2017)\citenamefont {{De Santis}}, \citenamefont {Anton}, \citenamefont {Reznychenko}, \citenamefont {Somaschi}, \citenamefont {Coppola}, \citenamefont {Senellart}, \citenamefont {Gomez}, \citenamefont {Lemaitre}, \citenamefont {Sagnes}, \citenamefont {White}, \citenamefont {Lanco}, \citenamefont {Auffeves},\ and\ \citenamefont {Senellart}}]{DeSantis17}%
  \BibitemOpen
  \bibfield  {author} {\bibinfo {author} {\bibfnamefont {L.}~\bibnamefont {{De Santis}}}, \bibinfo {author} {\bibfnamefont {C.}~\bibnamefont {Anton}}, \bibinfo {author} {\bibfnamefont {B.}~\bibnamefont {Reznychenko}}, \bibinfo {author} {\bibfnamefont {N.}~\bibnamefont {Somaschi}}, \bibinfo {author} {\bibfnamefont {G.}~\bibnamefont {Coppola}}, \bibinfo {author} {\bibfnamefont {J.}~\bibnamefont {Senellart}}, \bibinfo {author} {\bibfnamefont {C.}~\bibnamefont {Gomez}}, \bibinfo {author} {\bibfnamefont {A.}~\bibnamefont {Lemaitre}}, \bibinfo {author} {\bibfnamefont {I.}~\bibnamefont {Sagnes}}, \bibinfo {author} {\bibfnamefont {A.~G.}\ \bibnamefont {White}}, \bibinfo {author} {\bibfnamefont {L.}~\bibnamefont {Lanco}}, \bibinfo {author} {\bibfnamefont {A.}~\bibnamefont {Auffeves}}, \ and\ \bibinfo {author} {\bibfnamefont {P.}~\bibnamefont {Senellart}},\ }\bibfield  {title} {\enquote {\bibinfo {title} {A solid-state single-photon filter},}\ }\href {\doibase 10.1038/nnano.2017.85} {\bibfield  {journal} {\bibinfo
  {journal} {Nat. Nanotechnol.}\ }\textbf {\bibinfo {volume} {12}},\ \bibinfo {pages} {663--667} (\bibinfo {year} {2017})}\BibitemShut {NoStop}%
\bibitem [{\citenamefont {Staunstrup}\ \emph {et~al.}(2024)\citenamefont {Staunstrup}, \citenamefont {Tiranov}, \citenamefont {Wang}, \citenamefont {Scholz}, \citenamefont {Wieck}, \citenamefont {Ludwig}, \citenamefont {Midolo}, \citenamefont {Rotenberg}, \citenamefont {Lodahl},\ and\ \citenamefont {Le~Jeannic}}]{StraunstrupNL}%
  \BibitemOpen
  \bibfield  {author} {\bibinfo {author} {\bibfnamefont {Mathias J.~R.}\ \bibnamefont {Staunstrup}}, \bibinfo {author} {\bibfnamefont {Alexey}\ \bibnamefont {Tiranov}}, \bibinfo {author} {\bibfnamefont {Ying}\ \bibnamefont {Wang}}, \bibinfo {author} {\bibfnamefont {Sven}\ \bibnamefont {Scholz}}, \bibinfo {author} {\bibfnamefont {Andreas~D.}\ \bibnamefont {Wieck}}, \bibinfo {author} {\bibfnamefont {Arne}\ \bibnamefont {Ludwig}}, \bibinfo {author} {\bibfnamefont {Leonardo}\ \bibnamefont {Midolo}}, \bibinfo {author} {\bibfnamefont {Nir}\ \bibnamefont {Rotenberg}}, \bibinfo {author} {\bibfnamefont {Peter}\ \bibnamefont {Lodahl}}, \ and\ \bibinfo {author} {\bibfnamefont {Hanna}\ \bibnamefont {Le~Jeannic}},\ }\bibfield  {title} {\enquote {\bibinfo {title} {Direct observation of a few-photon phase shift induced by a single quantum emitter in a waveguide},}\ }\href {\doibase 10.1038/s41467-024-51805-9} {\bibfield  {journal} {\bibinfo  {journal} {Nature Communications}\ }\textbf {\bibinfo {volume} {15}} (\bibinfo
  {year} {2024}),\ 10.1038/s41467-024-51805-9}\BibitemShut {NoStop}%
\bibitem [{\citenamefont {Nielsen}\ and\ \citenamefont {Chuang}(2010)}]{chuang2010}%
  \BibitemOpen
  \bibfield  {author} {\bibinfo {author} {\bibfnamefont {M.~A.}\ \bibnamefont {Nielsen}}\ and\ \bibinfo {author} {\bibfnamefont {I.~L.}\ \bibnamefont {Chuang}},\ }\href@noop {} {\emph {\bibinfo {title} {Quantum Computation and Quantum Information}}}\ (\bibinfo  {publisher} {Cambridge University Press},\ \bibinfo {year} {2010})\BibitemShut {NoStop}%
\bibitem [{\citenamefont {Knill}\ \emph {et~al.}(2001)\citenamefont {Knill}, \citenamefont {Laflamme},\ and\ \citenamefont {Milburn}}]{Knill2001}%
  \BibitemOpen
  \bibfield  {author} {\bibinfo {author} {\bibfnamefont {E.}~\bibnamefont {Knill}}, \bibinfo {author} {\bibfnamefont {R.}~\bibnamefont {Laflamme}}, \ and\ \bibinfo {author} {\bibfnamefont {G.~J.}\ \bibnamefont {Milburn}},\ }\bibfield  {title} {\enquote {\bibinfo {title} {A scheme for efficient quantum computation with linear optics},}\ }\href {\doibase https://doi.org/10.1038/35051009} {\bibfield  {journal} {\bibinfo  {journal} {Nature}\ }\textbf {\bibinfo {volume} {409}},\ \bibinfo {pages} {46--52} (\bibinfo {year} {2001})}\BibitemShut {NoStop}%
\bibitem [{\citenamefont {Chang}\ \emph {et~al.}(2014)\citenamefont {Chang}, \citenamefont {Vuletić},\ and\ \citenamefont {Lukin}}]{Chang2014}%
  \BibitemOpen
  \bibfield  {author} {\bibinfo {author} {\bibfnamefont {Darrick~E.}\ \bibnamefont {Chang}}, \bibinfo {author} {\bibfnamefont {Vladan}\ \bibnamefont {Vuletić}}, \ and\ \bibinfo {author} {\bibfnamefont {Mikhail~D.}\ \bibnamefont {Lukin}},\ }\bibfield  {title} {\enquote {\bibinfo {title} {Quantum nonlinear optics — photon by photon},}\ }\href {\doibase 10.1038/nphoton.2014.192} {\bibfield  {journal} {\bibinfo  {journal} {Nature Photonics}\ }\textbf {\bibinfo {volume} {8}},\ \bibinfo {pages} {685–694} (\bibinfo {year} {2014})}\BibitemShut {NoStop}%
\bibitem [{\citenamefont {Dutt}\ \emph {et~al.}(2024)\citenamefont {Dutt}, \citenamefont {Mohanty}, \citenamefont {Gaeta},\ and\ \citenamefont {Lipson}}]{Dutt2024}%
  \BibitemOpen
  \bibfield  {author} {\bibinfo {author} {\bibfnamefont {Avik}\ \bibnamefont {Dutt}}, \bibinfo {author} {\bibfnamefont {Aseema}\ \bibnamefont {Mohanty}}, \bibinfo {author} {\bibfnamefont {Alexander~L.}\ \bibnamefont {Gaeta}}, \ and\ \bibinfo {author} {\bibfnamefont {Michal}\ \bibnamefont {Lipson}},\ }\bibfield  {title} {\enquote {\bibinfo {title} {Nonlinear and quantum photonics using integrated optical materials},}\ }\href {\doibase 10.1038/s41578-024-00668-z} {\bibfield  {journal} {\bibinfo  {journal} {Nature Reviews Materials}\ }\textbf {\bibinfo {volume} {9}},\ \bibinfo {pages} {321–346} (\bibinfo {year} {2024})}\BibitemShut {NoStop}%
\bibitem [{\citenamefont {Briegel}\ \emph {et~al.}(2009)\citenamefont {Briegel}, \citenamefont {Browne}, \citenamefont {D\"{u}r}, \citenamefont {Raussendorf},\ and\ \citenamefont {den Nest}}]{Briegel2009}%
  \BibitemOpen
  \bibfield  {author} {\bibinfo {author} {\bibfnamefont {H.~J.}\ \bibnamefont {Briegel}}, \bibinfo {author} {\bibfnamefont {D.~E.}\ \bibnamefont {Browne}}, \bibinfo {author} {\bibfnamefont {W.}~\bibnamefont {D\"{u}r}}, \bibinfo {author} {\bibfnamefont {R.}~\bibnamefont {Raussendorf}}, \ and\ \bibinfo {author} {\bibfnamefont {M.~Van}\ \bibnamefont {den Nest}},\ }\bibfield  {title} {\enquote {\bibinfo {title} {Measurement-based quantum computation},}\ }\href {\doibase 10.1038/nphys1157} {\bibfield  {journal} {\bibinfo  {journal} {Nature Physics}\ }\textbf {\bibinfo {volume} {5}},\ \bibinfo {pages} {19--26} (\bibinfo {year} {2009})}\BibitemShut {NoStop}%
\bibitem [{\citenamefont {Bartolucci}\ \emph {et~al.}(2023)\citenamefont {Bartolucci}, \citenamefont {Birchall}, \citenamefont {Bomb{\'{\i}}n}, \citenamefont {Cable}, \citenamefont {Dawson}, \citenamefont {Gimeno-Segovia}, \citenamefont {Johnston}, \citenamefont {Kieling}, \citenamefont {Nickerson}, \citenamefont {Pant}, \citenamefont {Pastawski}, \citenamefont {Rudolph},\ and\ \citenamefont {Sparrow}}]{Bartolucci2023}%
  \BibitemOpen
  \bibfield  {author} {\bibinfo {author} {\bibfnamefont {Sara}\ \bibnamefont {Bartolucci}}, \bibinfo {author} {\bibfnamefont {Patrick}\ \bibnamefont {Birchall}}, \bibinfo {author} {\bibfnamefont {Hector}\ \bibnamefont {Bomb{\'{\i}}n}}, \bibinfo {author} {\bibfnamefont {Hugo}\ \bibnamefont {Cable}}, \bibinfo {author} {\bibfnamefont {Chris}\ \bibnamefont {Dawson}}, \bibinfo {author} {\bibfnamefont {Mercedes}\ \bibnamefont {Gimeno-Segovia}}, \bibinfo {author} {\bibfnamefont {Eric}\ \bibnamefont {Johnston}}, \bibinfo {author} {\bibfnamefont {Konrad}\ \bibnamefont {Kieling}}, \bibinfo {author} {\bibfnamefont {Naomi}\ \bibnamefont {Nickerson}}, \bibinfo {author} {\bibfnamefont {Mihir}\ \bibnamefont {Pant}}, \bibinfo {author} {\bibfnamefont {Fernando}\ \bibnamefont {Pastawski}}, \bibinfo {author} {\bibfnamefont {Terry}\ \bibnamefont {Rudolph}}, \ and\ \bibinfo {author} {\bibfnamefont {Chris}\ \bibnamefont {Sparrow}},\ }\bibfield  {title} {\enquote {\bibinfo {title} {Fusion-based quantum computation},}\ }\href
  {\doibase 10.1038/s41467-023-36493-1} {\bibfield  {journal} {\bibinfo  {journal} {Nature Communications}\ }\textbf {\bibinfo {volume} {14}} (\bibinfo {year} {2023}),\ 10.1038/s41467-023-36493-1}\BibitemShut {NoStop}%
\bibitem [{\citenamefont {Aaronson}\ and\ \citenamefont {Arkhipov}(2011)}]{AA}%
  \BibitemOpen
  \bibfield  {author} {\bibinfo {author} {\bibfnamefont {S.}~\bibnamefont {Aaronson}}\ and\ \bibinfo {author} {\bibfnamefont {A.}~\bibnamefont {Arkhipov}},\ }\bibfield  {title} {\enquote {\bibinfo {title} {The computational complexity of linear optics},}\ }in\ \href {\doibase 10.1145/1993636.1993682} {\emph {\bibinfo {booktitle} {Proceedings of the 43rd annual ACM symposium on Theory of Computing}}},\ \bibinfo {editor} {edited by\ \bibinfo {editor} {\bibfnamefont {ACM}\ \bibnamefont {Press}}}\ (\bibinfo {year} {2011})\ pp.\ \bibinfo {pages} {333--342}\BibitemShut {NoStop}%
\bibitem [{\citenamefont {Hamilton}\ \emph {et~al.}(2017)\citenamefont {Hamilton}, \citenamefont {Kruse}, \citenamefont {Sansoni}, \citenamefont {Barkhofen}, \citenamefont {Silberhorn},\ and\ \citenamefont {Jex}}]{Hamilton2017}%
  \BibitemOpen
  \bibfield  {author} {\bibinfo {author} {\bibfnamefont {Craig~S.}\ \bibnamefont {Hamilton}}, \bibinfo {author} {\bibfnamefont {Regina}\ \bibnamefont {Kruse}}, \bibinfo {author} {\bibfnamefont {Linda}\ \bibnamefont {Sansoni}}, \bibinfo {author} {\bibfnamefont {Sonja}\ \bibnamefont {Barkhofen}}, \bibinfo {author} {\bibfnamefont {Christine}\ \bibnamefont {Silberhorn}}, \ and\ \bibinfo {author} {\bibfnamefont {Igor}\ \bibnamefont {Jex}},\ }\bibfield  {title} {\enquote {\bibinfo {title} {Gaussian boson sampling},}\ }\href {\doibase 10.1103/PhysRevLett.119.170501} {\bibfield  {journal} {\bibinfo  {journal} {Phys. Rev. Lett.}\ }\textbf {\bibinfo {volume} {119}},\ \bibinfo {pages} {170501} (\bibinfo {year} {2017})}\BibitemShut {NoStop}%
\bibitem [{\citenamefont {Zhong}\ \emph {et~al.}(2020)\citenamefont {Zhong}, \citenamefont {Wang}, \citenamefont {Deng}, \citenamefont {Chen}, \citenamefont {Peng}, \citenamefont {Luo}, \citenamefont {Qin}, \citenamefont {Wu}, \citenamefont {Ding}, \citenamefont {Hu}, \citenamefont {Hu}, \citenamefont {Yang}, \citenamefont {Zhang}, \citenamefont {Li}, \citenamefont {Li}, \citenamefont {Jiang}, \citenamefont {Gan}, \citenamefont {Yang}, \citenamefont {You}, \citenamefont {Wang}, \citenamefont {Li}, \citenamefont {Liu}, \citenamefont {Lu},\ and\ \citenamefont {Pan}}]{Zhong_GBS_supremacy}%
  \BibitemOpen
  \bibfield  {author} {\bibinfo {author} {\bibfnamefont {Han-Sen}\ \bibnamefont {Zhong}}, \bibinfo {author} {\bibfnamefont {Hui}\ \bibnamefont {Wang}}, \bibinfo {author} {\bibfnamefont {Yu-Hao}\ \bibnamefont {Deng}}, \bibinfo {author} {\bibfnamefont {Ming-Cheng}\ \bibnamefont {Chen}}, \bibinfo {author} {\bibfnamefont {Li-Chao}\ \bibnamefont {Peng}}, \bibinfo {author} {\bibfnamefont {Yi-Han}\ \bibnamefont {Luo}}, \bibinfo {author} {\bibfnamefont {Jian}\ \bibnamefont {Qin}}, \bibinfo {author} {\bibfnamefont {Dian}\ \bibnamefont {Wu}}, \bibinfo {author} {\bibfnamefont {Xing}\ \bibnamefont {Ding}}, \bibinfo {author} {\bibfnamefont {Yi}~\bibnamefont {Hu}}, \bibinfo {author} {\bibfnamefont {Peng}\ \bibnamefont {Hu}}, \bibinfo {author} {\bibfnamefont {Xiao-Yan}\ \bibnamefont {Yang}}, \bibinfo {author} {\bibfnamefont {Wei-Jun}\ \bibnamefont {Zhang}}, \bibinfo {author} {\bibfnamefont {Hao}\ \bibnamefont {Li}}, \bibinfo {author} {\bibfnamefont {Yuxuan}\ \bibnamefont {Li}}, \bibinfo {author} {\bibfnamefont {Xiao}\
  \bibnamefont {Jiang}}, \bibinfo {author} {\bibfnamefont {Lin}\ \bibnamefont {Gan}}, \bibinfo {author} {\bibfnamefont {Guangwen}\ \bibnamefont {Yang}}, \bibinfo {author} {\bibfnamefont {Lixing}\ \bibnamefont {You}}, \bibinfo {author} {\bibfnamefont {Zhen}\ \bibnamefont {Wang}}, \bibinfo {author} {\bibfnamefont {Li}~\bibnamefont {Li}}, \bibinfo {author} {\bibfnamefont {Nai-Le}\ \bibnamefont {Liu}}, \bibinfo {author} {\bibfnamefont {Chao-Yang}\ \bibnamefont {Lu}}, \ and\ \bibinfo {author} {\bibfnamefont {Jian-Wei}\ \bibnamefont {Pan}},\ }\bibfield  {title} {\enquote {\bibinfo {title} {Quantum computational advantage using photons},}\ }\href {\doibase 10.1126/science.abe8770} {\bibfield  {journal} {\bibinfo  {journal} {Science}\ }\textbf {\bibinfo {volume} {370}},\ \bibinfo {pages} {1460--1463} (\bibinfo {year} {2020})}\BibitemShut {NoStop}%
\bibitem [{\citenamefont {Zhong}\ \emph {et~al.}(2021)\citenamefont {Zhong}, \citenamefont {Deng}, \citenamefont {Qin}, \citenamefont {Wang}, \citenamefont {Chen}, \citenamefont {Peng}, \citenamefont {Luo}, \citenamefont {Wu}, \citenamefont {Gong}, \citenamefont {Su}, \citenamefont {Hu}, \citenamefont {Hu}, \citenamefont {Yang}, \citenamefont {Zhang}, \citenamefont {Li}, \citenamefont {Li}, \citenamefont {Jiang}, \citenamefont {Gan}, \citenamefont {Yang}, \citenamefont {You}, \citenamefont {Wang}, \citenamefont {Li}, \citenamefont {Liu}, \citenamefont {Renema}, \citenamefont {Lu},\ and\ \citenamefont {Pan}}]{zhong2021phaseprogrammable}%
  \BibitemOpen
  \bibfield  {author} {\bibinfo {author} {\bibfnamefont {Han-Sen}\ \bibnamefont {Zhong}}, \bibinfo {author} {\bibfnamefont {Yu-Hao}\ \bibnamefont {Deng}}, \bibinfo {author} {\bibfnamefont {Jian}\ \bibnamefont {Qin}}, \bibinfo {author} {\bibfnamefont {Hui}\ \bibnamefont {Wang}}, \bibinfo {author} {\bibfnamefont {Ming-Cheng}\ \bibnamefont {Chen}}, \bibinfo {author} {\bibfnamefont {Li-Chao}\ \bibnamefont {Peng}}, \bibinfo {author} {\bibfnamefont {Yi-Han}\ \bibnamefont {Luo}}, \bibinfo {author} {\bibfnamefont {Dian}\ \bibnamefont {Wu}}, \bibinfo {author} {\bibfnamefont {Si-Qiu}\ \bibnamefont {Gong}}, \bibinfo {author} {\bibfnamefont {Hao}\ \bibnamefont {Su}}, \bibinfo {author} {\bibfnamefont {Yi}~\bibnamefont {Hu}}, \bibinfo {author} {\bibfnamefont {Peng}\ \bibnamefont {Hu}}, \bibinfo {author} {\bibfnamefont {Xiao-Yan}\ \bibnamefont {Yang}}, \bibinfo {author} {\bibfnamefont {Wei-Jun}\ \bibnamefont {Zhang}}, \bibinfo {author} {\bibfnamefont {Hao}\ \bibnamefont {Li}}, \bibinfo {author} {\bibfnamefont {Yuxuan}\
  \bibnamefont {Li}}, \bibinfo {author} {\bibfnamefont {Xiao}\ \bibnamefont {Jiang}}, \bibinfo {author} {\bibfnamefont {Lin}\ \bibnamefont {Gan}}, \bibinfo {author} {\bibfnamefont {Guangwen}\ \bibnamefont {Yang}}, \bibinfo {author} {\bibfnamefont {Lixing}\ \bibnamefont {You}}, \bibinfo {author} {\bibfnamefont {Zhen}\ \bibnamefont {Wang}}, \bibinfo {author} {\bibfnamefont {Li}~\bibnamefont {Li}}, \bibinfo {author} {\bibfnamefont {Nai-Le}\ \bibnamefont {Liu}}, \bibinfo {author} {\bibfnamefont {Jelmer~J.}\ \bibnamefont {Renema}}, \bibinfo {author} {\bibfnamefont {Chao-Yang}\ \bibnamefont {Lu}}, \ and\ \bibinfo {author} {\bibfnamefont {Jian-Wei}\ \bibnamefont {Pan}},\ }\bibfield  {title} {\enquote {\bibinfo {title} {Phase-programmable gaussian boson sampling using stimulated squeezed light},}\ }\href {\doibase 10.1103/PhysRevLett.127.180502} {\bibfield  {journal} {\bibinfo  {journal} {Phys. Rev. Lett.}\ }\textbf {\bibinfo {volume} {127}},\ \bibinfo {pages} {180502} (\bibinfo {year} {2021})}\BibitemShut {NoStop}%
\bibitem [{\citenamefont {Madsen}\ \emph {et~al.}(2022)\citenamefont {Madsen}, \citenamefont {Laudenbach}, \citenamefont {Askarani}, \citenamefont {Rortais}, \citenamefont {Vincent}, \citenamefont {Bulmer}, \citenamefont {Miatto}, \citenamefont {Neuhaus}, \citenamefont {Helt}, \citenamefont {Collins}, \citenamefont {Lita}, \citenamefont {Gerrits}, \citenamefont {Nam}, \citenamefont {Vaidya}, \citenamefont {Menotti}, \citenamefont {Dhand}, \citenamefont {Vernon}, \citenamefont {Quesada},\ and\ \citenamefont {Lavoie}}]{Madsen2022}%
  \BibitemOpen
  \bibfield  {author} {\bibinfo {author} {\bibfnamefont {Lars~S.}\ \bibnamefont {Madsen}}, \bibinfo {author} {\bibfnamefont {Fabian}\ \bibnamefont {Laudenbach}}, \bibinfo {author} {\bibfnamefont {Mohsen~Falamarzi.}\ \bibnamefont {Askarani}}, \bibinfo {author} {\bibfnamefont {Fabien}\ \bibnamefont {Rortais}}, \bibinfo {author} {\bibfnamefont {Trevor}\ \bibnamefont {Vincent}}, \bibinfo {author} {\bibfnamefont {Jacob F.~F.}\ \bibnamefont {Bulmer}}, \bibinfo {author} {\bibfnamefont {Filippo~M.}\ \bibnamefont {Miatto}}, \bibinfo {author} {\bibfnamefont {Leonhard}\ \bibnamefont {Neuhaus}}, \bibinfo {author} {\bibfnamefont {Lukas~G.}\ \bibnamefont {Helt}}, \bibinfo {author} {\bibfnamefont {Matthew~J.}\ \bibnamefont {Collins}}, \bibinfo {author} {\bibfnamefont {Adriana~E.}\ \bibnamefont {Lita}}, \bibinfo {author} {\bibfnamefont {Thomas}\ \bibnamefont {Gerrits}}, \bibinfo {author} {\bibfnamefont {Sae~Woo}\ \bibnamefont {Nam}}, \bibinfo {author} {\bibfnamefont {Varun~D.}\ \bibnamefont {Vaidya}}, \bibinfo {author}
  {\bibfnamefont {Matteo}\ \bibnamefont {Menotti}}, \bibinfo {author} {\bibfnamefont {Ish}\ \bibnamefont {Dhand}}, \bibinfo {author} {\bibfnamefont {Zachary}\ \bibnamefont {Vernon}}, \bibinfo {author} {\bibfnamefont {Nicol{\'a}s}\ \bibnamefont {Quesada}}, \ and\ \bibinfo {author} {\bibfnamefont {Jonathan}\ \bibnamefont {Lavoie}},\ }\bibfield  {title} {\enquote {\bibinfo {title} {Quantum computational advantage with a programmable photonic processor},}\ }\href {\doibase 10.1038/s41586-022-04725-x} {\bibfield  {journal} {\bibinfo  {journal} {Nature}\ }\textbf {\bibinfo {volume} {606}},\ \bibinfo {pages} {75--81} (\bibinfo {year} {2022})}\BibitemShut {NoStop}%
\bibitem [{\citenamefont {Mart{\'{i}}nez-Cifuentes}\ \emph {et~al.}(2023)\citenamefont {Mart{\'{i}}nez-Cifuentes}, \citenamefont {Fonseca-Romero},\ and\ \citenamefont {Quesada}}]{MartinezCifuentes2023classicalmodelsmay}%
  \BibitemOpen
  \bibfield  {author} {\bibinfo {author} {\bibfnamefont {Javier}\ \bibnamefont {Mart{\'{i}}nez-Cifuentes}}, \bibinfo {author} {\bibfnamefont {K.~M.}\ \bibnamefont {Fonseca-Romero}}, \ and\ \bibinfo {author} {\bibfnamefont {Nicol{\'{a}}s}\ \bibnamefont {Quesada}},\ }\bibfield  {title} {\enquote {\bibinfo {title} {Classical models may be a better explanation of the {J}iuzhang 1.0 {G}aussian {B}oson {S}ampler than its targeted squeezed light model},}\ }\href {\doibase 10.22331/q-2023-08-08-1076} {\bibfield  {journal} {\bibinfo  {journal} {{Quantum}}\ }\textbf {\bibinfo {volume} {7}},\ \bibinfo {pages} {1076} (\bibinfo {year} {2023})}\BibitemShut {NoStop}%
\bibitem [{\citenamefont {Oh}\ \emph {et~al.}(2023{\natexlab{a}})\citenamefont {Oh}, \citenamefont {Liu}, \citenamefont {Alexeev}, \citenamefont {Fefferman},\ and\ \citenamefont {Jiang}}]{TensorNetworkGBS}%
  \BibitemOpen
  \bibfield  {author} {\bibinfo {author} {\bibfnamefont {Changhun}\ \bibnamefont {Oh}}, \bibinfo {author} {\bibfnamefont {Minzhao}\ \bibnamefont {Liu}}, \bibinfo {author} {\bibfnamefont {Yuri}\ \bibnamefont {Alexeev}}, \bibinfo {author} {\bibfnamefont {Bill}\ \bibnamefont {Fefferman}}, \ and\ \bibinfo {author} {\bibfnamefont {Liang}\ \bibnamefont {Jiang}},\ }\href@noop {} {\enquote {\bibinfo {title} {Tensor network algorithm for simulating experimental gaussian boson sampling},}\ } (\bibinfo {year} {2023}{\natexlab{a}}),\ \Eprint {http://arxiv.org/abs/2306.03709} {arXiv:2306.03709 [quant-ph]} \BibitemShut {NoStop}%
\bibitem [{\citenamefont {Stanev}\ \emph {et~al.}(2023{\natexlab{a}})\citenamefont {Stanev}, \citenamefont {Giordani}, \citenamefont {Spagnolo},\ and\ \citenamefont {Sciarrino}}]{stanev2023testing}%
  \BibitemOpen
  \bibfield  {author} {\bibinfo {author} {\bibfnamefont {Denis}\ \bibnamefont {Stanev}}, \bibinfo {author} {\bibfnamefont {Taira}\ \bibnamefont {Giordani}}, \bibinfo {author} {\bibfnamefont {Nicolò}\ \bibnamefont {Spagnolo}}, \ and\ \bibinfo {author} {\bibfnamefont {Fabio}\ \bibnamefont {Sciarrino}},\ }\href@noop {} {\enquote {\bibinfo {title} {Testing of on-cloud gaussian boson sampler "borealis'' via graph theory},}\ } (\bibinfo {year} {2023}{\natexlab{a}}),\ \Eprint {http://arxiv.org/abs/2306.12120} {arXiv:2306.12120 [quant-ph]} \BibitemShut {NoStop}%
\bibitem [{\citenamefont {Arrazola}\ and\ \citenamefont {Bromley}(2018)}]{GBSGraphTheory1}%
  \BibitemOpen
  \bibfield  {author} {\bibinfo {author} {\bibfnamefont {Juan~Miguel}\ \bibnamefont {Arrazola}}\ and\ \bibinfo {author} {\bibfnamefont {Thomas~R.}\ \bibnamefont {Bromley}},\ }\bibfield  {title} {\enquote {\bibinfo {title} {Using gaussian boson sampling to find dense subgraphs},}\ }\href {\doibase 10.1103/physrevlett.121.030503} {\bibfield  {journal} {\bibinfo  {journal} {Physical Review Letters}\ }\textbf {\bibinfo {volume} {121}},\ \bibinfo {pages} {030503} (\bibinfo {year} {2018})}\BibitemShut {NoStop}%
\bibitem [{\citenamefont {Schuld}\ \emph {et~al.}(2020)\citenamefont {Schuld}, \citenamefont {Br{\'{a}}dler}, \citenamefont {Israel}, \citenamefont {Su},\ and\ \citenamefont {Gupt}}]{GBSGraphTheory2}%
  \BibitemOpen
  \bibfield  {author} {\bibinfo {author} {\bibfnamefont {Maria}\ \bibnamefont {Schuld}}, \bibinfo {author} {\bibfnamefont {Kamil}\ \bibnamefont {Br{\'{a}}dler}}, \bibinfo {author} {\bibfnamefont {Robert}\ \bibnamefont {Israel}}, \bibinfo {author} {\bibfnamefont {Daiqin}\ \bibnamefont {Su}}, \ and\ \bibinfo {author} {\bibfnamefont {Brajesh}\ \bibnamefont {Gupt}},\ }\bibfield  {title} {\enquote {\bibinfo {title} {Measuring the similarity of graphs with a gaussian boson sampler},}\ }\href {\doibase 10.1103/physreva.101.032314} {\bibfield  {journal} {\bibinfo  {journal} {Physical Review A}\ }\textbf {\bibinfo {volume} {101}},\ \bibinfo {pages} {032314} (\bibinfo {year} {2020})}\BibitemShut {NoStop}%
\bibitem [{\citenamefont {Banchi}\ \emph {et~al.}(2020)\citenamefont {Banchi}, \citenamefont {Fingerhuth}, \citenamefont {Babej}, \citenamefont {Ing},\ and\ \citenamefont {Arrazola}}]{Banchi_vibronic}%
  \BibitemOpen
  \bibfield  {author} {\bibinfo {author} {\bibfnamefont {Leonardo}\ \bibnamefont {Banchi}}, \bibinfo {author} {\bibfnamefont {Mark}\ \bibnamefont {Fingerhuth}}, \bibinfo {author} {\bibfnamefont {Tomas}\ \bibnamefont {Babej}}, \bibinfo {author} {\bibfnamefont {Christopher}\ \bibnamefont {Ing}}, \ and\ \bibinfo {author} {\bibfnamefont {Juan~Miguel}\ \bibnamefont {Arrazola}},\ }\bibfield  {title} {\enquote {\bibinfo {title} {Molecular docking with gaussian boson sampling},}\ }\href {\doibase 10.1126/sciadv.aax1950} {\bibfield  {journal} {\bibinfo  {journal} {Science Advances}\ }\textbf {\bibinfo {volume} {6}} (\bibinfo {year} {2020}),\ 10.1126/sciadv.aax1950}\BibitemShut {NoStop}%
\bibitem [{\citenamefont {Jahangiri}\ \emph {et~al.}(2020)\citenamefont {Jahangiri}, \citenamefont {Arrazola}, \citenamefont {Quesada},\ and\ \citenamefont {Killoran}}]{Jahangiri_point_process}%
  \BibitemOpen
  \bibfield  {author} {\bibinfo {author} {\bibfnamefont {Soran}\ \bibnamefont {Jahangiri}}, \bibinfo {author} {\bibfnamefont {Juan~Miguel}\ \bibnamefont {Arrazola}}, \bibinfo {author} {\bibfnamefont {Nicol\'as}\ \bibnamefont {Quesada}}, \ and\ \bibinfo {author} {\bibfnamefont {Nathan}\ \bibnamefont {Killoran}},\ }\bibfield  {title} {\enquote {\bibinfo {title} {Point processes with gaussian boson sampling},}\ }\href {\doibase 10.1103/PhysRevE.101.022134} {\bibfield  {journal} {\bibinfo  {journal} {Phys. Rev. E}\ }\textbf {\bibinfo {volume} {101}},\ \bibinfo {pages} {022134} (\bibinfo {year} {2020})}\BibitemShut {NoStop}%
\bibitem [{\citenamefont {Spagnolo}\ \emph {et~al.}(2023)\citenamefont {Spagnolo}, \citenamefont {Brod}, \citenamefont {Galvao},\ and\ \citenamefont {Sciarrino}}]{SpagnoloNLBS}%
  \BibitemOpen
  \bibfield  {author} {\bibinfo {author} {\bibfnamefont {Nicol{\`{o}}}\ \bibnamefont {Spagnolo}}, \bibinfo {author} {\bibfnamefont {Daniel~J.}\ \bibnamefont {Brod}}, \bibinfo {author} {\bibfnamefont {Ernesto~F.}\ \bibnamefont {Galvao}}, \ and\ \bibinfo {author} {\bibfnamefont {Fabio}\ \bibnamefont {Sciarrino}},\ }\bibfield  {title} {\enquote {\bibinfo {title} {Non-linear boson sampling},}\ }\href {\doibase 10.1038/s41534-023-00676-x} {\bibfield  {journal} {\bibinfo  {journal} {npj Quantum Inf.}\ }\textbf {\bibinfo {volume} {9}},\ \bibinfo {pages} {3} (\bibinfo {year} {2023})}\BibitemShut {NoStop}%
\bibitem [{\citenamefont {Nokkala}\ \emph {et~al.}(2021)\citenamefont {Nokkala}, \citenamefont {Mart{\'{\i}}nez-Pe{\~{n}}a}, \citenamefont {Giorgi}, \citenamefont {Parigi}, \citenamefont {Soriano},\ and\ \citenamefont {Zambrini}}]{Nokkala2021}%
  \BibitemOpen
  \bibfield  {author} {\bibinfo {author} {\bibfnamefont {Johannes}\ \bibnamefont {Nokkala}}, \bibinfo {author} {\bibfnamefont {Rodrigo}\ \bibnamefont {Mart{\'{\i}}nez-Pe{\~{n}}a}}, \bibinfo {author} {\bibfnamefont {Gian~Luca}\ \bibnamefont {Giorgi}}, \bibinfo {author} {\bibfnamefont {Valentina}\ \bibnamefont {Parigi}}, \bibinfo {author} {\bibfnamefont {Miguel~C.}\ \bibnamefont {Soriano}}, \ and\ \bibinfo {author} {\bibfnamefont {Roberta}\ \bibnamefont {Zambrini}},\ }\bibfield  {title} {\enquote {\bibinfo {title} {Gaussian states of continuous-variable quantum systems provide universal and versatile reservoir computing},}\ }\href {\doibase 10.1038/s42005-021-00556-w} {\bibfield  {journal} {\bibinfo  {journal} {Communications Physics}\ }\textbf {\bibinfo {volume} {4}} (\bibinfo {year} {2021}),\ 10.1038/s42005-021-00556-w}\BibitemShut {NoStop}%
\bibitem [{\citenamefont {Innocenti}\ \emph {et~al.}(2023)\citenamefont {Innocenti}, \citenamefont {Lorenzo}, \citenamefont {Palmisano}, \citenamefont {Ferraro}, \citenamefont {Paternostro},\ and\ \citenamefont {Palma}}]{Innocenti2023}%
  \BibitemOpen
  \bibfield  {author} {\bibinfo {author} {\bibfnamefont {L.}~\bibnamefont {Innocenti}}, \bibinfo {author} {\bibfnamefont {S.}~\bibnamefont {Lorenzo}}, \bibinfo {author} {\bibfnamefont {I.}~\bibnamefont {Palmisano}}, \bibinfo {author} {\bibfnamefont {A.}~\bibnamefont {Ferraro}}, \bibinfo {author} {\bibfnamefont {M.}~\bibnamefont {Paternostro}}, \ and\ \bibinfo {author} {\bibfnamefont {G.~M.}\ \bibnamefont {Palma}},\ }\bibfield  {title} {\enquote {\bibinfo {title} {Potential and limitations of quantum extreme learning machines},}\ }\href {\doibase 10.1038/s42005-023-01233-w} {\bibfield  {journal} {\bibinfo  {journal} {Communications Physics}\ }\textbf {\bibinfo {volume} {6}} (\bibinfo {year} {2023}),\ 10.1038/s42005-023-01233-w}\BibitemShut {NoStop}%
\bibitem [{\citenamefont {Garc\'{\i}a-Beni}\ \emph {et~al.}(2023)\citenamefont {Garc\'{\i}a-Beni}, \citenamefont {Giorgi}, \citenamefont {Soriano},\ and\ \citenamefont {Zambrini}}]{Zambrini_RC}%
  \BibitemOpen
  \bibfield  {author} {\bibinfo {author} {\bibfnamefont {Jorge}\ \bibnamefont {Garc\'{\i}a-Beni}}, \bibinfo {author} {\bibfnamefont {Gian~Luca}\ \bibnamefont {Giorgi}}, \bibinfo {author} {\bibfnamefont {Miguel~C.}\ \bibnamefont {Soriano}}, \ and\ \bibinfo {author} {\bibfnamefont {Roberta}\ \bibnamefont {Zambrini}},\ }\bibfield  {title} {\enquote {\bibinfo {title} {Scalable photonic platform for real-time quantum reservoir computing},}\ }\href {\doibase 10.1103/PhysRevApplied.20.014051} {\bibfield  {journal} {\bibinfo  {journal} {Phys. Rev. Appl.}\ }\textbf {\bibinfo {volume} {20}},\ \bibinfo {pages} {014051} (\bibinfo {year} {2023})}\BibitemShut {NoStop}%
\bibitem [{\citenamefont {Spagnolo}\ \emph {et~al.}(2022)\citenamefont {Spagnolo}, \citenamefont {Morris}, \citenamefont {Piacentini}, \citenamefont {Antesberger}, \citenamefont {Massa}, \citenamefont {Crespi}, \citenamefont {Ceccarelli}, \citenamefont {Osellame},\ and\ \citenamefont {Walther}}]{Spagnolo2022}%
  \BibitemOpen
  \bibfield  {author} {\bibinfo {author} {\bibfnamefont {Michele}\ \bibnamefont {Spagnolo}}, \bibinfo {author} {\bibfnamefont {Joshua}\ \bibnamefont {Morris}}, \bibinfo {author} {\bibfnamefont {Simone}\ \bibnamefont {Piacentini}}, \bibinfo {author} {\bibfnamefont {Michael}\ \bibnamefont {Antesberger}}, \bibinfo {author} {\bibfnamefont {Francesco}\ \bibnamefont {Massa}}, \bibinfo {author} {\bibfnamefont {Andrea}\ \bibnamefont {Crespi}}, \bibinfo {author} {\bibfnamefont {Francesco}\ \bibnamefont {Ceccarelli}}, \bibinfo {author} {\bibfnamefont {Roberto}\ \bibnamefont {Osellame}}, \ and\ \bibinfo {author} {\bibfnamefont {Philip}\ \bibnamefont {Walther}},\ }\bibfield  {title} {\enquote {\bibinfo {title} {Experimental photonic quantum memristor},}\ }\href {\doibase 10.1038/s41566-022-00973-5} {\bibfield  {journal} {\bibinfo  {journal} {Nature Photonics}\ }\textbf {\bibinfo {volume} {16}},\ \bibinfo {pages} {318--323} (\bibinfo {year} {2022})}\BibitemShut {NoStop}%
\bibitem [{\citenamefont {Suprano}\ \emph {et~al.}(2024)\citenamefont {Suprano}, \citenamefont {Zia}, \citenamefont {Innocenti}, \citenamefont {Lorenzo}, \citenamefont {Cimini}, \citenamefont {Giordani}, \citenamefont {Palmisano}, \citenamefont {Polino}, \citenamefont {Spagnolo}, \citenamefont {Sciarrino}, \citenamefont {Palma}, \citenamefont {Ferraro},\ and\ \citenamefont {Paternostro}}]{Suprano2024}%
  \BibitemOpen
  \bibfield  {author} {\bibinfo {author} {\bibfnamefont {Alessia}\ \bibnamefont {Suprano}}, \bibinfo {author} {\bibfnamefont {Danilo}\ \bibnamefont {Zia}}, \bibinfo {author} {\bibfnamefont {Luca}\ \bibnamefont {Innocenti}}, \bibinfo {author} {\bibfnamefont {Salvatore}\ \bibnamefont {Lorenzo}}, \bibinfo {author} {\bibfnamefont {Valeria}\ \bibnamefont {Cimini}}, \bibinfo {author} {\bibfnamefont {Taira}\ \bibnamefont {Giordani}}, \bibinfo {author} {\bibfnamefont {Ivan}\ \bibnamefont {Palmisano}}, \bibinfo {author} {\bibfnamefont {Emanuele}\ \bibnamefont {Polino}}, \bibinfo {author} {\bibfnamefont {Nicol\`o}\ \bibnamefont {Spagnolo}}, \bibinfo {author} {\bibfnamefont {Fabio}\ \bibnamefont {Sciarrino}}, \bibinfo {author} {\bibfnamefont {G.~Massimo}\ \bibnamefont {Palma}}, \bibinfo {author} {\bibfnamefont {Alessandro}\ \bibnamefont {Ferraro}}, \ and\ \bibinfo {author} {\bibfnamefont {Mauro}\ \bibnamefont {Paternostro}},\ }\bibfield  {title} {\enquote {\bibinfo {title} {Experimental property reconstruction in a
  photonic quantum extreme learning machine},}\ }\href {\doibase 10.1103/PhysRevLett.132.160802} {\bibfield  {journal} {\bibinfo  {journal} {Phys. Rev. Lett.}\ }\textbf {\bibinfo {volume} {132}},\ \bibinfo {pages} {160802} (\bibinfo {year} {2024})}\BibitemShut {NoStop}%
\bibitem [{\citenamefont {Steinbrecher}\ \emph {et~al.}(2019)\citenamefont {Steinbrecher}, \citenamefont {Olson}, \citenamefont {Englund},\ and\ \citenamefont {Carolan}}]{steinbrecherQONN}%
  \BibitemOpen
  \bibfield  {author} {\bibinfo {author} {\bibfnamefont {Gregory~R.}\ \bibnamefont {Steinbrecher}}, \bibinfo {author} {\bibfnamefont {Jonathan~P.}\ \bibnamefont {Olson}}, \bibinfo {author} {\bibfnamefont {Dirk}\ \bibnamefont {Englund}}, \ and\ \bibinfo {author} {\bibfnamefont {Jacques}\ \bibnamefont {Carolan}},\ }\bibfield  {title} {\enquote {\bibinfo {title} {Deterministic optimal quantum cloning via a quantum-optical neural network},}\ }\href {\doibase 10.1038/s41534-019-0174-7} {\bibfield  {journal} {\bibinfo  {journal} {npj Quantum Inf.}\ }\textbf {\bibinfo {volume} {5}},\ \bibinfo {pages} {60} (\bibinfo {year} {2019})}\BibitemShut {NoStop}%
\bibitem [{\citenamefont {Ewaniuk}\ \emph {et~al.}(2023)\citenamefont {Ewaniuk}, \citenamefont {Carolan}, \citenamefont {Shastri},\ and\ \citenamefont {Rotenberg}}]{ewaniukQONN}%
  \BibitemOpen
  \bibfield  {author} {\bibinfo {author} {\bibfnamefont {Jacob}\ \bibnamefont {Ewaniuk}}, \bibinfo {author} {\bibfnamefont {Jacques}\ \bibnamefont {Carolan}}, \bibinfo {author} {\bibfnamefont {Bhavin~J.}\ \bibnamefont {Shastri}}, \ and\ \bibinfo {author} {\bibfnamefont {Nir}\ \bibnamefont {Rotenberg}},\ }\bibfield  {title} {\enquote {\bibinfo {title} {Realistic quantum photonic neural networks},}\ }\href {\doibase 10.1002/qute.202200125} {\bibfield  {journal} {\bibinfo  {journal} {Adv. Quantum Technol.}\ }\textbf {\bibinfo {volume} {6}},\ \bibinfo {pages} {2200125} (\bibinfo {year} {2023})}\BibitemShut {NoStop}%
\bibitem [{\citenamefont {Stanev}\ \emph {et~al.}(2023{\natexlab{b}})\citenamefont {Stanev}, \citenamefont {Spagnolo},\ and\ \citenamefont {Sciarrino}}]{stanevQONN}%
  \BibitemOpen
  \bibfield  {author} {\bibinfo {author} {\bibfnamefont {Denis}\ \bibnamefont {Stanev}}, \bibinfo {author} {\bibfnamefont {Nicol\`o}\ \bibnamefont {Spagnolo}}, \ and\ \bibinfo {author} {\bibfnamefont {Fabio}\ \bibnamefont {Sciarrino}},\ }\bibfield  {title} {\enquote {\bibinfo {title} {Deterministic optimal quantum cloning via a quantum-optical neural network},}\ }\href {\doibase 10.1103/PhysRevResearch.5.013139} {\bibfield  {journal} {\bibinfo  {journal} {Phys. Rev. Res.}\ }\textbf {\bibinfo {volume} {5}},\ \bibinfo {pages} {013139} (\bibinfo {year} {2023}{\natexlab{b}})}\BibitemShut {NoStop}%
\bibitem [{\citenamefont {Wright}\ and\ \citenamefont {McMahon}(2020)}]{quantumneuralnetwork}%
  \BibitemOpen
  \bibfield  {author} {\bibinfo {author} {\bibfnamefont {Logan~G.}\ \bibnamefont {Wright}}\ and\ \bibinfo {author} {\bibfnamefont {Peter~L.}\ \bibnamefont {McMahon}},\ }\bibfield  {title} {\enquote {\bibinfo {title} {The capacity of quantum neural networks},}\ }in\ \href {https://opg.optica.org/abstract.cfm?URI=CLEO_SI-2020-JM4G.5} {\emph {\bibinfo {booktitle} {Conference on Lasers and Electro-Optics}}}\ (\bibinfo  {publisher} {Optica Publishing Group},\ \bibinfo {year} {2020})\ p.\ \bibinfo {pages} {JM4G.5}\BibitemShut {NoStop}%
\bibitem [{\citenamefont {Chabaud}\ \emph {et~al.}(2021)\citenamefont {Chabaud}, \citenamefont {Markham},\ and\ \citenamefont {Sohbi}}]{Chabaud2021quantummachine}%
  \BibitemOpen
  \bibfield  {author} {\bibinfo {author} {\bibfnamefont {Ulysse}\ \bibnamefont {Chabaud}}, \bibinfo {author} {\bibfnamefont {Damian}\ \bibnamefont {Markham}}, \ and\ \bibinfo {author} {\bibfnamefont {Adel}\ \bibnamefont {Sohbi}},\ }\bibfield  {title} {\enquote {\bibinfo {title} {Quantum machine learning with adaptive linear optics},}\ }\href {\doibase 10.22331/q-2021-07-05-496} {\bibfield  {journal} {\bibinfo  {journal} {{Quantum}}\ }\textbf {\bibinfo {volume} {5}},\ \bibinfo {pages} {496} (\bibinfo {year} {2021})}\BibitemShut {NoStop}%
\bibitem [{Note1()}]{Note1}%
  \BibitemOpen
  \bibinfo {note} {Here, universal refers to the fact that these circuits can implement any linear optical computation, not any quantum computation.}\BibitemShut {Stop}%
\bibitem [{\citenamefont {Osellame}\ \emph {et~al.}(2012)\citenamefont {Osellame}, \citenamefont {Cerullo},\ and\ \citenamefont {Ramponi}}]{osellame2012femtosecond}%
  \BibitemOpen
  \bibfield  {author} {\bibinfo {author} {\bibfnamefont {Roberto}\ \bibnamefont {Osellame}}, \bibinfo {author} {\bibfnamefont {Giulio}\ \bibnamefont {Cerullo}}, \ and\ \bibinfo {author} {\bibfnamefont {Roberta}\ \bibnamefont {Ramponi}},\ }\href {\doibase https://doi.org/10.1007/978-3-642-23366-1} {\emph {\bibinfo {title} {Femtosecond laser micromachining: photonic and microfluidic devices in transparent materials}}},\ Vol.\ \bibinfo {volume} {123}\ (\bibinfo  {publisher} {Springer},\ \bibinfo {year} {2012})\BibitemShut {NoStop}%
\bibitem [{\citenamefont {Corrielli}\ \emph {et~al.}(2021)\citenamefont {Corrielli}, \citenamefont {Crespi},\ and\ \citenamefont {Osellame}}]{Corrielli2021}%
  \BibitemOpen
  \bibfield  {author} {\bibinfo {author} {\bibfnamefont {Giacomo}\ \bibnamefont {Corrielli}}, \bibinfo {author} {\bibfnamefont {Andrea}\ \bibnamefont {Crespi}}, \ and\ \bibinfo {author} {\bibfnamefont {Roberto}\ \bibnamefont {Osellame}},\ }\bibfield  {title} {\enquote {\bibinfo {title} {Femtosecond laser micromachining for integrated quantum photonics},}\ }\href {\doibase 10.1515/nanoph-2021-0419} {\bibfield  {journal} {\bibinfo  {journal} {Nanophotonics}\ }\textbf {\bibinfo {volume} {10}},\ \bibinfo {pages} {3789–3812} (\bibinfo {year} {2021})}\BibitemShut {NoStop}%
\bibitem [{\citenamefont {Wang}\ \emph {et~al.}(2020)\citenamefont {Wang}, \citenamefont {Sciarrino}, \citenamefont {Laing},\ and\ \citenamefont {Thompson}}]{Wang2020_review}%
  \BibitemOpen
  \bibfield  {author} {\bibinfo {author} {\bibfnamefont {Jianwei}\ \bibnamefont {Wang}}, \bibinfo {author} {\bibfnamefont {Fabio}\ \bibnamefont {Sciarrino}}, \bibinfo {author} {\bibfnamefont {Anthony}\ \bibnamefont {Laing}}, \ and\ \bibinfo {author} {\bibfnamefont {Mark~G.}\ \bibnamefont {Thompson}},\ }\bibfield  {title} {\enquote {\bibinfo {title} {Integrated photonic quantum technologies},}\ }\href {\doibase 10.1038/s41566-019-0532-1} {\bibfield  {journal} {\bibinfo  {journal} {Nature Photonics}\ }\textbf {\bibinfo {volume} {14}},\ \bibinfo {pages} {273--284} (\bibinfo {year} {2020})}\BibitemShut {NoStop}%
\bibitem [{\citenamefont {Gattass}\ and\ \citenamefont {Mazur}(2008)}]{Gattass2008}%
  \BibitemOpen
  \bibfield  {author} {\bibinfo {author} {\bibfnamefont {Rafael~R.}\ \bibnamefont {Gattass}}\ and\ \bibinfo {author} {\bibfnamefont {Eric}\ \bibnamefont {Mazur}},\ }\bibfield  {title} {\enquote {\bibinfo {title} {Femtosecond laser micromachining in transparent materials},}\ }\href {\doibase 10.1038/nphoton.2008.47} {\bibfield  {journal} {\bibinfo  {journal} {Nature Photonics}\ }\textbf {\bibinfo {volume} {2}},\ \bibinfo {pages} {219--225} (\bibinfo {year} {2008})}\BibitemShut {NoStop}%
\bibitem [{\citenamefont {Clements}\ \emph {et~al.}(2016)\citenamefont {Clements}, \citenamefont {Humphreys}, \citenamefont {Metcalf}, \citenamefont {Kolthammer},\ and\ \citenamefont {Walmsley}}]{Clements:16}%
  \BibitemOpen
  \bibfield  {author} {\bibinfo {author} {\bibfnamefont {William~R.}\ \bibnamefont {Clements}}, \bibinfo {author} {\bibfnamefont {Peter~C.}\ \bibnamefont {Humphreys}}, \bibinfo {author} {\bibfnamefont {Benjamin~J.}\ \bibnamefont {Metcalf}}, \bibinfo {author} {\bibfnamefont {W.~Steven}\ \bibnamefont {Kolthammer}}, \ and\ \bibinfo {author} {\bibfnamefont {Ian~A.}\ \bibnamefont {Walmsley}},\ }\bibfield  {title} {\enquote {\bibinfo {title} {Optimal design for universal multiport interferometers},}\ }\href {\doibase 10.1364/OPTICA.3.001460} {\bibfield  {journal} {\bibinfo  {journal} {Optica}\ }\textbf {\bibinfo {volume} {3}},\ \bibinfo {pages} {1460--1465} (\bibinfo {year} {2016})}\BibitemShut {NoStop}%
\bibitem [{\citenamefont {Ollivier}\ \emph {et~al.}(2021)\citenamefont {Ollivier}, \citenamefont {Thomas}, \citenamefont {Wein}, \citenamefont {de~Buy~Wenniger}, \citenamefont {Coste}, \citenamefont {Loredo}, \citenamefont {Somaschi}, \citenamefont {Harouri}, \citenamefont {Lemaitre}, \citenamefont {Sagnes}, \citenamefont {Lanco}, \citenamefont {Simon}, \citenamefont {Anton}, \citenamefont {Krebs},\ and\ \citenamefont {Senellart}}]{Olivier_2021}%
  \BibitemOpen
  \bibfield  {author} {\bibinfo {author} {\bibfnamefont {H.}~\bibnamefont {Ollivier}}, \bibinfo {author} {\bibfnamefont {S.~E.}\ \bibnamefont {Thomas}}, \bibinfo {author} {\bibfnamefont {S.~C.}\ \bibnamefont {Wein}}, \bibinfo {author} {\bibfnamefont {I.~Maillette}\ \bibnamefont {de~Buy~Wenniger}}, \bibinfo {author} {\bibfnamefont {N.}~\bibnamefont {Coste}}, \bibinfo {author} {\bibfnamefont {J.~C.}\ \bibnamefont {Loredo}}, \bibinfo {author} {\bibfnamefont {N.}~\bibnamefont {Somaschi}}, \bibinfo {author} {\bibfnamefont {A.}~\bibnamefont {Harouri}}, \bibinfo {author} {\bibfnamefont {A.}~\bibnamefont {Lemaitre}}, \bibinfo {author} {\bibfnamefont {I.}~\bibnamefont {Sagnes}}, \bibinfo {author} {\bibfnamefont {L.}~\bibnamefont {Lanco}}, \bibinfo {author} {\bibfnamefont {C.}~\bibnamefont {Simon}}, \bibinfo {author} {\bibfnamefont {C.}~\bibnamefont {Anton}}, \bibinfo {author} {\bibfnamefont {O.}~\bibnamefont {Krebs}}, \ and\ \bibinfo {author} {\bibfnamefont {P.}~\bibnamefont {Senellart}},\ }\bibfield  {title}
  {\enquote {\bibinfo {title} {Hong-ou-mandel interference with imperfect single photon sources},}\ }\href {\doibase 10.1103/PhysRevLett.126.063602} {\bibfield  {journal} {\bibinfo  {journal} {Phys. Rev. Lett.}\ }\textbf {\bibinfo {volume} {126}},\ \bibinfo {pages} {063602} (\bibinfo {year} {2021})}\BibitemShut {NoStop}%
\bibitem [{\citenamefont {Pont}\ \emph {et~al.}(2022)\citenamefont {Pont}, \citenamefont {Albiero}, \citenamefont {Thomas}, \citenamefont {Spagnolo}, \citenamefont {Ceccarelli}, \citenamefont {Corrielli}, \citenamefont {Brieussel}, \citenamefont {Somaschi}, \citenamefont {Huet}, \citenamefont {Harouri}, \citenamefont {Lema\^{\i}tre}, \citenamefont {Sagnes}, \citenamefont {Belabas}, \citenamefont {Sciarrino}, \citenamefont {Osellame}, \citenamefont {Senellart},\ and\ \citenamefont {Crespi}}]{Pont_2022}%
  \BibitemOpen
  \bibfield  {author} {\bibinfo {author} {\bibfnamefont {Mathias}\ \bibnamefont {Pont}}, \bibinfo {author} {\bibfnamefont {Riccardo}\ \bibnamefont {Albiero}}, \bibinfo {author} {\bibfnamefont {Sarah~E.}\ \bibnamefont {Thomas}}, \bibinfo {author} {\bibfnamefont {Nicolo}\ \bibnamefont {Spagnolo}}, \bibinfo {author} {\bibfnamefont {Francesco}\ \bibnamefont {Ceccarelli}}, \bibinfo {author} {\bibfnamefont {Giacomo}\ \bibnamefont {Corrielli}}, \bibinfo {author} {\bibfnamefont {Alexandre}\ \bibnamefont {Brieussel}}, \bibinfo {author} {\bibfnamefont {Niccolo}\ \bibnamefont {Somaschi}}, \bibinfo {author} {\bibfnamefont {H\^elio}\ \bibnamefont {Huet}}, \bibinfo {author} {\bibfnamefont {Abdelmounaim}\ \bibnamefont {Harouri}}, \bibinfo {author} {\bibfnamefont {Aristide}\ \bibnamefont {Lema\^{\i}tre}}, \bibinfo {author} {\bibfnamefont {Isabelle}\ \bibnamefont {Sagnes}}, \bibinfo {author} {\bibfnamefont {Nadia}\ \bibnamefont {Belabas}}, \bibinfo {author} {\bibfnamefont {Fabio}\ \bibnamefont {Sciarrino}}, \bibinfo {author}
  {\bibfnamefont {Roberto}\ \bibnamefont {Osellame}}, \bibinfo {author} {\bibfnamefont {Pascale}\ \bibnamefont {Senellart}}, \ and\ \bibinfo {author} {\bibfnamefont {Andrea}\ \bibnamefont {Crespi}},\ }\bibfield  {title} {\enquote {\bibinfo {title} {Quantifying $n$-photon indistinguishability with a cyclic integrated interferometer},}\ }\href {\doibase 10.1103/PhysRevX.12.031033} {\bibfield  {journal} {\bibinfo  {journal} {Phys. Rev. X}\ }\textbf {\bibinfo {volume} {12}},\ \bibinfo {pages} {031033} (\bibinfo {year} {2022})}\BibitemShut {NoStop}%
\bibitem [{\citenamefont {Valeri}\ \emph {et~al.}(2024)\citenamefont {Valeri}, \citenamefont {Barigelli}, \citenamefont {Polacchi}, \citenamefont {Rodari}, \citenamefont {Santis}, \citenamefont {Giordani}, \citenamefont {Carvacho}, \citenamefont {Spagnolo},\ and\ \citenamefont {Sciarrino}}]{Valeri2024}%
  \BibitemOpen
  \bibfield  {author} {\bibinfo {author} {\bibfnamefont {Mauro}\ \bibnamefont {Valeri}}, \bibinfo {author} {\bibfnamefont {Paolo}\ \bibnamefont {Barigelli}}, \bibinfo {author} {\bibfnamefont {Beatrice}\ \bibnamefont {Polacchi}}, \bibinfo {author} {\bibfnamefont {Giovanni}\ \bibnamefont {Rodari}}, \bibinfo {author} {\bibfnamefont {Gianluca~De}\ \bibnamefont {Santis}}, \bibinfo {author} {\bibfnamefont {Taira}\ \bibnamefont {Giordani}}, \bibinfo {author} {\bibfnamefont {Gonzalo}\ \bibnamefont {Carvacho}}, \bibinfo {author} {\bibfnamefont {Nicolò}\ \bibnamefont {Spagnolo}}, \ and\ \bibinfo {author} {\bibfnamefont {Fabio}\ \bibnamefont {Sciarrino}},\ }\bibfield  {title} {\enquote {\bibinfo {title} {Generation and characterization of polarization-entangled states using quantum dot single-photon sources},}\ }\href {\doibase 10.1088/2058-9565/ad1c44} {\bibfield  {journal} {\bibinfo  {journal} {Quantum Science and Technology}\ }\textbf {\bibinfo {volume} {9}},\ \bibinfo {pages} {025002} (\bibinfo {year}
  {2024})}\BibitemShut {NoStop}%
\bibitem [{\citenamefont {Yang}\ \emph {et~al.}(2021)\citenamefont {Yang}, \citenamefont {Wu}, \citenamefont {Peng}, \citenamefont {Okolo}, \citenamefont {Zhang}, \citenamefont {Zhao},\ and\ \citenamefont {Sun}}]{yang2021parameter}%
  \BibitemOpen
  \bibfield  {author} {\bibinfo {author} {\bibfnamefont {Jiawei}\ \bibnamefont {Yang}}, \bibinfo {author} {\bibfnamefont {Zeping}\ \bibnamefont {Wu}}, \bibinfo {author} {\bibfnamefont {Ke}~\bibnamefont {Peng}}, \bibinfo {author} {\bibfnamefont {Patrick~N}\ \bibnamefont {Okolo}}, \bibinfo {author} {\bibfnamefont {Weihua}\ \bibnamefont {Zhang}}, \bibinfo {author} {\bibfnamefont {Hailong}\ \bibnamefont {Zhao}}, \ and\ \bibinfo {author} {\bibfnamefont {Jingbo}\ \bibnamefont {Sun}},\ }\bibfield  {title} {\enquote {\bibinfo {title} {Parameter selection of gaussian kernel svm based on local density of training set},}\ }\href@noop {} {\bibfield  {journal} {\bibinfo  {journal} {Inverse Problems in Science and Engineering}\ }\textbf {\bibinfo {volume} {29}},\ \bibinfo {pages} {536--548} (\bibinfo {year} {2021})}\BibitemShut {NoStop}%
\bibitem [{\citenamefont {MacQueen}(1967)}]{macqueen1967some}%
  \BibitemOpen
  \bibfield  {author} {\bibinfo {author} {\bibfnamefont {J.}~\bibnamefont {MacQueen}},\ }\bibfield  {title} {\enquote {\bibinfo {title} {Some methods for classification and analysis of multivariate observations},}\ }in\ \href@noop {} {\emph {\bibinfo {booktitle} {Proceedings of the fifth Berkeley symposium on mathematical statistics and probability}}},\ Vol.~\bibinfo {volume} {1}\ (\bibinfo {organization} {Oakland, CA, USA},\ \bibinfo {year} {1967})\ pp.\ \bibinfo {pages} {281--297}\BibitemShut {NoStop}%
\bibitem [{\citenamefont {Aaronson}\ and\ \citenamefont {Brod}(2016)}]{Aronson_losses}%
  \BibitemOpen
  \bibfield  {author} {\bibinfo {author} {\bibfnamefont {Scott}\ \bibnamefont {Aaronson}}\ and\ \bibinfo {author} {\bibfnamefont {Daniel~J.}\ \bibnamefont {Brod}},\ }\bibfield  {title} {\enquote {\bibinfo {title} {Bosonsampling with lost photons},}\ }\href {\doibase 10.1103/PhysRevA.93.012335} {\bibfield  {journal} {\bibinfo  {journal} {Phys. Rev. A}\ }\textbf {\bibinfo {volume} {93}},\ \bibinfo {pages} {012335} (\bibinfo {year} {2016})}\BibitemShut {NoStop}%
\bibitem [{\citenamefont {Oszmaniec}\ and\ \citenamefont {Brod}(2018)}]{Oszmaniec2018}%
  \BibitemOpen
  \bibfield  {author} {\bibinfo {author} {\bibfnamefont {Michał}\ \bibnamefont {Oszmaniec}}\ and\ \bibinfo {author} {\bibfnamefont {Daniel~J}\ \bibnamefont {Brod}},\ }\bibfield  {title} {\enquote {\bibinfo {title} {Classical simulation of photonic linear optics with lost particles},}\ }\href {\doibase 10.1088/1367-2630/aadfa8} {\bibfield  {journal} {\bibinfo  {journal} {New Journal of Physics}\ }\textbf {\bibinfo {volume} {20}},\ \bibinfo {pages} {092002} (\bibinfo {year} {2018})}\BibitemShut {NoStop}%
\bibitem [{\citenamefont {García-Patrón}\ \emph {et~al.}(2019)\citenamefont {García-Patrón}, \citenamefont {Renema},\ and\ \citenamefont {Shchesnovich}}]{GarcaPatrn2019}%
  \BibitemOpen
  \bibfield  {author} {\bibinfo {author} {\bibfnamefont {Raúl}\ \bibnamefont {García-Patrón}}, \bibinfo {author} {\bibfnamefont {Jelmer~J.}\ \bibnamefont {Renema}}, \ and\ \bibinfo {author} {\bibfnamefont {Valery}\ \bibnamefont {Shchesnovich}},\ }\bibfield  {title} {\enquote {\bibinfo {title} {Simulating boson sampling in lossy architectures},}\ }\href {\doibase 10.22331/q-2019-08-05-169} {\bibfield  {journal} {\bibinfo  {journal} {Quantum}\ }\textbf {\bibinfo {volume} {3}},\ \bibinfo {pages} {169} (\bibinfo {year} {2019})}\BibitemShut {NoStop}%
\bibitem [{\citenamefont {Oh}\ \emph {et~al.}(2021)\citenamefont {Oh}, \citenamefont {Noh}, \citenamefont {Fefferman},\ and\ \citenamefont {Jiang}}]{Oh2021}%
  \BibitemOpen
  \bibfield  {author} {\bibinfo {author} {\bibfnamefont {Changhun}\ \bibnamefont {Oh}}, \bibinfo {author} {\bibfnamefont {Kyungjoo}\ \bibnamefont {Noh}}, \bibinfo {author} {\bibfnamefont {Bill}\ \bibnamefont {Fefferman}}, \ and\ \bibinfo {author} {\bibfnamefont {Liang}\ \bibnamefont {Jiang}},\ }\bibfield  {title} {\enquote {\bibinfo {title} {Classical simulation of lossy boson sampling using matrix product operators},}\ }\href {\doibase 10.1103/physreva.104.022407} {\bibfield  {journal} {\bibinfo  {journal} {Physical Review A}\ }\textbf {\bibinfo {volume} {104}} (\bibinfo {year} {2021}),\ 10.1103/physreva.104.022407}\BibitemShut {NoStop}%
\bibitem [{\citenamefont {Oh}\ \emph {et~al.}(2022)\citenamefont {Oh}, \citenamefont {Lim}, \citenamefont {Fefferman},\ and\ \citenamefont {Jiang}}]{Oh2022}%
  \BibitemOpen
  \bibfield  {author} {\bibinfo {author} {\bibfnamefont {Changhun}\ \bibnamefont {Oh}}, \bibinfo {author} {\bibfnamefont {Youngrong}\ \bibnamefont {Lim}}, \bibinfo {author} {\bibfnamefont {Bill}\ \bibnamefont {Fefferman}}, \ and\ \bibinfo {author} {\bibfnamefont {Liang}\ \bibnamefont {Jiang}},\ }\bibfield  {title} {\enquote {\bibinfo {title} {Classical simulation of boson sampling based on graph structure},}\ }\href {\doibase 10.1103/physrevlett.128.190501} {\bibfield  {journal} {\bibinfo  {journal} {Physical Review Letters}\ }\textbf {\bibinfo {volume} {128}} (\bibinfo {year} {2022}),\ 10.1103/physrevlett.128.190501}\BibitemShut {NoStop}%
\bibitem [{\citenamefont {Oh}\ \emph {et~al.}(2023{\natexlab{b}})\citenamefont {Oh}, \citenamefont {Jiang},\ and\ \citenamefont {Fefferman}}]{oh2023}%
  \BibitemOpen
  \bibfield  {author} {\bibinfo {author} {\bibfnamefont {Changhun}\ \bibnamefont {Oh}}, \bibinfo {author} {\bibfnamefont {Liang}\ \bibnamefont {Jiang}}, \ and\ \bibinfo {author} {\bibfnamefont {Bill}\ \bibnamefont {Fefferman}},\ }\href {https://arxiv.org/abs/2301.11532} {\enquote {\bibinfo {title} {On classical simulation algorithms for noisy boson sampling},}\ } (\bibinfo {year} {2023}{\natexlab{b}}),\ \Eprint {http://arxiv.org/abs/2301.11532} {arXiv:2301.11532 [quant-ph]} \BibitemShut {NoStop}%
\bibitem [{\citenamefont {Leverrier}\ and\ \citenamefont {Garcia-Patron}(2015)}]{leverrier2015}%
  \BibitemOpen
  \bibfield  {author} {\bibinfo {author} {\bibfnamefont {Anthony}\ \bibnamefont {Leverrier}}\ and\ \bibinfo {author} {\bibfnamefont {Raul}\ \bibnamefont {Garcia-Patron}},\ }\bibfield  {title} {\enquote {\bibinfo {title} {Analysis of circuit imperfections in bosonsampling},}\ }\href@noop {} {\bibfield  {journal} {\bibinfo  {journal} {Quantum Inf. Comput.}\ }\textbf {\bibinfo {volume} {15}},\ \bibinfo {pages} {489--512} (\bibinfo {year} {2015})}\BibitemShut {NoStop}%
\bibitem [{Note2()}]{Note2}%
  \BibitemOpen
  \bibinfo {note} {Note: U.\ Chabaud et al, ''Photonic quantum kernel methods beyond Boson Sampling'', \protect \textit {in preparation}.}\BibitemShut {Stop}%
\bibitem [{\citenamefont {Seron}\ \emph {et~al.}(2024)\citenamefont {Seron}, \citenamefont {Novo}, \citenamefont {Arkhipov},\ and\ \citenamefont {Cerf}}]{Seron2024}%
  \BibitemOpen
  \bibfield  {author} {\bibinfo {author} {\bibfnamefont {Benoit}\ \bibnamefont {Seron}}, \bibinfo {author} {\bibfnamefont {Leonardo}\ \bibnamefont {Novo}}, \bibinfo {author} {\bibfnamefont {Alex}\ \bibnamefont {Arkhipov}}, \ and\ \bibinfo {author} {\bibfnamefont {Nicolas~J.}\ \bibnamefont {Cerf}},\ }\bibfield  {title} {\enquote {\bibinfo {title} {Efficient validation of boson sampling from binned photon-number distributions},}\ }\href {\doibase 10.22331/q-2024-09-19-1479} {\bibfield  {journal} {\bibinfo  {journal} {Quantum}\ }\textbf {\bibinfo {volume} {8}},\ \bibinfo {pages} {1479} (\bibinfo {year} {2024})}\BibitemShut {NoStop}%
\bibitem [{\citenamefont {Gan}\ \emph {et~al.}(2022)\citenamefont {Gan}, \citenamefont {Leykam},\ and\ \citenamefont {Angelakis}}]{gan2022fock}%
  \BibitemOpen
  \bibfield  {author} {\bibinfo {author} {\bibfnamefont {Beng~Yee}\ \bibnamefont {Gan}}, \bibinfo {author} {\bibfnamefont {Daniel}\ \bibnamefont {Leykam}}, \ and\ \bibinfo {author} {\bibfnamefont {Dimitris~G}\ \bibnamefont {Angelakis}},\ }\bibfield  {title} {\enquote {\bibinfo {title} {Fock state-enhanced expressivity of quantum machine learning models},}\ }\href@noop {} {\bibfield  {journal} {\bibinfo  {journal} {EPJ Quantum Technology}\ }\textbf {\bibinfo {volume} {9}},\ \bibinfo {pages} {16} (\bibinfo {year} {2022})}\BibitemShut {NoStop}%
\bibitem [{\citenamefont {Huang}\ \emph {et~al.}(2021)\citenamefont {Huang}, \citenamefont {Broughton}, \citenamefont {Mohseni}, \citenamefont {Babbush}, \citenamefont {Boixo}, \citenamefont {Neven},\ and\ \citenamefont {McClean}}]{Huang2021}%
  \BibitemOpen
  \bibfield  {author} {\bibinfo {author} {\bibfnamefont {Hsin-Yuan}\ \bibnamefont {Huang}}, \bibinfo {author} {\bibfnamefont {Michael}\ \bibnamefont {Broughton}}, \bibinfo {author} {\bibfnamefont {Masoud}\ \bibnamefont {Mohseni}}, \bibinfo {author} {\bibfnamefont {Ryan}\ \bibnamefont {Babbush}}, \bibinfo {author} {\bibfnamefont {Sergio}\ \bibnamefont {Boixo}}, \bibinfo {author} {\bibfnamefont {Hartmut}\ \bibnamefont {Neven}}, \ and\ \bibinfo {author} {\bibfnamefont {Jarrod~R.}\ \bibnamefont {McClean}},\ }\bibfield  {title} {\enquote {\bibinfo {title} {Power of data in quantum machine learning},}\ }\href {\doibase 10.1038/s41467-021-22539-9} {\bibfield  {journal} {\bibinfo  {journal} {Nature Communications}\ }\textbf {\bibinfo {volume} {12}} (\bibinfo {year} {2021}),\ 10.1038/s41467-021-22539-9}\BibitemShut {NoStop}%
\bibitem [{\citenamefont {Jerbi}\ \emph {et~al.}(2023)\citenamefont {Jerbi}, \citenamefont {Fiderer}, \citenamefont {Poulsen~Nautrup}, \citenamefont {K\"{u}bler}, \citenamefont {Briegel},\ and\ \citenamefont {Dunjko}}]{Jerbi2023}%
  \BibitemOpen
  \bibfield  {author} {\bibinfo {author} {\bibfnamefont {Sofiene}\ \bibnamefont {Jerbi}}, \bibinfo {author} {\bibfnamefont {Lukas~J.}\ \bibnamefont {Fiderer}}, \bibinfo {author} {\bibfnamefont {Hendrik}\ \bibnamefont {Poulsen~Nautrup}}, \bibinfo {author} {\bibfnamefont {Jonas~M.}\ \bibnamefont {K\"{u}bler}}, \bibinfo {author} {\bibfnamefont {Hans~J.}\ \bibnamefont {Briegel}}, \ and\ \bibinfo {author} {\bibfnamefont {Vedran}\ \bibnamefont {Dunjko}},\ }\bibfield  {title} {\enquote {\bibinfo {title} {Quantum machine learning beyond kernel methods},}\ }\href {\doibase 10.1038/s41467-023-36159-y} {\bibfield  {journal} {\bibinfo  {journal} {Nature Communications}\ }\textbf {\bibinfo {volume} {14}} (\bibinfo {year} {2023}),\ 10.1038/s41467-023-36159-y}\BibitemShut {NoStop}%
\bibitem [{\citenamefont {Zhang}(2012)}]{Zhang2012}%
  \BibitemOpen
  \bibfield  {author} {\bibinfo {author} {\bibfnamefont {Shengyu}\ \bibnamefont {Zhang}},\ }\enquote {\bibinfo {title} {Bqp-complete problems},}\ in\ \href {\doibase 10.1007/978-3-540-92910-9_46} {\emph {\bibinfo {booktitle} {Handbook of Natural Computing}}},\ \bibinfo {editor} {edited by\ \bibinfo {editor} {\bibfnamefont {Grzegorz}\ \bibnamefont {Rozenberg}}, \bibinfo {editor} {\bibfnamefont {Thomas}\ \bibnamefont {B{\"a}ck}}, \ and\ \bibinfo {editor} {\bibfnamefont {Joost~N.}\ \bibnamefont {Kok}}}\ (\bibinfo  {publisher} {Springer Berlin Heidelberg},\ \bibinfo {address} {Berlin, Heidelberg},\ \bibinfo {year} {2012})\ pp.\ \bibinfo {pages} {1545--1571}\BibitemShut {NoStop}%
\bibitem [{Note3()}]{Note3}%
  \BibitemOpen
  \bibinfo {note} {Actually, both problems are formally \protect \textsf {PromiseBQP}-complete.}\BibitemShut {Stop}%
\bibitem [{\citenamefont {Buhrman}\ \emph {et~al.}(2001)\citenamefont {Buhrman}, \citenamefont {Cleve}, \citenamefont {Watrous},\ and\ \citenamefont {De~Wolf}}]{buhrman2001quantum}%
  \BibitemOpen
  \bibfield  {author} {\bibinfo {author} {\bibfnamefont {Harry}\ \bibnamefont {Buhrman}}, \bibinfo {author} {\bibfnamefont {Richard}\ \bibnamefont {Cleve}}, \bibinfo {author} {\bibfnamefont {John}\ \bibnamefont {Watrous}}, \ and\ \bibinfo {author} {\bibfnamefont {Ronald}\ \bibnamefont {De~Wolf}},\ }\bibfield  {title} {\enquote {\bibinfo {title} {Quantum fingerprinting},}\ }\href {\doibase 10.1103/PhysRevLett.87.167902} {\bibfield  {journal} {\bibinfo  {journal} {Physical Review Letters}\ }\textbf {\bibinfo {volume} {87}},\ \bibinfo {pages} {167902} (\bibinfo {year} {2001})}\BibitemShut {NoStop}%
\end{thebibliography}
%\bibliographystyle{apsrev4-1}

\providecommand{\noopsort}[1]{}\providecommand{\singleletter}[1]{#1}%

\section*{Acknowledgments}
This work is supported by the ERC Advanced Grant QU-BOSS (QUantum advantage via non-linear BOSon Sampling, Grant Agreement No. 884676), the PNRR MUR project PE0000023-NQSTI (Spoke 4 and Spoke 7) and the European Union’s
Horizon Europe research and innovation program under
EPIQUE Project (Grant Agreement No. 101135288). U.C.\ thanks P.E.\ Emeriau, A.\ Sohbi, E.\ Kashefi and D.\ Markham for interesting discussions. Fabrication of the femtosecond laser-written integrated circuits was partially performed at PoliFAB, the micro and nano-fabrication facility of Politecnico di Milano (https://www.polifab.polimi.it). M.G., F.C. and R.O. wish to thank the PoliFAB staff for the valuable technical support.

%---------------------------------------------

{
\section*{Author contributions}
F.H., E.C., G.R., T.F., A.S., T.G., G.Ca., N.S., S.K., M.P., C.L., F.C., M.D., U.C., and F.S. developed the theoretical framework and conceived the experimental scheme. 
R.A., N.D.G., M.G., F.C., G.Co., and R.O. fabricated the photonic chip and characterized the integrated devices using classical optics.  F.H., E.C., G.R., T.F., A.S., T.G.,  G.Ca., N.S. and F.S. carried out the quantum experiments and performed the data analysis. All the authors discussed the results and contributed to the preparation of the manuscript.}

\section*{Competing Interests}
The authors declare no competing interests.

%---------------------------------------------

%---------------------------------------------

%\end{appendix}

\end{document}

\title{Supplementary Information: {Quantum machine learning with Adaptive Boson Sampling via post-selection}}%Experimental photonic quantum machine learning via Adaptive Boson Sampling}
\author{Francesco Hoch}
\affiliation{Dipartimento di Fisica, Sapienza Universit\`{a} di Roma, Piazzale Aldo Moro 5, I-00185 Roma, Italy}

\author{Eugenio Caruccio}
\affiliation{Dipartimento di Fisica, Sapienza Universit\`{a} di Roma, Piazzale Aldo Moro 5, I-00185 Roma, Italy}

\author{Giovanni Rodari}
\affiliation{Dipartimento di Fisica, Sapienza Universit\`{a} di Roma, Piazzale Aldo Moro 5, I-00185 Roma, Italy}

\author{Tommaso Francalanci}
\affiliation{Dipartimento di Fisica, Sapienza Universit\`{a} di Roma, Piazzale Aldo Moro 5, I-00185 Roma, Italy}

\author{Alessia Suprano}
\affiliation{Dipartimento di Fisica, Sapienza Universit\`{a} di Roma, Piazzale Aldo Moro 5, I-00185 Roma, Italy}

\author{Taira Giordani}
\affiliation{Dipartimento di Fisica, Sapienza Universit\`{a} di Roma, Piazzale Aldo Moro 5, I-00185 Roma, Italy}

\author{Gonzalo Carvacho}
\affiliation{Dipartimento di Fisica, Sapienza Universit\`{a} di Roma, Piazzale Aldo Moro 5, I-00185 Roma, Italy}

\author{Nicol\`o Spagnolo}
\affiliation{Dipartimento di Fisica, Sapienza Universit\`{a} di Roma, Piazzale Aldo Moro 5, I-00185 Roma, Italy}

\author{Seid Koudia}
\affiliation{Leonardo S.p.A., Leonardo Labs, Quantum technologies lab, Via Tiburtina, KM 12.400, 00131 Roma, Italy}

\author{Massimiliano Proietti}
\affiliation{Leonardo S.p.A., Leonardo Labs, Quantum technologies lab, Via Tiburtina, KM 12.400, 00131 Roma, Italy}

\author{Carlo Liorni}
\affiliation{Leonardo S.p.A., Leonardo Labs, Quantum technologies lab, Via Tiburtina, KM 12.400, 00131 Roma, Italy}

\author{Filippo Cerocchi}
\affiliation{Leonardo S.p.A., Cyber \& Security Solutions Division, Via Laurentina - 760, 00143 Rome, Italy}

\author{Riccardo Albiero}
\affiliation{Dipartimento di Fisica, Politecnico di Milano, Piazza Leonardo da Vinci 32, 20133 Milano, Italy}
\affiliation{Istituto di Fotonica e Nanotecnologie, Consiglio Nazionale delle Ricerche (IFN-CNR), 
Piazza Leonardo da Vinci, 32, I-20133 Milano, Italy}

\author{Niki Di Giano}
\affiliation{Dipartimento di Fisica, Politecnico di Milano, Piazza Leonardo da Vinci 32, 20133 Milano, Italy}
\affiliation{Istituto di Fotonica e Nanotecnologie, Consiglio Nazionale delle Ricerche (IFN-CNR), 
Piazza Leonardo da Vinci, 32, I-20133 Milano, Italy}

\author{Marco Gardina}
\affiliation{Istituto di Fotonica e Nanotecnologie, Consiglio Nazionale delle Ricerche (IFN-CNR), 
Piazza Leonardo da Vinci, 32, I-20133 Milano, Italy}

\author{Francesco Ceccarelli}
\affiliation{Istituto di Fotonica e Nanotecnologie, Consiglio Nazionale delle Ricerche (IFN-CNR), 
Piazza Leonardo da Vinci, 32, I-20133 Milano, Italy}

\author{Giacomo Corrielli}
\affiliation{Istituto di Fotonica e Nanotecnologie, Consiglio Nazionale delle Ricerche (IFN-CNR), 
Piazza Leonardo da Vinci, 32, I-20133 Milano, Italy}

\author{Ulysse Chabaud}
\affiliation{DIENS, \'Ecole Normale Sup\'erieure, PSL University, CNRS, INRIA, 45 rue d’Ulm, Paris, 75005, France}

\author{Roberto Osellame}
\affiliation{Istituto di Fotonica e Nanotecnologie, Consiglio Nazionale delle Ricerche (IFN-CNR), 
Piazza Leonardo da Vinci, 32, I-20133 Milano, Italy}

\author{Massimiliano Dispenza}
\affiliation{Leonardo S.p.A., Leonardo Labs, Quantum technologies lab, Via Tiburtina, KM 12.400, 00131 Roma, Italy}

\author{Fabio Sciarrino}
\email{fabio.sciarrino@uniroma1.it}
\affiliation{Dipartimento di Fisica, Sapienza Universit\`{a} di Roma, Piazzale Aldo Moro 5, I-00185 Roma, Italy}

\maketitle

\section{Single-photon sources, {photons synchronization and detection}}

The first experimental platform makes use of a parametric source based on the spontaneous down-conversion process. A pulsed-laser at the wavelength of $\lambda = 392.5$ nm pumps a nonlinear crystal of Beta-Barium Borate. The parametric process generates photon pairs at $\lambda = 785$ nm.  
The pump-power is controlled in such a way that the single-pair emission scenario is much more likely to happen than the multi-pair generation. 
The two photons are prepared to be fully indistinguishable in the polarization state, wavelength and time arrival in the integrated circuit. The visibility of the Hong-Ou-Mandel (HOM) dip measured at the output of the chip is up to $99.8\pm0.1 \%$. The two-fold coincidences rate of the experiment is around $\sim 2.5 $ Khz.

The second source is based on a quantum dot emitter. The main property of this near-deterministic source is the high brightness that allows to perform multi-photon experiments. %\cite{gazzano2013bright, nowak2014deterministic, somaschi2016near, ollivier2020reproducibility, thomas2021bright} 
In particular, the single photon source is a commercial solution by \textit{Quandela} which operates at the cryogenic temperature of $4\,K$. The emitter is a single self-assembled InGaAs quantum dot embedded in an electrically controlled micro-cavity \cite{Somaschi2016}. The sources operates in the longitudinal acoustic (LA) phonon-assisted configuration \cite{Thomas2021}. 
This means that the excitation laser and the single-photon have different wavelengths. In fact, the pump pulsed-laser with a repetition rate of $\sim 79$ Mhz excites the source at a wavelength of $927.2$ nm. In contrast, the wavelength of the emitted photons is slightly red-shifted ($927.8 \pm 0.2$ nm) so that the single-photon signal can be separated from the residual pump laser through a sequence of narrow band-pass filters. The single-photon counts rate is $\sim 3.5$  MHz measured through an avalanche photo-diode with $\sim 30$ \% efficiency at $927$ nm. The typical single-photon purity summarized by the second-order correlation function is $g^{(2)}(0) \sim 2 \%$. To perform the two- and the tree-photon experiments we need to route the emitted photons in different paths towards the input waveguides of the 8-mode device. This is achieved via a time-to-space demultiplexer (DMX) based on acoustic-optical modulation which actively distributes the stream of photons over three output channels. The photons are then synchronized and prepared in the same polarization state. The average visibility of the HOM of the three possible photon pairs among the three synchronized photons estimated at the outputs of the integrated circuit is $\sim 83 \%$.
{The overall transmission of the experiment depends, for platform B, on the polarized brightness of the QD source together with the transmission efficiency of the demultiplexing setup. Namely, we estimated a degree of linear polarization of $\sim 80 \%$, corresponding to a polarized fibered brightness of $\sim 11 \%$. Then, we have the transmission efficiency of the DMX setup, which has been measured via classical light to be approximately $80\%$.}

{In platform A, the single-photon detectors are avalanche photo-diodes (APDs), model 'SPCM-NIR' manufactured by 'Excelitas Technologies'. These detectors have a detection efficiency of approximately $70 \%$ for the wavelength of $785$ nm, in which platform A operates. The same model of detector is employed in platform B as well, where the wavelength is $928$ nm. For this wavelength, the detection efficiency is lower as stated in the datasheet of the detectors. In platform B, a second class of single-photon detectors is also employed: superconductive nano-wires single-photon detectors (SNSPDs), designed by 'Single Quantum'. Specifically, the efficiency of this detectors depends on the injected bias current and the polarization of the incoming photons. Once these two quantities are both optimized, the efficiency of the superconducting detectors is in the interval between $75$ and $90 \%$.}

\section{Universal integrated photonic circuits}
The integrated circuits exploited in platform A and B were fabricated by using the femtosecond laser writing technology. The layouts, based on a rectangular mesh of Mach-Zehnder interferometers (MZIs) \cite{Clements:16}, enable the generation of any linear transformation on 6 (A) and 8 (B) modes. Each MZI cell comprises two cascaded waveguide directional couplers and is equipped with two thermal phase shifters (i.e. resistive heaters), one to reconfigure the relative phase between the two modes and the other one to control the intensity distribution at the output of the cell. 
Both the circuits were manufactured in a Corning EAGLE XG glass substrate. Single-mode waveguides were inscribed directly into the substrate at $25 \mu$m below its surface, with an interwaveguide pitch of $80 \mu$m. This process was followed by a thermal annealing process to remove the induced stress and reduce losses \cite{corrielli2018symmetric}. The irradiation parameters were optimized according to the operation wavelength, which is $785$ nm (A) and $927$ nm (B). The heaters were fabricated either (A) by femtosecond laser ablation of a thin gold layer deposited through thermal evaporation \cite{ceccarelli2020low} or (B) by two-step photolithography and subsequent etching of chromium (for the heaters) and copper (for the metal connections) \cite{albiero2022toward}. Further microstructuration of the chip surface, i.e. the fabrication of thermal insulation trenches around the heaters, enhances their power efficiency and reduces crosstalk. Both devices occupy an area of $80×20 mm^2$ and feature fiber arrays glued to the input/output ports to ease the optical connection, with a fiber-to-fiber optical loss $< 3$ dB.
{Such a value corresponds to the total insertion loss of the whole integrated photon circuits with a measured total transmission amounting to $\approx 50 \%$. Thus, it accounts for the overall losses of the photonic system comprising of the losses due to transmission in the integrated chips together with the input and output losses of the fiber arrays %coupling light into the chip, 
 which are directly pigtailed to the chip itself.}
Before the quantum experiments, the calibration of the circuits was conducted by using coherent light. The accuracy reached by this process was evaluated by implementing 1000 (A) and 100 (B) randomly selected unitary transformations and calculating the amplitude fidelity according to this formula:

\begin{equation}
    F_{ampl}=\frac{1}{N}tr(|U_{th}^{\dagger}||U_{exp}|)
\end{equation}

where $N$ is the number of modes, $U_{exp}$ is the unitary matrix whose moduli are measured at the output of the circuit and $U_{th}$ is the target matrix. The result was an average fidelity of $99.7\%$ (A) and $99.1\%$ (B). Further details on the circuit characterization with classical light can be found in \cite{pentangelo2024high}.

\section{Data analysis}

The quantum feature map of the experiment encodes classical data in path-encoded qubits and qutrit. A state tomography is performed by measuring the three observables $\sigma_x,\, \sigma_y, \, \sigma_y $ for the qubit and the generalization of Pauli operators for qutris and thus retrieving the Bloch vectors $\vec{v}$ \cite{Munro2002,qudit_tomography2}. The components of such vectors will be the expectation values of the Pauli operators. %=\{\langle{\sigma_x}\rangle, \,\langle \sigma_y \rangle, \,\langle {\sigma_z} \rangle \}$. 
The density matrix of each qudit will be
\begin{equation}
    \rho = \frac{\mathbb{I}}{d}+ \frac{\vec{v} \cdot \vec{\sigma}}{2},
    \label{eq:rho_matrix}
\end{equation}
where $\vec{\sigma}$ is the vector of Pauli matrices, $\mathbb{I}$ the identity operator and $d$ is the dimension of the state. In the experiment, we reconstructed qubits $d=2$ and qutrits $d=3$. Each state is reconstructed from the outcomes of the tomography projective measurements with a maximum-likelihood approach, which consists in finding, among all possible density matrices describing the state, the one which maximizes the probability of obtaining the experimental data \cite{max_like}. The kernels of the ABS are estimated through the Uhlmann state fidelity $\mathcal{F}$ 
\begin{equation}
    K_{ij}=\mathcal{F}(\rho_i,\rho_j),
\end{equation} 
where $\mathcal{F}$ is defined as 
\begin{equation}
    \mathcal{F}(\rho_i, \rho_j)=\biggl(\Tr\sqrt{\sqrt{\rho_i}\rho_j \sqrt{\rho_i}}\biggr)^2.
    \label{eq:fid}
\end{equation}
We used the last two expressions for estimating the kernels and the comparison with a theoretical modeling of the experiment. In particular, we calculated the fidelity in \eqref{eq:fid} between the experimental states and the expected ones.

\begin{figure}
    \centering
    \includegraphics[width = \textwidth]{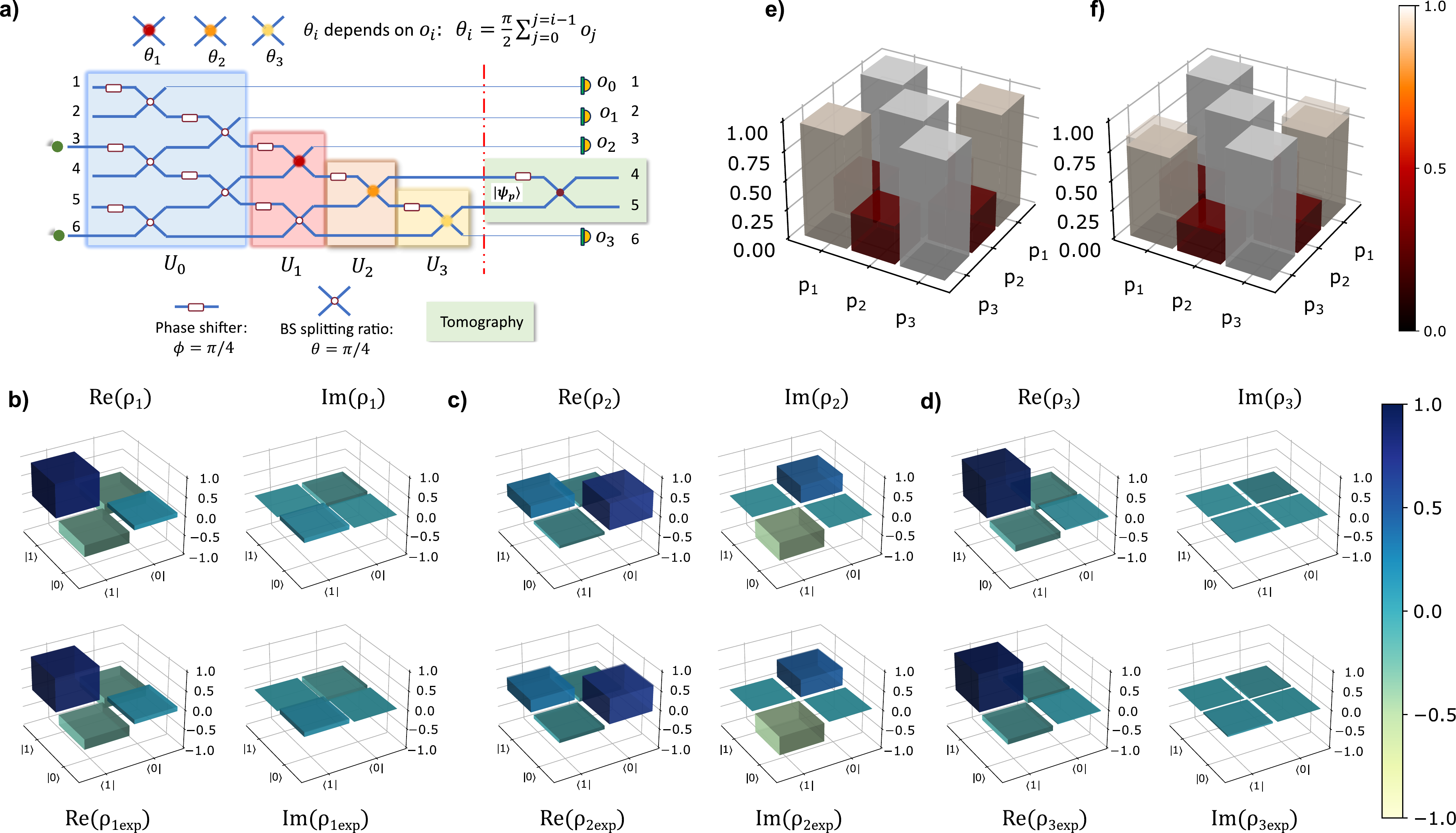}
    \caption{\textbf{Two-photon ABS - platform A.}
        a) The experiment implements a feature map for qubits by making interfere $n=2$ photons in $m=6$ modes and applying $D=3$ adaptive measurements %$(o_1, o_2, o_3)$ 
    on $r=1$ photons. %The circuit is encoded in a 6-mode universal and programmable chip. In particular, we have a first six-mode $U_0$ and then the three adaptive transformations $U_k$. The phase shifters $\phi$ (rectangles in the figure) and beam-splitter reflectivities $\theta$ (circles) are set to the angle $\theta,\phi = \pi/4$ except for the $\theta_k$ of the $U_k$ highlighted in red, orange and yellow that depends on the detection of one photon in the $o_j$. Pairs of indistinguishable photons generated by parametric down-conversion evolve in such an interferometer and a qubit tomography conditioned to the detection $o_k$ is performed in the green part of the circuit. b-
    b) Theoretical $\rho_i$ (top) and experimental $\rho_{i\,exp}$ (bottom) density matrices for the quantum states b) $\rho_{1}$, c) $\rho_{2}$ and d) $\rho_{3}$. We retrieved the density matrix by performing the qubit tomography with a tunable beam-splitter and phase-shifter. The uncertainties due to photon-counting statistics are not appreciable at the image scale. 
    e) Theoretical kernel for an ideal device and experiment. f) The kernel reconstructed from the states measured at the output of the ABS.
    The shaded area on the top of the bars represents the uncertainty due to the Poisson statistics of single-photon counts.}
    \label{fig:six_mode}
\end{figure}

The fidelities for the three qubits of platform A are reported in the main text. Fig. \ref{fig:six_mode} integrates with some results not shown in the main text. We provide the comparison with the theoretical kernel of an ideal photon source, i.e. perfect indistinguishability. We also report all the the experimental density matrices and their corresponding theoretical simulations. 

For what concerns platform B, here we first the results of the preliminary two photon experiment not shown in the main text. In this case, the ABS scheme is $[8,2,2,6]$ that means the encoding of $D=6$ classical data in qubits states. In Fig. \ref{fig:exp_8m3p} the $6 \times 6$ experimental and expected kernels are reported. In this case, the theoretical modeling considers an imperfect single-photon source that generates partially indistinguishable photons and with an imperfect single-photon purity.
Finally, the six fidelities and the 15 fidelities of the schemes B1, B2, and B3 between the experimental density matrices and the modeling of the source of platform B are reported in Fig. \ref{fig:fid}.

\begin{figure}[t!]
    \centering
    \includegraphics[width = 1\textwidth]{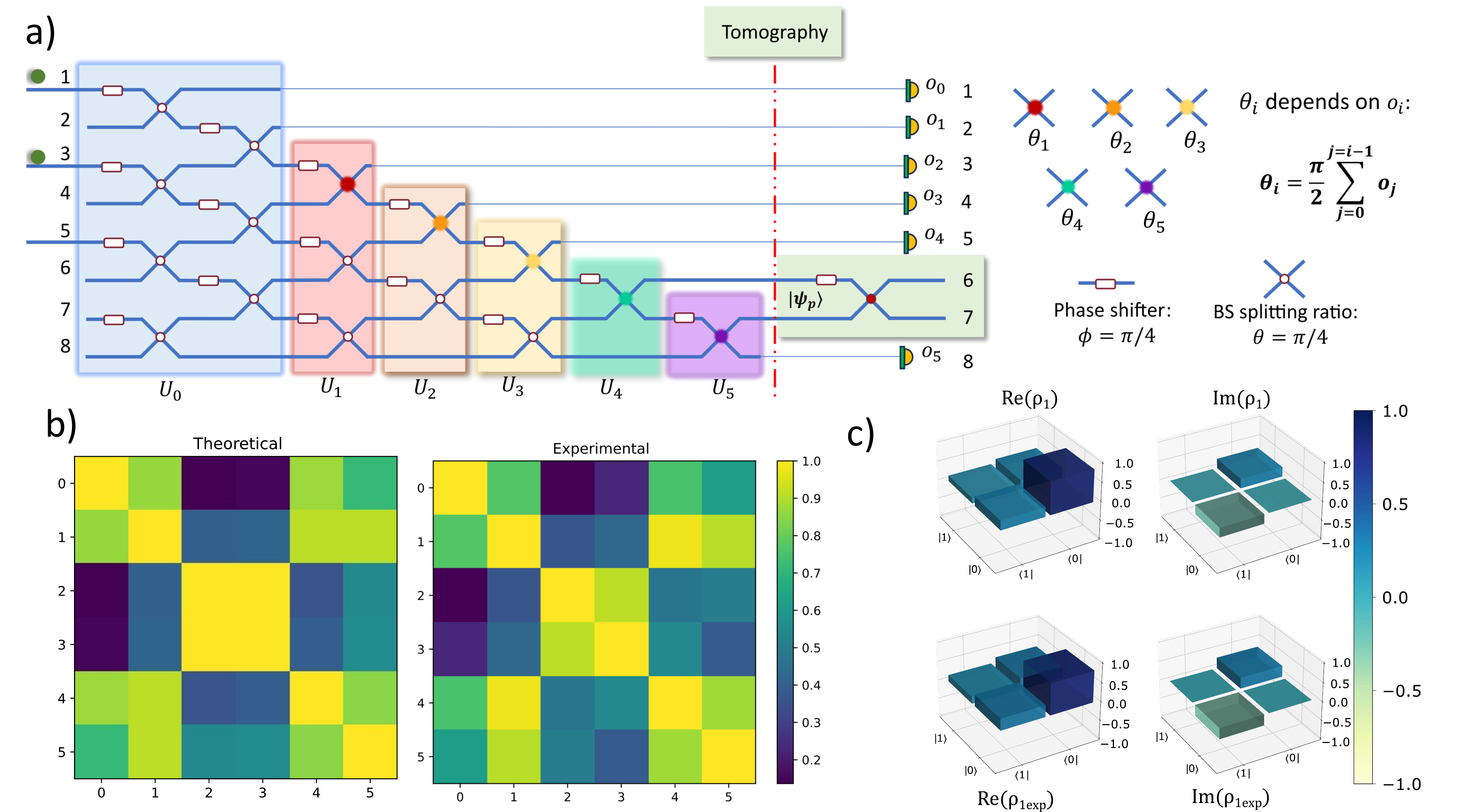}
    \caption{\textbf{Two-photon ABS %in an eight-mode adaptive Boson Sampling interferometer 
    - platform B.} a) We implemented the experiment with $n=2$ from the quantum dot source that are processed in the $m=8$ mode universal integrated circuit.  $r=1$ photon is measured in one of the $6$ modes and thus triggers the application of one of the $6$ transformations.  %In the 3-photon case, $r=2$ photons are detected in the $6$ adaptive modes. In this last scenario, we have a total amount of $15$  transformations each of them triggered by the detection of two photons in a pair of the $6$ outputs. In both experiments 
    The optical circuit is divided into an 8-mode unitary $U_0$, and five transformations $U_i$ that are activated and combined according to the configurations of the $r$ photons detected in the $6$ modes. 
    The reflectivity values of the beam-splitters $\theta_i$ in red, orange, yellow, teal, and violet, depend on where the ancillary photons are detected, according to the formula displayed in the figure. The green part of the circuit is the tomography station that analyzes the dual rail qubit encoded in the remaining photon conditioned on the detection of the $r$ photons in the other $o_i$ outputs.
    b) Comparison of the $6 \times 6$ kernels %for the 2-photon experiment computed 
    according to the theoretical modeling which assumes an imperfect single-photon source and the kernel reconstructed from the states measured at the output. %of the $n=2$, $m=8$, $k=6$, and $r=1$ adaptive boson sampler. 
    c) One of the $6$ experimental density matrices and the comparison with the expectations.}
    \label{fig:exp_8m3p}
\end{figure}

\begin{figure}[t]
    \centering
    \includegraphics[width=0.9\textwidth]{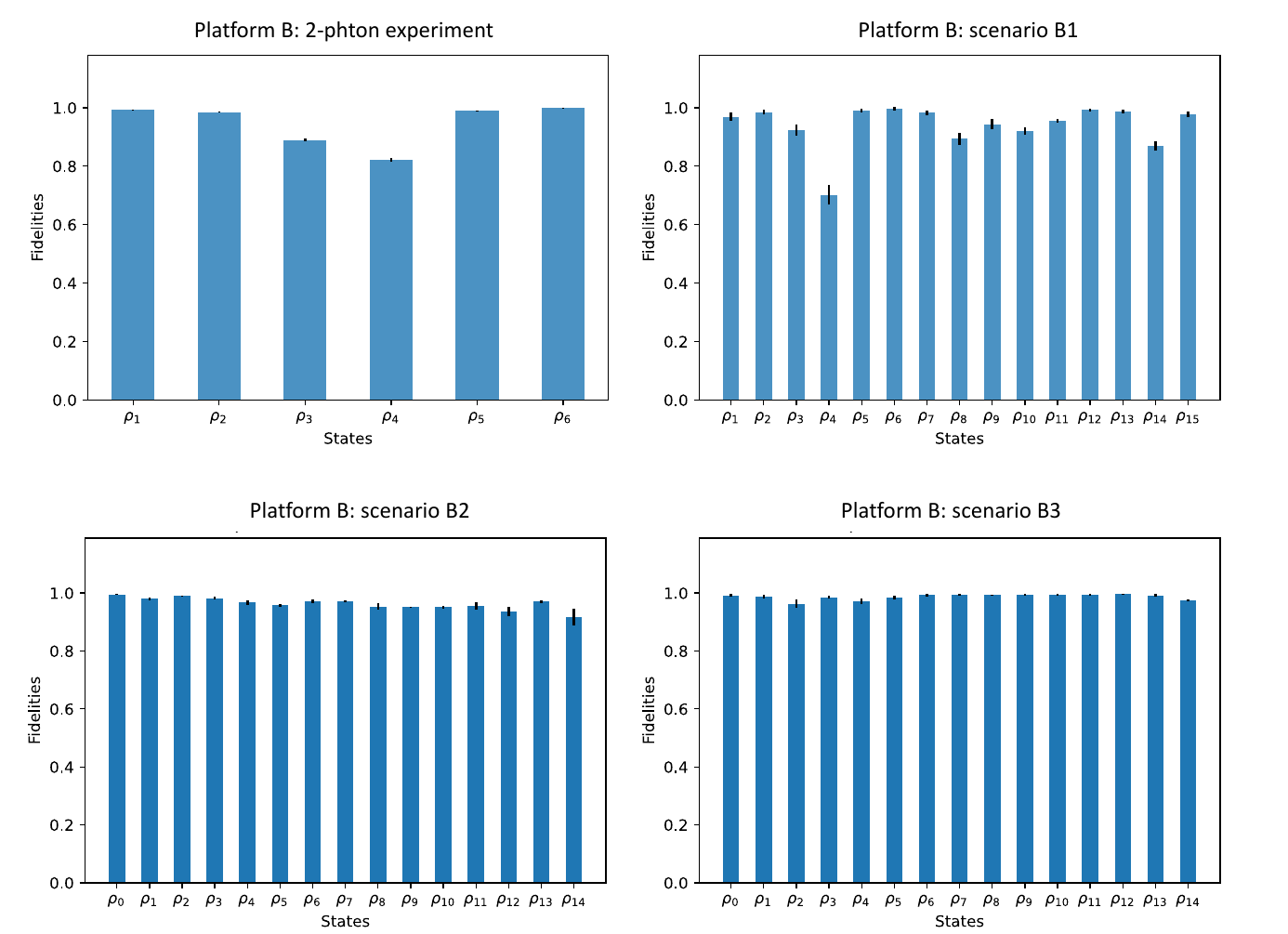}
    \caption{\textbf{Quantum states fidelities - platform B.} Values of the fidelities of the experimental states with the expected ones according to the model of an imperfect single-photon source of platform B. Error bars derive from uncertainties in single-photon counts due to Poisson statistics. The errors are smaller in B2 and B3 compared to B1 since fidelities were estimated on a larger data sample due to the higher detection efficiency of the SNSPDs.}
    \label{fig:fid}
\end{figure}
\begin{comment}
Then, we put the definition of the $\rho$ matrix in the $\mathcal{F}(\rho_i, \rho_j)$ expression and we obtain
%
\begin{equation}
    \mathcal{F}(\rho_i, \rho_j) = \frac{1+\vec{v}_i \cdot \vec{v}_j + \sqrt{(1-\abs{\vec{v}_1}^2)(1-\abs{\vec{v}_2}^2)}}{2}
\end{equation}
%

Theoretically, we perform a simulation according to the scheme of Fig.\ \ref{fig:exp} and we obtain the Block vectors
%
\begin{align}
    \vec{v}_1 = (-0.5954,  0.2302, -0.7698) \\
    \vec{v}_2 = (-0.1474, -0.9203,  0.3625) \\
    \vec{v}_3 = (-0.3269,  0.0465, -0.9439),
\end{align}
%
and a kernel
%
\begin{equation}
    K = \begin{pmatrix}
            1.000 & 0.298 & 0.966 \\
            0.298 & 1.000 & 0.331 \\
            0.966 & 0.331 &  1.000 
        \end{pmatrix}.
\end{equation}
%
Experimentally we obtain the Block vectors
\begin{align}
    \vec{v}_1 &= (-0.523 \pm 0.004,  0.242 \pm 0.006, -0.769 \pm 0.006) \\
    \vec{v}_2 &= (-0.14 \pm 0.04,  -0.93 \pm 0.02 0.34 \pm 0.03) \\
    \vec{v}_3 &= (0.252 \pm 0.005, 0.078\pm 0.005, -0.964\pm 0.006),
\end{align}
%
for an experimental kernel
\begin{equation}
\label{Eq.13}
    K_e = \begin{pmatrix}
            1.0\pm0 & 0.31 \pm 0.05 & 0.82 \pm 0.04 \\
            0.31 \pm 0.05 & 1.0 \pm 0 & 0.28 \pm 0.02 \\
            0.82 \pm 0.04 & 0.28 \pm 0.02 & 1.0 \pm 0
        \end{pmatrix}.
\end{equation}
%
The fidelity between the theoretical states and the experimental ones are 

\begin{align}
    F_1 = 0.98 \pm 0.03 \\
    F_2 = 0.99 \pm 0.01 \\
    F_3 = 0.92 \pm 0.03.
\end{align}
\end{comment}

%---------------------------------------------

\section{Other kernels collected in the experiment}
In the main text we have shown in B2 and B3 how to engineer kernels for classification tasks. A limitation of the present ABS experiments is that we did not actively apply adaptive operations. The ABS feature map has been realized and tested in post-selection. Obviously, such an implementation is not scalable and at higher-dimension, the use of active and fast optical operations will be unavoidable. However, the experiment carried out in post-selection allows us to collect for a given adaptive $U_i$ the photon counts in which the $r$ photons are detected in the output configuration $i$ and also other events where the photons exit from other $r$-output arrangements. Such ''undesired'' configurations generate further kernels since trigger the generation of different qubit and qutrit states to which the adaptive operation $U_i$ is applied. This implies that in principle it is possible to extract up to $D!$ kernels from photon samples collected in an ABS in post-selection.
In Fig. \ref{fig:post_kernel} we report some examples of these kernels for platform B2 a) and B3 b) that come out from analyzing all the data collected for each $U_i$. We use such sets of kernels to estimate the histograms of the classification section in the main text. The average accuracy in the classification of the 1D dataset has been calculated for 50 of these kernels.

\begin{figure}[h]
    \centering
    \includegraphics[width=\textwidth]{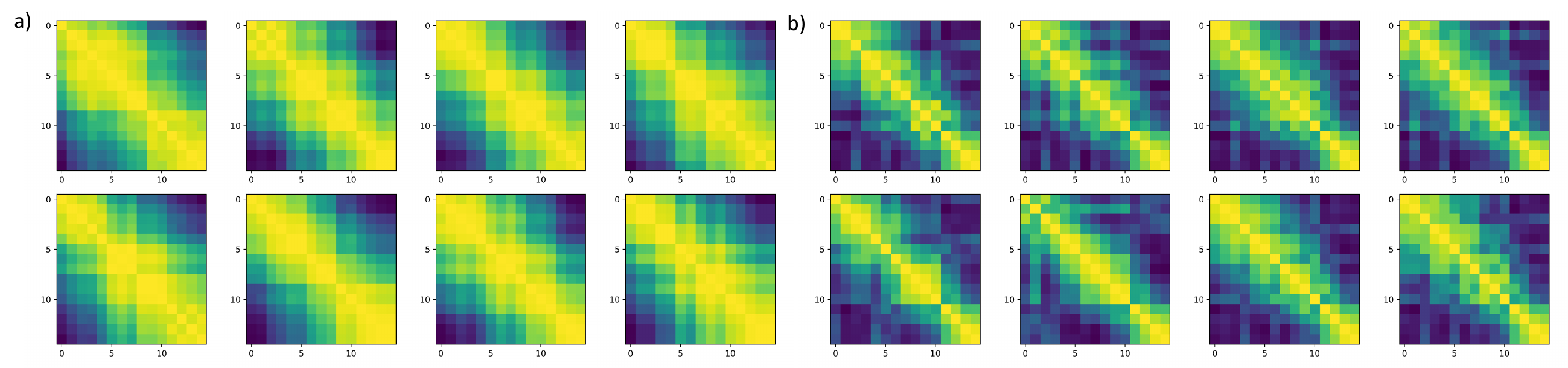}
    \caption{\textbf{Post-selected kernels.} a) Examples of kernels collected in platform B2 for qubits with different assignments of $U_i$ to the post-selected modes. b) The qutrit case in platform B3.}
    \label{fig:post_kernel}
\end{figure}

\section{Support Vector Machine Optimization}

Support Vector Machines (SVM) are machine learning algorithms for tackling binary classification problems. % by leveraging the principles of quantum computing.  
Given a set of training binary data $\mathcal{T}=\{ (x_0,y_0),\cdots,(x_l,y_l)\}$, with $y_i\in\{-1,1\}$ being the label of each data point $x_l$, the SVM is a decision function given by:
\begin{equation}
    f(x)=\mathrm{sign}\Big(\sum_{l=1}^{|\mathcal{T}|}\alpha_l y_l K(x_l,x)+b\Big)
\end{equation}
where $K(x_l, x)$ is the kernel.
The input of this decision function is an unlabeled data point $x$ and its output is the designated label $y$ of $x$. 

The optimization of the SVM parameters ${\alpha_l}$ for binary decisions giving the training data $\mathcal{T}$ can be expressed as the optimization of primal and dual cost functions. The primal optimization problem is expressed as:
\begin{align}
    &\text{minimize}_{\Vec{w},b}\quad\frac{1}{2}\|\bm{w}\|^2\\
    &\quad+\lambda \sum_{i=1}^{N} \max\left(0, 1 - y_i\left(\sum_{j=1}^{N} \alpha_j y_j K(x_j, x_i) + b\right)\right).\nonumber
\end{align}
The dual problem can be expressed as a quadratic programming problem given explicitly as:
\begin{align}
    \text{maximize}\quad & \sum_{l}^{|\mathcal{T}|}\alpha_l-\frac{1}{2}\sum_{l,l'}\alpha_ly_l\alpha_{l'}y_{l'}K(x_l,x_{l'}) \nonumber\\
    \text{subject to}\quad & \sum_{l}^{|\mathcal{T}|}\alpha_ly_l=0\\
    & 0\leq \alpha_l \leq \lambda, \quad \forall l \in \{\ 1,\dots,|\mathcal{T}|\},\nonumber
\end{align}
where the last constraint is a consequence of considering a soft-margin condition of the classifier with $\lambda>0$. For the purpose of this paper, we consider a soft-margin hyper parameter $\lambda=50$.% for the %CSVM and QSVM 
%optimizations to not over-fit the training.

{\section{Model for platform B}}

{Here we briefly discuss the theoretical model employed to obtain the predicted states from the experimental data acquired via platform B. In particular, we describe below how the different imperfections are taken into account, so as to obtain a model which allows to benchmark the output from the experiment. Different imperfections have to be taken into account.}

{\textit{(i) Partial photon indistinguishability.} The quantum-dot source emits a train of single photons emitted at subsequent time bins with a fixed interval, depending on the repetition rate of the laser source. Due mainly to spectral wandering of the photon emitted by the source, the output photons emitted a different time bins are not completely indistinguishable in all degrees of freedom. This amounts to provide a certain degree of reduction in the quantum interference process. Such an effect can be taken into account by describing the n-photon input states via the Gram matrix $S_{ij} = \langle \psi_{i} \vert \psi_{j} \rangle$, whose elements are the (complex) overlaps between the photons wavefunction $\vert \psi_{k} \rangle$ in all the internal degrees of freedom. The moduli of the Gram matrix elements can be reconstructed by measuring the pairwise Hong-Ou-Mandel visibilities between the photons. In the three photon case, there is also a non-trivial phase which can be reconstructed following the methodology reported in \cite{Pont_2022} In our case, this complex phase is found to be zero \cite{rodari2024semidevice}, thus leading to a positive real-value Gram matrix. The output probabilities can be then calculated following the theory of \cite{GBSVal2, Renema_2018_classical} including partial photon distinguishability.}

{\textit{(ii) Multiphoton emission from the source.} Another source of imperfections regards the possibility that the source emits two photons in the same time probability. This occur with small (but non-zero probability), which can be characterized via measurement of the second order correlation function $g^{(2)}(0)$. As discussed in \cite{Olivier_2021}, the noise photon is distinguishable from the principal ones, and thus this should be properly taken into account from partial distinguishability theory.}

{\textit{(iii) Losses.} Losses in the experimental can be modeled by considering that most of the apparatus is characterized by almost balanced losses between the modes, the only exception being difference in the detection efficiencies. The balanced losses commute with linear optical apparatuses, and thus can be all be grouped as a loss parameter $\eta$ per photon right after the source. Conversely, unbalanced losses are placed directly at the detection stage.}

{Overall, we included all three sources of noise in the model, enabling the possibility of calculating the output probabilities in the different configurations. Then, the output probabilities obtained via this method are processed, as performed with the experimental distributions, to obtain the predicitions for the output states.}

\vspace{1em}
{\section{Scaling up the approach: binning the adaptive measured output modes}}

{Here, we consider strategies to make ABS scalable in the number of adaptive measurements.  In particular, we will focus on avoiding exponentially decreasing postselection probability and simultaneously ensuring that the overall system is not classically simulable. To address this, we propose an approach in the ABS scheme based on binning the outcomes into polynomially-many bins, so that probabilities of binned adaptive measurement outcomes are inverse-polynomially large.
While it is known that binning could reduce the complexity of Boson Sampling in certain regimes, as demonstrated in \cite{Seron2024}, the ABS scheme operates in a regime where that simulation strategy is not efficient. We provide two main reasons hereafter, when the number of modes/photons is scaled up and when the number of adaptive measurements is scaled up, thus making binning a viable strategy for scaling the approach.}

{Firstly, the classical simulation algorithm taking advantage of binning in \cite{Seron2024} is in fact only efficient when the number of bins is a constant, and generally the complexity of such a simulation algorithm scales exponentially with the number of bins.
As shown in section 2.3 of \cite{Seron2024} the computational complexity to simulate a binned Boson Sampling distribution is $O(n^{2K+2} \log(n) \beta^{-2})$ where $n$ is the number of photons, $K$ is the total number of bins in which the modes are divided and $\beta$ is the error tolerance on the estimation of the event probabilities.
Using such an algorithm in the ABS regime, which requires polynomially many bins, would lead to a super-exponential runtime for the binned probability estimations.}% Of course, one cannot rule out that new and more efficient classical algorithms may be found, but given the current state-of-the-art, the regime in which the ABS scheme operates evades efficient classical simulability, when scaled up to larger instances.}

{Secondly, any classical simulation strategy taking advantage of binning would fail to be efficient as more adaptive operations are included in the ABS scheme. Indeed, estimating the probability of the absence of photons in a single output mode becomes a \textsf{BQP}-complete problem. Note that this corresponds to a binning of the outcome space into two bins. As such, even if the classical algorithm from \cite{Seron2024} could be extended from Boson Sampling to ABS, its worst-case runtime would have to scale super-polynomially with the number of adaptive measurements, unless \textsf{BPP}=\textsf{BQP}. A similar argument is discussed in the Methods section of the manuscript.}

\vspace{1em}
{\section{Adaptive Boson Sampling scheme: beyond kernel estimation}
Beyond data classification, the ABS scheme also allows to implementat variational machine learning algorithms, following the theoretical proposal from~\cite{Chabaud2021quantummachine}. These algorithms have wide-ranging applications for machine learning and can for instance train quantum models, approximate complex families of distributions and optimize cost functions. For instance, implementing such an optimization using the ABS scheme could work as follows: programmable photonic chip parameters would serve as variational parameters; multiple runs of the same instance of ABS would allow for estimating outcome probabilities; these estimates would be fed to a classical optimizer computing the gradient of the parameters based on a classical cost function; the programmable variational parameters would then be updated according to these gradients; the whole operation would be repeated until convergence of the cost function.
In the main text, we investigated the use of the ABS scheme for quantum kernels estimation and provided an efficient and scalable approach for their experimental measurement (see Fig. \ref{fig:abs_scheme}a). Amongst the other possible applications of the Adaptive Boson Sampling scheme for Quantum Machine Learning purposes, with the goal of increasing the complexity of its applications, a fully quantum approach to data classification can be devised, as already described in \cite{Chabaud2021quantummachine}. As shown in Fig. \ref{fig:abs_scheme}b), probability estimation can be carried out in order to directly assign a label to classical data (encoded in the adaptive string $\vb{p}$) with a procedure entirely carried out on the quantum device. %As shown in \cite{Chabaud2021quantummachine}, for a binary classification model this encompasses the definition of a function $g: \Phi_{m-k,n-r} \rightarrow \{-1,1\}$ associating an outcome on the bottom, non-adaptive layer of the interferometer, to a binary label $y$. In such a way, an output observable $\mathcal{N} = \sum_{t \in \Phi_{m-k,n-r}} g(t) \dyad{\vb{t}}$ can be defined. 
Here, a parameterized linear interferometer $U(\vb{\theta})$ is also supplied, acting on the output state $\ket{\psi_{\vb{p}}}$ before the photon-counting measurement. Overall, the probability of obtaining an outcome $y$ given that the adaptive outcome $\vb{p}$ has also been registered will be given by \cite{Chabaud2021quantummachine}:
\begin{equation}
    p(y|\vb{p}) = \frac{1}{2}(1 + y\expval{U^\dag(\vb{\theta}) \mathcal{N} U(\vb{\theta})}{\psi_{\vb{p}}})
\end{equation}
where $\mathcal{N}$ represent a suitably defined output measurement operator \cite{Chabaud2021quantummachine} with outcomes the corresponding binary labels $y$ with a certain probability. The latter can be experimentally estimated with the frequency $p(y|\vb{p}_i) = T_{y|i}/T_{i}$ as the ratio between the times in which outcome $y$ has been obtained conditioned on the observation of $\vb{p}$ and total the times in which $\vb{p}$ has been observed. Here, the parameters of the linear circuit $U(\vb{\theta})$ applied to the ABS output can be variationally trained in order to minimize a suitable cost function \cite{Chabaud2021quantummachine}, e.g. representing the overall classification error. In such a way, a fully on-device variational quantum classifier can be obtained, within the ABS framework.}

\begin{figure}[tb]
    \centering
    \includegraphics[width=1\linewidth]{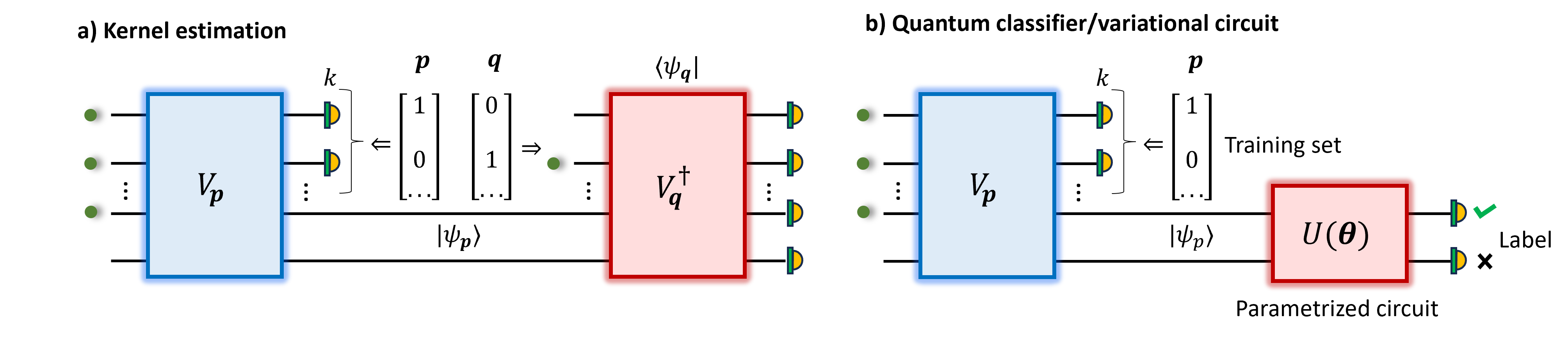}
    \caption{{\textbf{Boson Sampling with adaptive measurements.} a) ABS schemes for kernel estimation. Data $\bm{p}$ and $\bm{q}$ with a different number of photons $r$ will have a null kernel since have zero overlap. The kernel element between states with the same $r$ is estimated by applying the circuit $V_{\bm q}^\dagger$ ($V_{\bm p}^\dagger$) to $\bm{q}$ ($\bm{p}$) every time the data $\bm{p}$ ($\bm{q}$) is measured. The kernel element will be the number of times the input state of the protocol is detected at the output. b) ABS for data classification and variational optimization. A parametrized circuit $U({\bm{\theta}})$ is applied to the quantum states $\ket{\psi_{\bm p}}$ to be trained to label classical data string $\bm{p}$ or to compute gradient for minimization purposes in a variational fashion.}}
    \label{fig:abs_scheme}
\end{figure}

%\bibliography{main.bib}
%\bibliographystyle{apsrev4-1}

\providecommand{\noopsort}[1]{}\providecommand{\singleletter}[1]{#1}%